\documentclass[11pt]{article}

\usepackage{jheppub}
\pdfoutput=1

\usepackage{amsmath, amssymb}
\usepackage{amsthm}
\usepackage{jheppub}
\usepackage{mathrsfs}
\usepackage{psfrag}
\usepackage{color}
\usepackage{hyperref}
\usepackage{enumitem}
\usepackage{graphicx,import}
\usepackage{slashed}

\usepackage{simplewick}
\usepackage{cancel}

\usepackage{todonotes}
\usepackage{tikz, pict2e}

\usepackage{lmodern}
\usepackage{graphicx}
\usepackage{float}
\usepackage{wrapfig}
\usepackage{subcaption}
\usepackage{tkz-euclide}
\usepackage[utf8]{inputenc}

\usetikzlibrary{calc,backgrounds, intersections}
\usetikzlibrary{arrows,shapes,positioning,shadows,trees}
\usetikzlibrary{mindmap}
\usetikzlibrary{decorations.pathreplacing}
\usetikzlibrary{decorations.pathmorphing}
\usetikzlibrary{decorations.markings}
\usetikzlibrary{matrix}

\tikzstyle arrowstyle=[scale=1]
\tikzstyle directed=[postaction={decorate,decoration={markings,
    mark=at position .5 with {\arrow[arrowstyle]{stealth}}}}]
\tikzstyle reverse directed=[postaction={decorate,decoration={markings,
    mark=at position .5 with {\arrowreversed[arrowstyle]{stealth};}}}]
    
\allowdisplaybreaks

\usepackage{mathtools,bbm,array,amsfonts,graphicx,wrapfig,arydshln,lscape,float,multirow,longtable,rotating,makecell}
\usepackage{url}

\definecolor{redb}{rgb}{0.700, 0.000, 0.300}
\DeclareMathAlphabet\mathbfcal{OMS}{cmsy}{b}{n}

\newtheorem{theorem}{Theorem}[section]
\newtheorem{prop}{Proposition}[section]

\newtheorem{lemma}[theorem]{Lemma}

\DeclareMathOperator*{\Res}{\mbox{Res}}
\DeclareMathOperator*{\ResT}{\mbox{\emph{Res}}}
\DeclareMathOperator*{\Lor}{\mathcal{L}}

\newcommand{\xdownarrow}[1]{%
  {\left\downarrow\vbox to #1{}\right.\kern-\nulldelimiterspace}
}

\title{Positive Geometries and Differential Forms with Non-Logarithmic Singularities I}

\author[1]{Paolo Benincasa,}\emailAdd{pablowellinhouse@anche.no}
\author[2]{and Matteo Parisi}\emailAdd{parisi@maths.ox.ac.uk}

\affiliation[1]{Niels Bohr International Academy, Niels Bohr Institute,\\ University of Copenhagen, \\ Blegdamsvej 17, 2100 København, Denmark }
\affiliation[2]{Mathematical Institute, University of Oxford,\\ Andrew Wiles Building, Radcliffe Observatory Quarter,\\ Woodstock Road, Oxford, OX2 6GG, U.K.}

\abstract{Positive geometries encode the physics of scattering amplitudes in flat space-time and the wavefunction of the universe in cosmology for a large class of models. Their unique canonical forms, providing such quantum mechanical observables, are characterised by having only logarithmic singularities along all the boundaries of the positive geometry. However, physical observables have logarithmic singularities just for a subset of theories. Thus, it becomes crucial to understand whether a similar paradigm can underlie their structure in more general cases. In this paper we start a systematic investigation of a geometric-combinatorial characterisation of differential forms with non-logarithmic singularities, focusing on projective polytopes and related meromorphic forms with multiple poles. We introduce the notions of {\it covariant forms} and {\it covariant pairings}. Covariant forms have poles only along the boundaries of the given polytope; moreover, their leading Laurent coefficients along any of the boundaries are still covariant forms on the specific boundary. Whereas meromorphic forms in covariant pairing with a polytope are associated to a specific (signed) triangulation, in which poles on spurious boundaries do not cancel completely, but their order is lowered. These meromorphic forms can be fully characterised if the polytope they are associated to is viewed as the restriction of a higher dimensional one onto a hyperplane. The canonical form of the latter can be mapped into a covariant form or a form in covariant pairing via a {\it covariant restriction}. We show how the geometry of the higher dimensional polytope determines the structure of these differential forms. Finally, we discuss how these notions are related to Jeffrey-Kirwan residues and cosmological polytopes.}

\begin{document}

\maketitle


\section{Introduction}\label{sec:Intro}

The study of positive geometries has been increasingly acquiring relevance in physics as they turned out to be the underlying mathematical structure for quantum mechanical observables in a quite large class of theories in particle physics and cosmology. 

The original interpretation of scattering amplitudes in gauge theories as volumes of certain polytopes \cite{Hodges:2009hk, ArkaniHamed:2010gg} suggested that geometrical and combinatorial ideas could play a more fundamental role in understanding the structure of scattering amplitudes and the physics they encode. This became clearer when first the geometry and combinatorics of the \emph{positive Grassmannian} \cite{Postnikov:2006kva} was introduced to describe the integrand of the perturbative scattering amplitudes in planar $\mathcal{N}\,=\,4$ supersymmetric Yang-Mills theory at all loop order \cite{ArkaniHamed:2012nw}, and then the very same amplitudes turned out to be encoded in the canonical differential form associated to the  \emph{amplituhedron} \cite{Arkani-Hamed:2013jha}, a geometrical structure which generalises both (certain types of) polytopes and the positive Grassmannian.

A further indication was provided by the fact that positive geometries did not get confined to the realm of the maximally supersymmetric gauge theory, but they also emerged in the context of scalar scattering in the form of the ABHY associahedron \cite{Arkani-Hamed:2017mur, Frost:2018djd} for the bi-adjoint cubic interactions, and Stokes polytopes \cite{Banerjee:2018tun, Salvatori:2019phs} for planar quartic interactions. Even more surprisingly, they appeared in cosmology, where the canonical form of the so-called cosmological polytopes encodes the wavefunction of the universe \cite{Arkani-Hamed:2017fdk}, which is the relevant quantum mechanical observable, for a large class of toy models. Finally, it was recently introduced an extension of canonical forms for general polytopes, named stringy canonical form, which depends on a certain deformation parameter (which resembles the $\alpha'$ parameter in string theory) and, when applied to the ABHY associahedron return the Koba-Nielsen integral known in string theory \cite{Arkani-Hamed:2019mrd}.

From the physics perspective, the excitement about such geometric-combinatorial picture on scattering amplitude and the wavefunction of the unverse is indeed not due to having acquired new computational tools to play with. These quantum mechanical observables carry the imprint of the fundamental rules for the physics in flat and cosmological space-times respectively, and, in particular, how causal time evolution is encoded into them is far from being understood as well as it is not understood what fundamentally fixes their properties. Positive geometries offer a new perspective on these basic questions: they are mathematical structures with their own first principle intrinsic definition and no a priori reference to any physics notion, and the principles and properties we ascribe to scattering amplitudes and the wavefunction of the universe can be seen as emergent from these mathematical principles.

An example of such emergence phenomenon was observed in the context of the cosmological polytopes. The wavefunction of the universe is a non Lorentz invariant quantity defined on a space-like surface, and {\it contains} the flat space scattering amplitudes \cite{Maldacena:2011nz}: the vertex structure of a specific facet of the cosmological polytopes provides a geometrical-combinatorial origin for the cutting rules determining the unitarity of the scattering amplitudes, while the structure of its dual does it for Lorentz invariance \cite{Benincasa:2018ssx}.

The common denominator among all the positive geometries is the fact that they can be characterised by a canonical form, which has logarithmic singularities on all its boundaries, and it is precisely such a canonical form which returns the quantum mechanical observables in flat space-times and in cosmology. However, meromorphic functions with logarithmic singularities represent a special corner: in general both scattering amplitudes and the wavefunction of the universe have a much more complicated structure. There is a plethora of examples of theories whose (integrand of the) scattering amplitudes or wavefunction of the universe possess non-logarithmic singularities: from less supersymmetric gauge theories \cite{Benincasa:2015zna, Benincasa:2016awv}, to gravity \cite{Herrmann:2016qea, Heslop:2016plj} and pretty much any theory in cosmology. Hence, in order for the geometrical-combinatorial principles behind the positive geometries to have any chance to play any fundamental role in the understanding of physical processes in both flat space-time and cosmology, it is necesary either to make a connection between positive geometries and functions with non-logarithmic singularities, or to find new ideas that generalise the positive geometries to include function with non-logarithmic singularities.

While a systematic characterisation of positive geometries and canonical form has already started independently of any physical interpretation \cite{Arkani-Hamed:2017tmz}, to our knowledge a link between positive geometries and functions with non-logarithmic singularities has not been explored. There is a beautiful exception in the context of the cosmological polytopes \cite{Benincasa:2019vqr}. The definition of cosmological polytopes as generated from a space of triangles embedded in projective space by intersecting them in the midpoints of their sides and taking the convex hull of their vertices, can be generalised by including a collection of segments in the building blocks allowing to get intersected in their only midpoint; a specific limit of the canonical form of the polytopes constructed in this way  returns a differential form with higher order poles whose coefficient represents the correct wavefunction of the universe for certain scalar states in cosmology. As a striking feature, all the information encoded into such a differential form with higher order poles could be extracted from the canonical form of the (generalised) cosmological polytope: as the flat-space limit is encoded in the leading Laurent coefficient of the diffrential form, it can be extracted from the polytope as the canonical form of a higher codimension face, whose codimension provides the multiplicity of the relevant pole in the differential form \cite{Benincasa:2019vqr}.

These results constitutes a first example of association of a differential form with non-logarithmic singularities with a projective polytope, and brings the question of whether a similar construction might exist also for scattering amplitudes in less supersymmetric gauge theories and in gravity, which admit a description in terms of the Grassmannian.

In this paper we start a systematic exploration of a geometrical-combinatorial characterisation of differential forms with non-logarithmic singularities, focusing on meromorphic forms with multiple poles on one side and projective polytopes on the other. In Section \ref{sec:PG_CF} we review the general concepts of positive geometries and canonical forms, and projective polytopes in particular, as well as the {\it Jeffrey-Kirwan residue} \cite{JEFFREY1995291} and the cosmological polytopes which will be used as important examples in the rest of the paper. In particular, the Jeffrey-Kirwan residue can be used to compute the canonical form of any projective polytope and provides a way to capture all (regular) triangulations at once \cite{Ferro:2018vpf}. In Section \ref{sec:DFHP} we define the association of classes of meromorphic differential forms to projective polytope, introducing the notions of {\it covariant forms}, as meromorphic forms with a certain $GL(1)$-scaling and multiple poles along the boundaries of the associated polytope such that its leading Laurent coefficient is still a meromorphic form with the same properties, and {\it covariant pairings} as a pairing between a given polytope and a meromorphic form with multiple poles with a certain $GL(1)$ scaling and poles along the boundaries of the elements of a certain signed triangulation of the paired polytope such that the multiplicity of the poles related to a subset of faces which sign-triangulate the empty set is lowered but still non-zero. The geometry and combinatorics of the projective polytope {\it partially} characterise and determine these differential forms, as for each projective polytope these associations are not unique. We complete these characterisation by constructing a projective polytope as a restriction of a higher dimensional one onto a hyperplane, and associating it the meromorphic differential form with multiple pole via the {\it covariant restriction} of the canonical form of the higher dimensional projective polytope. {\it i.e.} as the leading Laurent coefficient along the hyperplane the higher dimensional projective polytope is restricted onto. We provide a geometrical interpretation of the multiplicity of the poles of the meromorphic form generated in this way and relates its leading Laurent coefficients to the faces of the higher dimensional polytope. We also provide a number of explicit examples. Sections \ref{subsec:jkex} and \ref{subsec:CPCF} are respectively devoted to the discussion of the relation between the notions we introduced and the Jeffrey Kirwan method, and their realisation in the context of the cosmological polytopes. Finally, Section \ref{sec:Concl} contains our conclusions and future directions.


\section{Positive Geometries and Canonical Forms}\label{sec:PG_CF}

In this section we briefly review the definition as well as the salient features of positive geometries and the associated canonical forms. This allows us to set the notation and make our discussion self-contained. For a detailed treatment of the subject, see \cite{Arkani-Hamed:2017tmz} and references therein. We will explicitly discuss the projective polytopes, and a special subclass of them, the so called-cosmological polytopes \cite{Arkani-Hamed:2017fdk, Benincasa:2019vqr}.


\subsection{Generalities on Positive Geometries and Canonical Forms}\label{subsec:PG_CF_Gen}

Let us consider a pair $\left(X,\,X_{\ge0}\right)$, where:
\begin{enumerate}[label=(\roman*)]
    \item $X$ is an {\it (irreducible) complex projective variety}  of complex dimension $D$, {\it i.e.} the set of solutions of homogeneous polynomial equations in the complex projective space $\mathbb{P}^N(\mathbb{C})$                                $\left\{x\,\in\,\mathbb{P}^N(\mathbb{C})\,\Big|\,p(x)\,=\,0\right\}$, whose coefficients are taken to be real by assumption;
    \item $X_{\ge0}\,\subset\,X(\mathbb{R})$ is a (non-empty) closed {\it semi-algebraic} set of real dimension $D$, {\it i.e.} a finite union of subsets in $X(\mathbb{R})$, which is the set of solutions in $\mathbb{P}(\mathbb{R})$         of the very same homogeneous polynomial equations $\{x\,\in\,\mathbb{P}^N(\mathbb{R})\,\Big|\,p(x)\,=\,0\}$ defining $X$ cut out by homogeneous real polynomial inequalities                                    
          $\{x\,\in\,\mathbb{P}^N(\mathbb{R})\,|\,q(x)\,>\,0\}$. The interior $X_{>0}$ of $X_{\ge0}\,\subset\,X(\mathbb{R})$ is assumed to be a $D$-dimensional open oriented real submanifold of $X(\mathbb{R})$, and $\overline{X_{>0}}\,=\,X_{\ge0}$;
    \item its boundary components are given by the pairs $(C^{\mbox{\tiny $(j)$}},\,C^{\mbox{\tiny $(j)$}}_{\ge0})$ ($j\,=\,1,\ldots,\,\tilde{\nu}$), with $C^{\mbox{\tiny $(j)$}}$ ($j\,=\,1,\ldots,\,\tilde{\nu}$) being the 
          irreducible components of the set $\partial X$ of the homogeneous polynomial equations which are satisfied in $X$ if they are satisfied in any arbitrary point of $\partial X_{\ge0}\,:=\:X_{\ge0}\setminus X_{>0}$, and $C^{\mbox{\tiny $(i)$}}_{\ge0}$ being the closure of the interior of $C^{\mbox{\tiny $(i)$}}\,\bigcap\,\partial X_{\ge0}$ inside $C^{\mbox{\tiny $(i)$}}(\mathbb{R})$.
\end{enumerate}
Then a {\it positive geometry} is defined as such a pair with the following features:
\begin{enumerate}[label=(\alph*)]
    \item if $D\,>\,0$, then all the boundary components $(C^{\mbox{\tiny $(j)$}},\,C^{\mbox{\tiny $(j)$}}_{\ge0})$ ($j\,=\,1,\,\ldots,\,\tilde{\nu}$) of the positive geometry $\left(X,\,X_{\ge0}\right)$ is a codimension-one positive geometry;
    \item if $D\,=\,0$, there is a unique positive geometry $(X,\,X)$, with $X$ being a point and $X_{\ge0}\,=\,X$.
\end{enumerate} 

\noindent
Any positive geometry $(X,\,X_{\ge0})$ is in $1-1$ correspondence\footnote{In principle, the canonical forms are defined up to an overall constant $a\,\in\,\mathbb{R}$, so that the highest codimension singularity turns out to be $\pm a$ depending on the orientation. As we will see shortly afterwards, such a constant can be fixed {\it by convention} with a requirement on the {\it leading singularities} or, which is the same, on the canonical form for $D\,=\,0$. Once this freedom is fixed, the canonical form is defined univocally.} with a {\it canonical form} $\omega(X,\,X_{\ge0})$, {\it i.e.} a non-zero meromorphic $D$-form on $X$, such that its residue along any of the boundary components $(C^{\mbox{\tiny $(j)$}},\,C^{\mbox{\tiny $(j)$}}_{\ge0})$ is the canonical form of the positive geometry constituted by the boundary component $(C^{\mbox{\tiny $(j)$}},\,C^{\mbox{\tiny $(j)$}}_{\ge0})$ itself:
\begin{equation}\label{eq:CFdef}
 \Res_{\mbox{\tiny $C^{\mbox{\tiny $(j)$}}$}}\left\{\omega(X,\,X_{\ge0})\right\}\:=\:\omega(C^{\mbox{\tiny $(j)$}},\,C^{\mbox{\tiny $(j)$}}_{\ge0}).
\end{equation}
Let us parametrise $X$ with a set of local holomorphic coordinates $(y_j,\,h_j)$ such that the locus $h_j\,=\,0$ locally identifies $C^{\mbox{\tiny $(j)$}}$, while $y_j$ collectively indicates the remaining local coordinates. Then the canonical form $\omega(X,\,X_{\ge0})$ shows a simple pole in $h_j\,=\,0$ such that
\begin{equation}\label{eq:CFX}
 \omega(X,\,X_{\ge0})\:=\:\omega(y_j)\wedge\,\frac{dh_j}{h_j}+\tilde{\omega},
\end{equation}
with $\tilde{\omega}$ being the part of the canonical form which does not have a pole in $h_j\,=\,0$ and thus does not contribute to its residue. Hence, the residue of the canonical form with respect to such a simple pole is nothing but the codimension one differential form $\omega(y_j)$ which depends only on the collective local coordinates $c_j$ and it constitutes the canonical form of the (codimension-one) boundary component $(C^{\mbox{\tiny $(j)$}},\,C^{\mbox{\tiny $(j)$}}_{\ge0})$:
\begin{equation}\label{eq:CFC}
 \Res_{\mbox{\tiny $C^{\mbox{\tiny $(j)$}}$}}\left\{\omega(X,\,X_{\ge0})\right\} \:=\: \Res_{\mbox{\tiny $h_j\,=\,0$}} \left\{\omega(X,\,X_{\ge0})\right\}\:=\:
 \omega(y_j)\,=\, \omega(C^{\mbox{\tiny $(j)$}},\,C^{\mbox{\tiny $(j)$}}_{\ge0}),
\end{equation}
where the equalities are valid locally. Applying the Res operator \eqref{eq:CFdef} on $\omega(X,\,X_{\ge0})$ iteratively $D$ times along different boundary components one must obtain $\pm1$, depending on the orientation. Such highest codimension singularities are the {\it leading singularities}. 

For $D\,=\,0$, when $X$ is a single point and $X_{\ge0}\,=\,X$, the associated canonical form on $X$ is the $0$-form $\pm1$ depending on the orientation of $X_{\ge0}$. Notice that the leading singularities are associated to points, whose canonical form is precisely $\pm1$.

The canonical form $\omega(X,\,X_{\ge0})$ provides a characterisation of the positive geometry $(X,\,X_{\ge0})$, associating the boundary components $\{(C^{\mbox{\tiny $(j)$}},\,C^{\mbox{\tiny $(j)$}}_{\ge0})\}$ of $(X,\,X_{\ge0})$ to its singularities.


\subsection{Projective Polytopes}\label{subsec:PrPol}

We now specialize to a specific class of positive geometries, the {\it projective polytopes}. Given a set of vectors $Z_k\,\in\,\mathbb{R}^{N+1}$ ($k\,=\,1,\ldots,\,\nu$), then a projective polytope is defined as the pair $(\mathbb{P}^{N},\,\mathcal{P})$, where $\mathcal{P}\,\subset\,\mathbb{P}^{N}(\mathbb{R})$ is the convex hull identified by
\begin{equation}\label{eq:Polyt}
 \mathcal{P}(\mathcal{Y},\,Z)\: :=\:\left\{\mathcal{Y}\,=\,\sum_{k=1}^{\nu}c_k Z_k\,\in\,\mathbb{P}^{N}(\mathbb{R})\,\Big|\,c_k\,\ge\,0,\,\forall k\,=\,1,\ldots,\,\nu\right\},
\end{equation}
with the $Z_k$'s being its vertices, and $\mathcal{Y}$ which can vanish if and only if $c_k\,=\,0$ for all $k\,=\,1,\ldots,\,\nu$. Notice that any projective polytope is invariant under the transformation $\mathcal{Y}\,\longrightarrow\,\lambda\,\mathcal{Y}$ ($\lambda\,\in\,\mathbb{R}_{+}$) -- or, equivalently, $Z_k\,\longrightarrow\,\lambda Z_k,\,\forall\,k\,=\,1,\ldots,\,\nu$.
\\

Alternatively, $\mathcal{P}$ can be defined via a set of homogeneous polynomial inequalities $q_j(\mathcal{Y})\,\ge\,0$ ($\mathcal{Y}\,\in\,\mathbb{P}^N(\mathbb{R}),\;j\,=\,1,\ldots,\tilde{\nu}$), with every polynomial $q_j(\mathcal{Y})$ being {\it linear}, {\it i.e.} via $q_j(\mathcal{Y})\,\equiv\,\mathcal{Y}^I\mathcal{W}_I^{\mbox{\tiny $(j)$}}\,\ge\,0$, where the {\it dual vectors} $\mathcal{W}_I^{\mbox{\tiny $(j)$}}$ are co-vectors in $\mathbb{R}^{N+1}$ and correspond to the facets of the polytope. Given a certain facet identified by $\mathcal{W}_I^{\mbox{\tiny $(j)$}}$, a vertex $Z_k$ is on it if and only if it satisfies the relation $\mathcal{W}_I^{\mbox{\tiny $(j)$}}Z^I_k\,=\,0$. Let $Z_{\mbox{\tiny $a_{j+1}$}},\,\ldots,\,Z_{\mbox{\tiny $a_{j+N}$}}$ a subset of vertices of $\mathcal{P}$ on the facet $\mathcal{W}_I^{\mbox{\tiny $(j)$}}$ forming a basis in $\mathbb{R}^N$, then
\begin{equation}\label{eq:Wfac}
 \mathcal{W}_I^{\mbox{\tiny $(j)$}}\:=\:(-1)^{(j-1)(N-1)}\varepsilon_{\mbox{\tiny $IK_1\ldots K_N$}}Z_{\mbox{\tiny $a_{j+1}$}}^{\mbox{\tiny $K_1$}}\ldots Z_{\mbox{\tiny $a_{j+N}$}}^{\mbox{\tiny $K_N$}},
\end{equation}
$\varepsilon_{IK_1\ldots K_N}$ being the totally anti-symmetric $(N+1)$-dimensional Levi-Civita symbol.
\\

Given a projective polytope $(\mathbb{P}^N,\,\mathcal{P})$, with $\mathcal{P}$ defined via a set of homogeneous linear polynomial inequalities $q_j(\mathcal{Y})\,\ge\,0$ ($j\,=\,1,\ldots\,\tilde{\nu}$), its associated canonical form $\omega(\mathcal{Y},\,\mathcal{P})$ is given by a meromorphic form with singularities only where the homogeneous linear polynomials $q_j(\mathcal{Y})$, $j\,=\,1,\ldots\,\tilde{\nu}$, vanish, and whose numerator $\mathfrak{n}(\mathcal{Y})$ is a polynomial of degree $\tilde{\nu}-N-1$, such that the residue of $\omega(\mathcal{Y},\,\mathcal{P})$ at any of poles $q_j(\mathcal{Y})\,=\,0$ is the canonical form of a codimension-one boundary component:
\begin{equation}\label{eq:CFP}
 \omega(\mathcal{Y},\,\mathcal{P})\:=\:\frac{\mathfrak{n}(\mathcal{Y})\langle\mathcal{Y}d^N\mathcal{Y}\rangle}{\prod_{j=1}^{\tilde{\nu}}q_j(\mathcal{Y})},\hspace{1cm}
 \mbox{deg}\{\mathfrak{n}\}\:=\:\tilde{\nu}-N-1,
\end{equation}
where $\mbox{deg}\{\mathfrak{n}\}$ is the degree of the homogeneous polynomial $\mathfrak{n}(\mathcal{Y})$ and $\langle\mathcal{Y}d^N\mathcal{Y}\rangle$ is the \emph{standard measure} in $\mathbb{P}^N$, which is defined as
\begin{equation}\label{eq:TopF}
 \langle\mathcal{Y}d^N\mathcal{Y}\rangle\: := \: \varepsilon_{\mbox{\tiny $I_1I_2\ldots I_{N+1}$}}\mathcal{Y}^{I_1}d\mathcal{Y}^{I_2}\wedge\ldots\wedge\,d\mathcal{Y}^{I_{N+1}}.
\end{equation}
Importantly, the degree of the homogeneous polynomial $\mathfrak{n}(\mathcal{Y})$ makes the canonical form $\omega(\mathcal{Y},\,\mathcal{P})$ invariant under the $GL(1)$ transformation $\mathcal{Y}\,\longrightarrow\,\lambda\,\mathcal{Y}$, $\lambda\,\in\,\mathbb{R}_+$. Geometrically, it is fixed by the locus of the intersections of the faces of $\mathcal{P}$ outside $\mathcal{P}$ \cite{Arkani-Hamed:2014dca}.

The canonical form \eqref{eq:CFP} can be explicitly written in terms of the dual vectors $\mathcal{W}_I^{\mbox{\tiny $(j)$}}$ as well as in terms of the vertices $Z^I_k$ via \eqref{eq:Wfac}:
\begin{equation}\label{eq:CFPWZ}
 \omega(\mathcal{Y},\,\mathcal{P})\:=\:\frac{\mathfrak{n}(\mathcal{Y})\langle\mathcal{Y}d^N\mathcal{Y}\rangle}{\prod_{j=1}^{\tilde{\nu}}\left(\mathcal{Y}\cdot\mathcal{W}^{\mbox{\tiny $(j)$}}\right)}\:=\:
                                       \frac{\mathfrak{n}(\mathcal{Y})\langle\mathcal{Y}d^N\mathcal{Y}\rangle}{\prod_{j=1}^{\tilde{\nu}}\langle\mathcal{Y}Z_{a_{j+1}}\ldots Z_{a_{j+N}}\rangle}
\end{equation}
where $\mathcal{Y}\cdot\mathcal{W}^{\mbox{\tiny $(j)$}}\, :=\,\mathcal{Y}^I\mathcal{W}_I^{\mbox{\tiny $(j)$}}$ and $\langle\ldots\rangle$ identifies the contraction via the Levi-Civita symbol, {\it i.e.} the determinant of the matrix built out the vectors appearing inside the angular brackets.
\\

Given a projective polytope $(\mathbb{P}^N,\,\mathcal{P})$, with $\mathcal{P}$ defined via a set of vertices $Z_k$ ($k\,=\,1,\ldots\,\nu$), its associated canonical form can be expressed in terms of the so-called {\it canonical function} and the standard measure in $\mathbb{P}^N$, with the canonical function which has as a contour integral representation \cite{Arkani-Hamed:2017tmz}
\begin{equation}\label{eq:CFPci}
 \begin{split}
 &\hspace{2.5cm}\omega(\mathcal{Y},\,\mathcal{P})\:=\:\Omega(\mathcal{Y},\,\mathcal{P})\langle\mathcal{Y}d^{N}\mathcal{Y}\rangle,\\
 &\phantom{\ldots}\\
 &\Omega(\mathcal{Y},\,\mathcal{P})\:=\:\frac{1}{N!(2\pi i)^{\nu-N-1}}\int_{\mathbb{R}^{\nu}}\prod_{k=1}^{\nu}\frac{dc_k}{c_k-i\varepsilon_k}
   \delta^{\mbox{\tiny $(N+1)$}}\left(\mathcal{Y}-\sum_{k=1}^{\nu}c_k Z_k\right)
 \end{split}
\end{equation}
in the limit for $\varepsilon_k\,\longrightarrow\,0$, $\forall\,k\,=\,1,\ldots,\nu$. There are several contours along which the above integral can be performed and all of them provide different triangulations of the polytope.
\\

Integration contours capturing all the (regular)\footnote{Regular triangulations are a special class of triangulations which can be obtained in the following way. Consider a real-valued function $Z_i \mapsto \alpha(Z_i)$ on the vertices of $\mathcal{P}$. 
Then consider the points $(Z_i, \alpha(Z_i))$ and take their convex hull. Take the lower faces (those whose outwards normal vector have last component negative) and project them back down to $\mathcal{P}$. This gives a subdivision of $\mathcal{P}$, which is called \emph{regular}. In case its elements are all simplices it is a regular triangulation. See \cite{de2010triangulations} for a extensive review on the topic.} triangulations of a given polytope can be defined via a method \cite{Ferro:2018vpf} which relies on the \emph{Jeffrey-Kirwan Residue} \cite{JEFFREY1995291,2004InMat.158..453S}. We will give a brief review and refer to \cite{Ferro:2018vpf} and the upcoming work \cite{moh:2020} for further details. In Section \eqref{subsec:jkex} we will present an explicit example and explain its connection with this work. 
The computation of the canonical form of a polytope $\mathcal{P}$ can be recasted as a residue computation of a (covariant) $\nu-N-1$ differential form defined on $\mathbb{P}^{\nu-N-1}$ as:
\begin{equation} \label{def:formfiber}
    \tilde{\omega}_{\mathcal{Y}}(C,\mathcal{P}):=\bigwedge_{k=1}^{\nu} \frac{dc_k}{c_k} \,
   \delta^{\mbox{\tiny $(N+1)$}}\left(\mathcal{Y}-\sum_{k=1}^{\nu}c_k Z_k\right).
\end{equation}

Let $\tilde{\omega}$ be a top differential form in $\mathbb{P}^r$ which has poles on each of the hyperplanes $\lbrace \mathcal{H}_k \rbrace_{k=1}^{\nu}$ and let us denote their dual vectors as $\lbrace \mathfrak{B}_k \rbrace_{k=1}^{\nu}$, with $\nu \geq r$. For each collection $\lbrace \mathcal{H}_k \rbrace_{k \in I}$ of $r$ of such hyperplanes, with\footnote{Given $n \in \mathbb{N}$, throughout the text, we will denote as $[n]$ the set $\lbrace 1,\ldots, n \rbrace$. Moreover, ${[n]\choose r}$ will be the set of $r$-element subsets of $[n]$. } $I \in {[\nu] \choose r}$, let us define the \emph{cone} $\mathfrak{C}_I$ as the subset in $\mathbb{P}^r$ spanned by positive linear combinations of the corresponding dual vectors $\lbrace \mathfrak{B}_k \rbrace_{k \in I}$. Let us now fix a reference point $\xi \in \mathbb{P}^r$, then we define the Jeffrey-Kirwan residue as:
\begin{equation}\label{def:jk}
    \mathrm{JK}_{\xi} \tilde{\omega} := \sum_{\mathfrak{C}_I \ni \xi} \mbox{Res}_{\mathfrak{C}_I} \tilde{\omega}, 
\end{equation}
where the sum is over all cones $\mathfrak{C}_I$ containing the point $\xi$. Moreover, $\mbox{Res}_{\mathfrak{C}_I}$ is the multivariate residue of $\tilde{\omega}$ computed around the poles corresponding to the hyperplanes $\lbrace \mathcal{H}_k \rbrace_{k \in I}$ in the order $(\mathcal{H}_{k_1}, \ldots,\mathcal{H}_{k_r})$ such that the corresponding dual vectors $(\mathfrak{B}_{k_1},\ldots,\mathfrak{B}_{k_r})$ are positively oriented\footnote{I.e. in an affine chart they will have positive determinant. We recall indeed that if we compute multivariate residues iteratively, then the sign of the result depends on the order of the iterations.}.

The Jeffrey-Kirwan residue has remarkable properties. Let us consider two points $\xi, \xi' \in \mathbb{P}^r$ such that the set of cones which contain each of them is the same, then from Def. \eqref{def:jk} we have that:
\begin{equation}
     \mathrm{JK}_{\xi} \tilde{\omega}= \mathrm{JK}_{\xi'} \tilde{\omega}.
\end{equation}
Points $\xi,\xi'$ of this type are said to be in the same \emph{chamber} $\mathfrak{c}$. Chambers can be equivalently characterised as the disconnected components of the set $\mathbb{P}^r$ to which we remove all the codimension one boundaries of all cones $\mathfrak{C}$.

Considering the differential form $ \tilde{\omega}_{\mathcal{Y}}(C,\mathcal{P})$ in \eqref{def:formfiber}, one can show that it has poles on a set of hyperplanes $\mathcal{H}_1,\ldots,\mathcal{H}_{\nu}$. Therefore, we can apply Jeffrey-Kirwan residue to it, obtaining the following result:
\begin{theorem}[\cite{Ferro:2018vpf}]\label{th:jk}
 Given a projective polytope $(\mathbb{P}^N, \mathcal{P})$ with vertices $Z_k, k=1,\ldots,\nu$, its canonical function $\Omega(\mathcal{Y}, \mathcal{P})$ can be obtained by applying the Jeffrey-Kirwan residue to the (covariant) top form $ \tilde{\omega}_{\mathcal{Y}}(C,\mathcal{P})$ on $\mathbb{P}^{\nu-N-1}$ defined as in \eqref{def:formfiber}:
 \begin{equation}\label{eq:jkmasterformula}
      \Omega(\mathcal{Y}, \mathcal{P})=\mathrm{JK}_{\mathfrak{c}} \, \tilde{\omega}_{\mathcal{Y}}(C,\mathcal{P}), 
 \end{equation}
 where $\mathfrak{c}$ is a chamber in $\mathbb{P}^{\nu-N-1}$.
 Moreover, the result is independent form the chosen chamber: there is a bijection between chambers and representations of $ \Omega(\mathcal{Y}, \mathcal{P})$ associated to (regular) triangulations of the polytope $\mathcal{P}$.
\end{theorem}
In summary, the configuration of chambers beautifully encodes all (regular) triangulations of the polytope, and the Jeffrey-Kirwan translates this into an algebraic method to compute the canonical function of the polytope associated to each of these triangulations.
\\

Finally, given a projective polytope $(\mathbb{P}^N,\,\mathcal{P})$, its dual polytope $(\mathbb{P}^N,\,\tilde{\mathcal{P})}$ is defined as the convex hull $\tilde{\mathcal{P}}\,\subset\,\mathbb{P}^N(\mathbb{R})$ identified by the vertices $\mathcal{W}_I^{\mbox{\tiny $(j)$}}$, $j\,=\,1,\,\ldots,\tilde{\nu}$, in the linear dual $\mathbb{P}^N$ of $\mathbb{P}^N$:
\begin{equation}\label{eq:DPolyt}
 \tilde{\mathcal{P}}(\mathcal{Y},\,Z)\: :=\:\left\{\mathcal{Y}\,=\,\sum_{j=1}^{\tilde{\nu}}c_j\mathcal{W}^{\mbox{\tiny $(j)$}}\,\in\,\mathbb{P}^{N}(\mathbb{R})\,\Big|\,c_j\,\ge\,0,\,\forall j\,=\,1,\ldots,\,\tilde{\nu}\right\},
\end{equation}
Notice that the vertices and the facets of the $\tilde{\mathcal{P}}$ respectively correspond to the facets and vertices of $\mathcal{P}$: $\tilde{\mathcal{P}}$ can be defined via a set of inequalities $q_k(\mathcal{Y})\, :=\,\mathcal{Y}_IZ^I_k\,\ge\,0$, ($k\,=\,1,\,\ldots,\,\nu$), with the $Z_k$'s identifying the facets of $\tilde{\mathcal{P}}$. Given a $Z_k$, the vertices $\mathcal{W}_I^{\mbox{\tiny $(j)$}}$ of $\tilde{\mathcal{P}}$ on it are the ones satisfying the relation $Z_k^I\mathcal{W}_I^{\mbox{\tiny $(j_a)$}}\,=\,0$. Hence, considering the canonical form \eqref{eq:CFPWZ} written in terms of the dual vectors $\mathcal{W}$ and interpreting them as the vertices of $\mathcal{P}$, the canonical function $\Omega(\mathcal{Y},\,\mathcal{P})$ is the volume of $\tilde{\mathcal{P}}$.


\subsubsection{Disjoint Unions and Triangulations}\label{subsubsec:DuTr}

Let $(\mathbb{P}^N,\,\mathcal{P}^{\mbox{\tiny $(1)$}})$ and $(\mathbb{P}^N,\,\mathcal{P}^{\mbox{\tiny $(2)$}})$ be two projective polytopes such that $\mathcal{P}^{\mbox{\tiny $(1)$}}\,\cap\,\mathcal{P}^{\mbox{\tiny $(2)$}}\,=\,\varnothing$. Then, their disjoint union $(\mathbb{P}^N,\,\mathcal{P}^{\mbox{\tiny $(1)$}}\,\cup\,\mathcal{P}^{\mbox{\tiny $(2)$}})$ is still a positive geometry, whose boundary components are either boundary components of one $(\mathbb{P}^N,\,\mathcal{P}^{\mbox{\tiny $(j)$}})$, $j\,=\,1,\,2$,  or the disjoint union of the boundary components of $\mathcal{P}^{\mbox{\tiny $(1)$}}$ and $\mathcal{P}^{\mbox{\tiny $(2)$}}$. Furthermore, the canonical form associated with $(\mathbb{P}^N,\,\mathcal{P}^{\mbox{\tiny $(1)$}}\,\cup\,\mathcal{P}^{\mbox{\tiny $(2)$}})$ is given by the sum of the canonical forms of $(\mathbb{P}^N,\,\mathcal{P}^{\mbox{\tiny $(1)$}})$ and $(\mathbb{P}^N,\,\mathcal{P}^{\mbox{\tiny $(2)$}})$:
\begin{equation}\label{eq:DUcf}
    \omega\left(\mathcal{Y},\,\mathcal{P}^{\mbox{\tiny $(1)$}}\cup\mathcal{P}^{\mbox{\tiny $(2)$}}\right)\:=\:\omega(\mathcal{Y},\,\mathcal{P}^{\mbox{\tiny $(1)$}}) + \omega(\mathcal{Y},\,\mathcal{P}^{\mbox{\tiny $(2)$}}).
\end{equation}

Let $(\mathbb{P}^N,\,\mathcal{P}^{\mbox{\tiny $(1)$}})$ and $(\mathbb{P}^N,\,\mathcal{P}^{\mbox{\tiny $(2)$}})$ be two projective polytopes such that 
$\mathcal{P}_{\mbox{\tiny $>0$}}^{\mbox{\tiny $(1)$}}\,\cap\,\mathcal{P}_{\mbox{\tiny $>0$}}^{\mbox{\tiny $(2)$}}\,=\,\varnothing$ and $\mathcal{P}^{\mbox{\tiny $(1)$}}\,\cap\,\mathcal{P}^{\mbox{\tiny $(2)$}}\,=\,\partial\mathcal{P}^{\mbox{\tiny $(12)$}}$ with $(\mathbb{P}^{N-1},\,\partial\mathcal{P}^{\mbox{\tiny $(12)$}})$ having opposite orientation as a boundary component of $(\mathbb{P}^N,\,\mathcal{P}^{\mbox{\tiny $(1)$}})$ or $(\mathbb{P}^N,\,\mathcal{P}^{\mbox{\tiny $(2)$}})$, {\it i.e.} the two polytopes have their interiors disjoints and share a facet with opposite orientation. If $\mathcal{P}\,:=\,\mathcal{P}^{\mbox{\tiny $(1)$}}\,\cup\,\mathcal{P}^{\mbox{\tiny $(2)$}}$, then $(\mathbb{P}^N,\,\mathcal{P})$ is still a polytope, whose boundary components are either boundary components of one $(\mathbb{P}^N,\,\mathcal{P}^{\mbox{\tiny $(j)$}})$, except $(\mathbb{P}^{N-1},\,\partial\mathcal{P}^{\mbox{\tiny $(12)$}})$, or the union of the boundary components of $\mathcal{P}^{\mbox{\tiny $(1)$}}$ and $\mathcal{P}^{\mbox{\tiny $(2)$}}$. The canonical form associated to such an union is still given by the sum of the canonical forms of each polytope as in \eqref{eq:DUcf}, and it is such that the sum of the residues of each individual canonical form $\omega(\mathcal{Y},\,\mathcal{P}^{\mbox{\tiny $(j)$}})$ along the boundary $(\mathbb{P}^{N-1},\,\partial\mathcal{P}^{\mbox{\tiny $(12)$}})$ is zero. The polytopes $(\mathbb{P}^N,\,\mathcal{P}^{\mbox{\tiny $(1)$}})$ and $(\mathbb{P}^N,\,\mathcal{P}^{\mbox{\tiny $(2)$}})$ provide a triangulation of $(\mathbb{P}^N,\,\mathcal{P})$. More generally, if $(\mathbb{P}^N,\,\mathcal{P})$ is a polytope and $\{(\mathbb{P}^N,\,\mathcal{P}^{\mbox{\tiny $(j)$}})\}_{j=1}^{n}$ is a collection of polytopes, then the latter provide a \emph{triangulation} of the former if
\begin{enumerate}[label=(\roman*)]
    \item  $(\mathbb{P}^N,\,\mathcal{P}_{\mbox{\tiny $>0$}}^{\mbox{\tiny $(j)$}})\,\subset\,(\mathbb{P}^N,\,\mathcal{P}_{\mbox{\tiny $>0$}}),\;\forall\,j\:=\:1,\ldots,n$, with compatible orientations;
    \item given $(\mathbb{P}^N,\,\mathcal{P}_{\mbox{\tiny $>0$}}^{\mbox{\tiny $(j)$}})$ and $(\mathbb{P}^N,\,\mathcal{P}_{\mbox{\tiny $>0$}}^{\mbox{\tiny $(l)$}})$, then 
          $\mathcal{P}_{\mbox{\tiny $>0$}}^{\mbox{\tiny $(j)$}}\,\cap\,\mathcal{P}_{\mbox{\tiny $>0$}}^{\mbox{\tiny $(l)$}}\,=\,\varnothing$ $\quad\forall\: j,\,l\,=\,1,\ldots,\tilde{\nu}$ ($j\,\neq\,l$);
    \item $\displaystyle (\mathbb{P}^N,\,\bigcup_{j=1}^{n}\mathcal{P}^{\mbox{\tiny $(j)$}})\:=\:(\mathbb{P}^N,\,\mathcal{P})$;
\end{enumerate}
and, then, the canonical form of $(\mathbb{P}^N,\,\mathcal{P})$ is given by the sum of the canonical forms of the collection $\{(\mathbb{P}^N,\,\mathcal{P}^{\mbox{\tiny $(j)$}})\}_{j=1}^{n}$:
\begin{equation}\label{eq:TRcf}
    \omega(\mathcal{Y},\,\mathcal{P})\:=\:\sum_{j=1}^{n}\omega(\mathcal{Y},\,\mathcal{P}^{\mbox{\tiny $(j)$}}).
\end{equation}

It is possible to further generalise the notion of triangulation. 
Let $\lbrace (\mathbb{P}^N,\,\mathcal{P}^{\mbox{\tiny $(j)$}}) \rbrace _{j=1}^{n+1}$ a collection of polytopes. For any given point $\mathcal{Y}\in\mathbb{P}^N$, let $n^{\mbox{\tiny $(+)$}}_{\mbox{\tiny $\mathcal{Y}$}}$ and $n^{\mbox{\tiny $(-)$}}_{\mbox{\tiny $\mathcal{Y}$}}$ be respectively the number of $\mathcal{P}^{\mbox{\tiny $(j)$}}$ containing $\mathcal{Y}$ ($\mathcal{Y}\notin\partial\mathcal{P}^{\mbox{\tiny $(j)$}}$) with positive/negative orientation of $\mathcal{P}^{\mbox{\tiny $(j)$}}$ at $\mathcal{Y}$. If
\begin{equation}
    \forall\;\mathcal{Y}\in\bigcup_{j=1}^{n+1}\mathcal{P}^{\mbox{\tiny $(j)$}}\,\&\,\mathcal{Y}\notin\partial{P}^{\mbox{\tiny $(j)$}}\,(\forall\,j=1,\ldots,n+1)\: : 
    \:n^{\mbox{\tiny $(+)$}}_{\mbox{\tiny $\mathcal{Y}$}}\:=\:n^{\mbox{\tiny $(-)$}}_{\mbox{\tiny $\mathcal{Y}$}},
\end{equation}
then the collection $\lbrace (\mathbb{P}^N,\,\mathcal{P}^{\mbox{\tiny $(j)$}}) \rbrace_{j=1}^{n+1}$ {\it interior triangulates} the empty set. Consequently, given a collection $\lbrace (\mathbb{P}^N,\,\mathcal{P}^{\mbox{\tiny $(j)$}}) \rbrace_{j=1}^{n+1}$ of polytopes which {\it  interior trangulate} the empty set, then $\lbrace (\mathbb{P}^N,\,\mathcal{P}^{\mbox{\tiny $(n+1)$}}_{-}) \rbrace$ is interior triangulated by  $\lbrace (\mathbb{P}^N,\,\mathcal{P}^{\mbox{\tiny $(j)$}})\rbrace_{j=1}^{n}$\footnote{Here $\mathcal{P}^{\mbox{\tiny $(j)$}}_{-}$ denotes  $\mathcal{P}^{\mbox{\tiny $(j)$}}$ but with reversed orientation.}. If any point $\mathcal{Y}\in\mathbb{P}^N$ is contained in exactly one of the element of the collection, then the interior triangulation reduces to the previous notion of triangulation.

Given a collection $\lbrace (\mathbb{P}^N,\,\mathcal{P}^{\mbox{\tiny $(j)$}}) \rbrace_{j=1}^{n+1}$ of polytopes, it is a {\it canonical-form triangulation} of the empty set if
\begin{equation}\label{eq:CFtrD}
 \sum_{j=1}^{n+1}\omega(\mathcal{Y},\mathcal{P}^{\mbox{\tiny $(j)$}})\:=\:0.
\end{equation}
Consequently, given a collection $\left\{(\mathbb{P}^N,\,\mathcal{P}^{\mbox{\tiny $(j)$}})\right\}_{j=1}^{n+1}$ of projective polytopes which sign triangulates the empty set, we say $(\mathbb{P}^N,\,\mathcal{P}_{-}^{\mbox{\tiny $(n+1)$}})$ is \emph{canonical-form triangulated} by $\left\{(\mathbb{P}^N,\,\mathcal{P}^{\mbox{\tiny $(j)$}})\right\}_{j=1}^{n}$ with
\begin{equation}\label{eq:CFtrP}
    \omega(\mathcal{Y},\,\mathcal{P}_{-}^{\mbox{\tiny $(n+1)$}})\:=\:\sum_{j=1}^{n}\omega(\mathcal{Y},\mathcal{P}^{\mbox{\tiny $(j)$}}),
\end{equation}
where $\mathcal{P}_{-}$ denotes $\mathcal{P}$ with reversed orientation.

These are two different notions of \emph{signed triangulations} \cite{Arkani-Hamed:2017tmz}. In the rest of the paper we will use this latter term indistinctly for both of them.

\subsection{Cosmological Polytopes}\label{subsec:CP}

Let us further specialise to a special class of projective polytopes, the cosmological polytopes \cite{Arkani-Hamed:2017fdk, Benincasa:2019vqr}.

Let $(\mathbb{P}^2,\, \triangle)$ be a triangle and let $(\mathbb{P}^{3n_t-1},\,\{\triangle^{\mbox{\tiny $(j)$}}\}_{j=1}^{n_t})$ a collection of $n_t$ triangles whose vertices are linearly independent as vectors of $\mathbb{R}^{3n_t}$. The cosmological polytopes are defined as those polytopes obtained from such a collection of triangles by {\it intersecting} them in the midpoints of at most two out of their three facets. If $\{(Z_1^{\mbox{\tiny $(j)$}},\,Z_2^{\mbox{\tiny $(j)$}},\,Z_3^{\mbox{\tiny $(j)$}})\}_{j=1}^{n_e}$ are the vertices for $\{\triangle^{\mbox{\tiny $(j)$}}\}_{j=1}^{n_e}$, then the cosmological polytope is a projective polytope $(\mathbb{P}^{3n_t-r-1},\,\mathcal{P})$ with $\mathcal{P}$ being the convex hull
\begin{equation}\label{eq:CPch}
    \resizebox{0.95\hsize}{!}{$\displaystyle%
    \mathcal{P}(\mathcal{Y},\,Z)\, := \,
    \left\{
        \mathcal{Y}\,=\,\sum_{j=1}^{n_e}\sum_{k=1}^{3}c_k^{\mbox{\tiny $(j)$}}Z_k^{\mbox{\tiny $(j)$}}\:\in\:\mathbb{P}^{3n_t-r-1}\,\bigg|\,
        \begin{array}{l}
             c_k^{\mbox{\tiny $(j)$}}\,>\,0,\hspace{3.75cm}\forall\,k\,=\,1,2,3,\,\forall\,j\,\in\,[1,\,n_t]\\
             \{Z_{k-1}^{\mbox{\tiny $(j)$}}+Z_{k}^{\mbox{\tiny $(j)$}}\,\sim\,Z_{k-1}^{\mbox{\tiny $(j')$}}+Z_{k}^{\mbox{\tiny $(j')$}}\}_r, \:k\,=\,1,2,\:\,j\,\neq\,j'\in\,[1,\,n_t]
        \end{array}
    \right\}
    $},
\end{equation}
where $\{Z_{k-1}^{\mbox{\tiny $(j)$}}+Z_{k}^{\mbox{\tiny $(j)$}}\,\sim\,Z_{k-1}^{\mbox{\tiny $(j')$}}+Z_{k}^{\mbox{\tiny $(j')$}}\}_r$ indicates a set of $r\,\in\,[n_t-1,\,2(n_t-1)]$ relations between pairs of vertices of different triangles (see Figure \ref{fig:CP}).

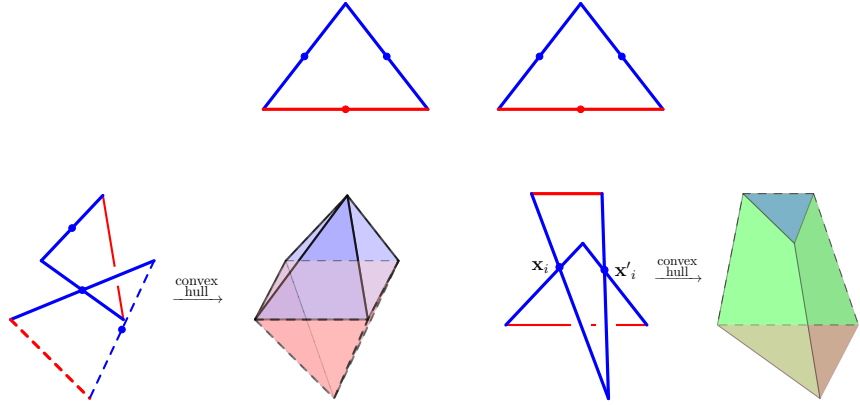
\begin{figure}[h]
 \centering
 \begin{tikzpicture}[line join = round, line cap = round, ball/.style = {circle, draw, align=center, anchor=north, inner sep=0}, 
                     axis/.style={very thick, ->, >=stealth'}, pile/.style={thick, ->, >=stealth', shorten <=2pt, shorten>=2pt}, every node/.style={color=black}, scale={1.25}]
   \begin{scope}[shift={(0,2.5)}, scale={.5}]
   \coordinate (A) at (0,0);
   \coordinate (B) at (-1.75,-2.25);
   \coordinate (C) at (+1.75,-2.25);
   \coordinate (m1) at ($(A)!0.5!(B)$);
   \coordinate (m2) at ($(A)!0.5!(C)$);
   \coordinate (m3) at ($(B)!0.5!(C)$);
   \tikzset{point/.style={insert path={ node[scale=2.5*sqrt(\pgflinewidth)]{.} }}} 

   \draw[color=blue,fill=blue] (m1) circle (2pt);
   \draw[color=blue,fill=blue] (m2) circle (2pt);
   \draw[color=red,fill=red] (m3) circle (2pt);

   \draw[-, very thick, color=blue] (B) -- (A);
   \draw[-, very thick, color=blue] (A) -- (C);  
   \draw[-, very thick, color=red] (B) -- (C);    
  \end{scope}
  \begin{scope}[shift={(2.5,2.5)}, scale={.5}]
   \coordinate (A) at (0,0);
   \coordinate (B) at (-1.75,-2.25);
   \coordinate (C) at (+1.75,-2.25);
   \coordinate (m1) at ($(A)!0.5!(B)$);
   \coordinate (m2) at ($(A)!0.5!(C)$);
   \coordinate (m3) at ($(B)!0.5!(C)$);
   \tikzset{point/.style={insert path={ node[scale=2.5*sqrt(\pgflinewidth)]{.} }}} 

   \draw[color=blue,fill=blue] (m1) circle (2pt);
   \draw[color=blue,fill=blue] (m2) circle (2pt);
   \draw[color=red,fill=red] (m3) circle (2pt);

   \draw[-, very thick, color=blue] (B) -- (A);
   \draw[-, very thick, color=blue] (A) -- (C);  
   \draw[-, very thick, color=red] (B) -- (C);    
  \end{scope}
  \begin{scope}[scale={.4}, shift={(-7,2)}, transform shape]
   \pgfmathsetmacro{\factor}{1/sqrt(2)};  
   \coordinate  (B2) at (1.5,-3,-1.5*\factor);
   \coordinate  (A1) at (-1.5,-3,-1.5*\factor);
   \coordinate  (B1) at (1.5,-3.75,1.5*\factor);
   \coordinate  (A2) at (-1.5,-3.75,1.5*\factor);  
   \coordinate  (C1) at (0.75,-.65,.75*\factor);
   \coordinate  (C2) at (0.4,-6.05,.75*\factor);
   \coordinate (Int) at (intersection of A2--B2 and B1--C1);
   \coordinate (Int2) at (intersection of A1--B1 and A2--B2);

   \tikzstyle{interrupt}=[
    postaction={
        decorate,
        decoration={markings,
                    mark= at position 0.5 
                          with
                          {
                            \node[rectangle, color=white, fill=white, below=-.1 of Int] {};
                          }}}
   ]
   
   \draw[interrupt,thick,color=red] (B1) -- (C1);
   \draw[-,very thick,color=blue] (A1) -- (B1);
   \draw[-,very thick,color=blue] (A2) -- (B2);
   \draw[-,very thick,color=blue] (A1) -- (C1);
   \draw[-, dashed, very thick, color=red] (A2) -- (C2);
   \draw[-, dashed, thick, color=blue] (B2) -- (C2);

   \coordinate (x2) at ($(A1)!0.5!(B1)$);
   \draw[fill,color=blue] (x2) circle (2.5pt);   
   \coordinate (x1) at ($(C1)!0.5!(A1)$);
   \draw[fill,color=blue] (x1) circle (2.5pt);
   \coordinate (x3) at ($(B2)!0.5!(C2)$);
   \draw[fill,color=blue] (x3) circle (2.5pt);

   \node[right=2.25cm of x2] (cht) {$\displaystyle\xrightarrow{\substack{\mbox{convex} \\ \mbox{hull}}}$};
  \end{scope}
  \begin{scope}[scale={.4}, shift={(-.5,2)}, transform shape]
   \pgfmathsetmacro{\factor}{1/sqrt(2)};  
   \coordinate  (B2) at (1.5,-3,-1.5*\factor);
   \coordinate  (A1) at (-1.5,-3,-1.5*\factor);
   \coordinate  (B1) at (1.5,-3.75,1.5*\factor);
   \coordinate  (A2) at (-1.5,-3.75,1.5*\factor);  
   \coordinate  (C1) at (0.75,-.65,.75*\factor);
   \coordinate  (C2) at (0.4,-6.05,.75*\factor);

   \draw[-,dashed,fill=blue!30, opacity=.7] (A1) -- (B2) -- (C1) -- cycle;
   \draw[-,thick,fill=blue!20, opacity=.7] (A1) -- (A2) -- (C1) -- cycle;
   \draw[-,thick,fill=blue!20, opacity=.7] (B1) -- (B2) -- (C1) -- cycle;
   \draw[-,thick,fill=blue!35, opacity=.7] (A2) -- (B1) -- (C1) -- cycle;

   \draw[-,dashed,fill=red!30, opacity=.3] (A1) -- (B2) -- (C2) -- cycle;
   \draw[-,dashed, thick, fill=red!50, opacity=.5] (B2) -- (B1) -- (C2) -- cycle;
   \draw[-,dashed,fill=red!40, opacity=.3] (A1) -- (A2) -- (C2) -- cycle;
   \draw[-,dashed, thick, fill=red!45, opacity=.5] (A2) -- (B1) -- (C2) -- cycle;
  \end{scope}
  \begin{scope}[scale={.5}, shift={(4.5,.75)}, transform shape]
   \pgfmathsetmacro{\factor}{1/sqrt(2)};
   \coordinate  (c1b) at (0.75,0,-.75*\factor);
   \coordinate  (b1b) at (-.75,0,-.75*\factor);
   \coordinate  (a2b) at (0.75,-.65,.75*\factor);
   
   \coordinate  (c2b) at (1.5,-3,-1.5*\factor);
   \coordinate  (b2b) at (-1.5,-3,-1.5*\factor);
   \coordinate  (a1b) at (1.5,-3.75,1.5*\factor); 

  \coordinate (Int1) at (intersection of b2b--c2b and b1b--a1b);
   \coordinate (Int2) at (intersection of b2b--c2b and c1b--a1b);
   \coordinate (Int3) at (intersection of b2b--a2b and b1b--a1b);
   \coordinate (Int4) at (intersection of a2b--c2b and c1b--a1b);
   \tikzstyle{interrupt}=[
    postaction={
        decorate,
        decoration={markings,
                    mark= at position 0.5 
                          with
                          {
                            \node[rectangle, color=white, fill=white] at (Int1) {};
                            \node[rectangle, color=white, fill=white] at (Int2) {};                            
                          }}}
   ]

   \node at (c1b) (c1c) {};
   \node at (b1b) (b1c) {};
   \node at (a2b) (a2c) {};
   \node at (c2b) (c2c) {};
   \node at (b2b) (b2c) {};
   \node at (a1b) (a1c) {};

   \draw[interrupt,thick,color=red] (b2b) -- (c2b);
   \draw[-,very thick,color=red] (b1b) -- (c1b);
   \draw[-,very thick,color=blue] (b1b) -- (a1b);
   \draw[-,very thick,color=blue] (a1b) -- (c1b);   
   \draw[-,very thick,color=blue] (b2b) -- (a2b);
   \draw[-,very thick,color=blue] (a2b) -- (c2b);

   \node[ball,text width=.15cm,fill,color=blue, above=-.06cm of Int3, label=left:{\large ${\bf x}_i$}] (Inta) {};
   \node[ball,text width=.15cm,fill,color=blue, above=-.06cm of Int4, label=right:{\large ${\bf x'}_i$}] (Intb) {};

   \node[right=.875cm of Intb, scale=.75] (chl) {$\displaystyle\xrightarrow{\substack{\mbox{convex} \\ \mbox{hull}}}$};
  \end{scope}
  \begin{scope}[scale={.5}, shift={(9,.75)}, transform shape]
   \pgfmathsetmacro{\factor}{1/sqrt(2)};
   \coordinate  (c1b) at (0.75,0,-.75*\factor);
   \coordinate  (b1b) at (-.75,0,-.75*\factor);
   \coordinate  (a2b) at (0.75,-.65,.75*\factor);
   
   \coordinate  (c2b) at (1.5,-3,-1.5*\factor);
   \coordinate  (b2b) at (-1.5,-3,-1.5*\factor);
   \coordinate  (a1b) at (1.5,-3.75,1.5*\factor);

   \draw[-,dashed,fill=green!50,opacity=.6] (c1b) -- (b1b) -- (b2b) -- (c2b) -- cycle;
   \draw[draw=none,fill=red!60, opacity=.45] (c2b) -- (b2b) -- (a1b) -- cycle;
   \draw[-,fill=blue!,opacity=.3] (c1b) -- (b1b) -- (a2b) -- cycle; 
   \draw[-,fill=green!50,opacity=.4] (b1b) -- (a2b) -- (a1b) -- (b2b) -- cycle;
   \draw[-,fill=green!45!black,opacity=.2] (c1b) -- (a2b) -- (a1b) -- (c2b) -- cycle;  
  \end{scope}
 \end{tikzpicture}
\caption{Cosmological polytopes constructed from $(\mathbb{P}^5,\,\{\triangle^{\mbox{\tiny $(j)$}}\}_{j=1}^2)$. The (red) blue facets in the triangles indicate the ones which can(not) be {\it intersected}. The figures at the bottom left and bottom right depict the convex hull of the vertices of triangles after one and two linear relations has been respectively imposed, so that the first one lives in $\mathbb{P}^4$ and the second in $\mathbb{P}^3$ . 
 }
 \label{fig:CP}
\end{figure}

The construction just presented can be extended. Let $(\mathbb{P}^1,\,\mathcal{S})$ be a segment, which can be seen as a codimension-$1$ projection of a triangle, so that its two intersectable facets are projected onto each other to be the interior of the segment, and its non-intersectable one is shrunk to a point \cite{Benincasa:2018ssx} that will be referred to as the non-intersectable vertex of the segment. Let $\left(\mathbb{P}^{3n_t+2n_h-1},\,\{\{\triangle^{\mbox{\tiny $(j)$}}\}_{j=1}^{n_t},\,\{\mathcal{S}^{\mbox{\tiny $(g)$}}\}_{g=1}^{n_h}\}\right)$ the collection of $n_e$ triangles and $n_h$ segments whose vertices are all linearly independent of each other as vectors of $\mathbb{R}^{3n_t+2n_h}$. The {\it extended cosmological polytopes} are then defined as the projective polytopes $(\mathbb{P}^{3n_t+2n_h-1},\,\mathcal{P})$, where $\mathcal{P}$ is the convex hull of all the vertices of the triangles and segments after triangles and segments are intersected in their midpoints. Hence $(\mathbb{P}^{3n_t+2n_h-r-1},\,\mathcal{P})$ can be constructed out of triangles only (for $n_h\,=\,0$), segments only (for $n_t\,=\,0$ -- in this case there is  just one of such polytopes for fixed $n_h$ given that any segment has only one midpoint where it can get intersected), or both triangles and segments, which is the most general case \cite{Benincasa:2019vqr}.

There is a $1-1$ correspondence between cosmological polytopes \footnote{Since now on with {\it cosmological polytopes} we will identify its extended notion, omitting to explicitly specify it for the sake of conciseness.} $(\mathbb{P}^{3n_t+2n_h-r-1},\,\mathcal{P})$, and graphs $\mathcal{G}_{\mbox{\tiny $\mathcal{P}$}}$. To each triangle let us associate a {\it two-site line graph}, {\it i.e.} a graph with two {\it sites} \footnote{In order to avoid confusion in the terminology, we will reserve {\it vertices} for the highest codimension face of the projective polytopes, and use {\it sites} for the graphs.}, one for each intersectable facet, and one edge, which corresponds to the non-intersectable facet; as far as the segment is concerned, thinking of it as a codimension-$1$ projection of a triangle, the associated graph is a {\it tadpole} (or {\it one-loop one-site}) graph, {\it i.e.} a graph with a single site, corresponding to its interior which is given by the two intersectable facets of the original triangle projected onto each other, and a loop closing itself onto this site, which correponds to its non-intersectable facet which got shrunk to a point. Then, given a cosmological polytope $(\mathbb{P}^{3n_t+2n_h-r-1},\,\mathcal{P})$ generated as an intersection of a collection of triangles and segments, its associated graph $\mathcal{G}_{\mbox{\tiny $\mathcal{P}$}}$ is obtained by merging a collection of two-site line graphs and tadpoles in their sites:
\begin{equation*}
 \begin{tikzpicture}[line join = round, line cap = round, ball/.style = {circle, draw, align=center, anchor=north, inner sep=0}, 
                     axis/.style={very thick, ->, >=stealth'}, pile/.style={thick, ->, >=stealth', shorten <=2pt, shorten>=2pt}, every node/.style={color=black}, scale={1}]
  \begin{scope}[scale={.5}, transform shape]
   \coordinate (A) at (0,0);
   \coordinate (B) at (-1.75,-2.25);
   \coordinate (C) at (+1.75,-2.25);
   \coordinate (m1) at ($(A)!0.5!(B)$);
   \coordinate (m2) at ($(A)!0.5!(C)$);
   \coordinate (m3) at ($(B)!0.5!(C)$);
   \tikzset{point/.style={insert path={ node[scale=2.5*sqrt(\pgflinewidth)]{.} }}} 

   \draw[color=blue,fill=blue] (m1) circle (2pt);
   \draw[color=blue,fill=blue] (m2) circle (2pt);
   \draw[color=red,fill=red] (m3) circle (2pt);

   \draw[-, very thick, color=blue] (B) -- (A);
   \draw[-, very thick, color=blue] (A) -- (C);  
   \draw[-, very thick, color=red] (B) -- (C);   
   
   \coordinate (s1) at ($(m2)+(4,0)$);
   \coordinate (s2) at ($(s1)+(3,0)$);
   
   \draw[color=blue,fill=blue] (s1) circle (3pt);
   \draw[color=blue,fill=blue] (s2) circle (3pt);
   \draw[-,very thick, color=red] (s1) -- (s2);
   
   \node[scale=1.75] (eq1) at ($(m2)!0.5!(s1)$) {$\displaystyle\longleftrightarrow$};
   
   \coordinate (V1) at ($(A)+(12,0)$);
   \coordinate (V2) at ($(V1)-(0,2.25)$);
   \coordinate (m3) at ($(V1)!0.5!(V2)$);
   
   \draw[-, very thick, color=blue] (V1) -- (V2);
   \draw[color=blue, fill=blue] (m3) circle (2pt);
   \draw[color=red, fill=red] (V2) circle (2pt);
   
   \coordinate (l3) at ($(m3)+(8,0)$);
   \coordinate (s3) at ($(l3)-(1.25cm,0)$);
   \draw[very thick, color=red] (l3) circle (1.25cm); 
   \draw[fill,color=blue] (s3) circle (3pt);
   
   \node[scale=1.75] (eq2) at ($(m3)!0.5!(s3)$) {$\displaystyle\longleftrightarrow$};
  \end{scope}
  \begin{scope}[shift={(0,-2)}, scale={.5}, transform shape]
   \pgfmathsetmacro{\factor}{1/sqrt(2)};
   \coordinate  (c1b) at (0.75,0,-.75*\factor);
   \coordinate  (b1b) at (-.75,0,-.75*\factor);
   \coordinate  (a2b) at (0.75,-.65,.75*\factor);
   
   \coordinate  (c2b) at (1.5,-3,-1.5*\factor);
   \coordinate  (b2b) at (-1.5,-3,-1.5*\factor);
   \coordinate  (a1b) at (1.5,-3.75,1.5*\factor); 

   \coordinate (Int1) at (intersection of b2b--c2b and b1b--a1b);
   \coordinate (Int2) at (intersection of b2b--c2b and c1b--a1b);
   \coordinate (Int3) at (intersection of b2b--a2b and b1b--a1b);
   \coordinate (Int4) at (intersection of a2b--c2b and c1b--a1b);
   \tikzstyle{interrupt}=[
    postaction={
        decorate,
        decoration={markings,
                    mark= at position 0.5 
                          with
                          {
                            \node[rectangle, color=white, fill=white] at (Int1) {};
                            \node[rectangle, color=white, fill=white] at (Int2) {};                            
                          }}}
   ]

   \node at (c1b) (c1c) {};
   \node at (b1b) (b1c) {};
   \node at (a2b) (a2c) {};
   \node at (c2b) (c2c) {};
   \node at (b2b) (b2c) {};
   \node at (a1b) (a1c) {};

   \draw[interrupt,thick,color=red] (b2b) -- (c2b);
   \draw[-,very thick,color=red] (b1b) -- (c1b);
   \draw[-,very thick,color=blue] (b1b) -- (a1b);
   \draw[-,very thick,color=blue] (a1b) -- (c1b);   
   \draw[-,very thick,color=blue] (b2b) -- (a2b);
   \draw[-,very thick,color=blue] (a2b) -- (c2b);

   \node[ball,text width=.15cm,fill,color=blue, above=-.06cm of Int3] (Inta) {};
   \node[ball,text width=.15cm,fill,color=blue, above=-.06cm of Int4] (Intb) {};

   \coordinate (Lc) at ($(Intb)+(5.25,0)$);
   \coordinate (S1) at ($(Lc)-(1.25cm,0)$);
   \coordinate (S2) at ($(Lc)+(1.25cm,0)$);
   
   \draw[-,color=red, very thick] (Lc) circle (1.25cm);
   \draw[color=blue, fill=blue] (S1) circle (3pt);
   \draw[color=blue, fill=blue] (S2) circle (3pt);
   
   \node[scale=1.75] (eq3) at ($(Intb)!0.5!(S1)$) {$\displaystyle\longleftrightarrow$};
  \end{scope}
  \begin{scope}[shift={(6,-1.75)}, scale={.5}, transform shape]
   \pgfmathsetmacro{\factor}{1/sqrt(2)};  
   \coordinate  (B2) at (1.5,-3,-1.5*\factor);
   \coordinate  (A1) at (-1.5,-3,-1.5*\factor);
   \coordinate  (B1) at (1.5,-3.75,1.5*\factor);
   \coordinate  (A2) at (-1.5,-3.75,1.5*\factor);  
   \coordinate  (C1) at (0.75,-.65,.75*\factor);
   \coordinate  (C2) at (0.4,-6.05,.75*\factor);
   \coordinate (Int) at (intersection of A2--B2 and B1--C1);
   \coordinate (Int2) at (intersection of A1--B1 and A2--B2);

   \tikzstyle{interrupt}=[
    postaction={
        decorate,
        decoration={markings,
                    mark= at position 0.5 
                          with
                          {
                            \node[rectangle, color=white, fill=white, below=-.1 of Int] {};
                          }}}
   ]
  
   \draw[interrupt,very thick,color=blue] (A1) -- (B1); 
   \draw[interrupt,very thick,color=blue] (A2) -- (B2);
   \draw[-,very thick,color=red] (B1) -- (C1);
   \draw[-,very thick,color=blue] (A1) -- (C1);

   \coordinate (x2) at ($(A1)!0.5!(B1)$);
   \draw[fill,color=blue] (x2) circle (2.5pt);   
   \coordinate (x1) at ($(C1)!0.5!(A1)$);
   \draw[fill,color=blue] (x1) circle (2.5pt);
   \draw[fill,color=red] (A2) circle (2.5pt);
   \draw[fill,color=blue] (B2) circle (2.5pt);

   \coordinate (B2b) at (B2);
   \coordinate (A2b) at (A2);
   
   \coordinate (lc) at ($(B2b)+(4,0)$);
   \coordinate (S3) at ($(lc)+(.75,0)$);
   \coordinate (S4) at ($(S3)+(3,0)$);
   
   \draw[very thick, color=red] (lc) circle (.75);
   \draw[-, very thick, color=red] (S3) -- (S4);
   \draw[color=blue, fill=blue] (S3) circle (3pt);
   \draw[color=blue, fill=blue] (S4) circle (3pt);
   
   \coordinate (tl) at ($(lc)-(.75,0)$);
   \node[scale=1.75] (eq4) at ($(B2b)!0.5!(tl)$) {$\displaystyle\longleftrightarrow$};
  \end{scope}
 \end{tikzpicture}
\end{equation*}

Notice that the number of edges $n_e$ of a graph $\mathcal{G}_{\mbox{\tiny $\mathcal{P}$}}$ is given by the sum of the number of triangles and segments, while its number of sites $n_s$ depends on the number $r$ of intersections: $n_e\:=\:n_t+n_h,\; n_s\:=\:2n_t+n_h-r$.  Thus, given a graph $\mathcal{G}$ it is possible to associate a polytope $(\mathbb{P}^{n_s+n_e-1},\,\mathcal{P}_{\mbox{\tiny $\mathcal{G}$}})$, with the convex hull $\mathcal{P}_{\mbox{\tiny $G$}}\,=\,\mathcal{P}$ as described above.

Given a cosmological polytope and its associated graph $\mathcal{G}$, there is a canonical way to assign a local coordinate chart in projective space for parametrising the polytope with a correspondence between such local coordinates and weights on sites and edges of $\mathcal{G}$. Let us consider the collection $\left(\mathbb{P}^{3n_t+2n_h-1},\,\{\{\triangle^{\mbox{\tiny $(j)$}}\}_{j=1}^{n_t},\,\{\mathcal{S}^{\mbox{\tiny $(g)$}}\}_{g=1}^{n_h}\}\right)$ of $n_t$ triangles and $n_h$ segments and choose the midpoints of the facets of the triangles, the midpoints of the segments as well as non-intersectable vertex of the segment, as a basis for $\mathbb{R}^{3n_e+2n_h}$. Let us indicate these vectors as $\{\mathbf{x}_j,\,\mathbf{y}_j,\,\mathbf{x}'_j\}$ for $\triangle^{\mbox{\tiny $(j)$}}$, where $\mathbf{x}_j,\,\mathbf{x}'_j$ are the midpoints of the intersectable sides of $\triangle^{\mbox{\tiny $(j)$}}$ and $\mathbf{y}_j$ is the midpoint for the non-intersectable one, and let $\{\mathbf{x}''_g,\,\mathbf{h}_g\}$ be the midpoint and the non-intersectable vertex for $\mathcal{S}^{\mbox{\tiny $(g)$}}$ respectively. Then, on this basis a generic point $\mathcal{Y}\,\in\,\mathbb{P}^{3n_e+2n_h-1}$ can be written as 
\begin{equation}
    \mathcal{Y}\:=\:\sum_{j=1}^{n_e}(x_j\mathbf{x}_j+y_j\mathbf{y}_j+x'_j\mathbf{x}'_j) + \sum_{g=1}^{n_h}(x''_g\mathbf{x}''_g+h_g\mathbf{h}_h),
\end{equation}
where the coefficients $\{\{x_j\,y_j,\,x'_j\},\,\{x''_g,\,h_g\}\}$ are the homogeneous coordinates in this patch. Then, in the association of a two-site line graph to a triangle, one assigns $x_j$ and $x'_j$ as weights of the graph sites and $y_j$ as weight of the edge connecting the sites; similarly, in the association of a tadpole to a segment, one assigns $x''_j$ and $h_j$ to the site and edge respectively. In constructing a cosmological polytope, each intersection condition \eqref{eq:CPch} identifies two elements of this basis, reducing the midpoint coordinates by one and, hence, each two-site line and tadpole subgraphs has the very same weight assignation as just described, but identifying the weights of common sites. Taking the midpoint basis, $\mathcal{P}$ is the convex hull of the vertices 
\begin{equation}\label{eq:CPxyx}
    \{\mathbf{x}_j-\mathbf{y}_j+\mathbf{x}'_j,\: \mathbf{x}_j+\mathbf{y}_j-\mathbf{x}'_j,\:-\mathbf{x}_j+\mathbf{y}_j+\mathbf{x}'_j\},\qquad 
    \{2\mathbf{x}''_g-\mathbf{h}_g,\:\mathbf{h}_g\}
\end{equation}
with suitable identifications among the midpoint vectors.

The definition of the cosmological polytopes as intersection of triangles and segments, allows for a simple and direct characterisation of its face structure. Given a cosmological polytope $(\mathbb{P}^{n_s+n_e-1},\,\mathcal{P})$ with associated graph $\mathcal{G}$, any of its faces $\mathcal{F}$ is given as a collection $\mathcal{V}_{\mathcal{F}}$ of vertices $Z_a^I$ ($a\,=\,1,\ldots,3n_t+2n_h$) of $\mathcal{P}$ such that $\mathcal{W}_IZ^I_a\,=\,0$, where $\mathcal{W}_I\,:=\,\tilde{x}_{sI}\mathbf{\tilde{x}}_{sI}+\tilde{y}_{eI}\mathbf{\tilde{y}}_e+\tilde{h}_g\mathbf{\tilde{h}}_g$\footnote{Here the summation over the indices $s$, $e$ and $g$ is understood, with $s$ running on the number of sites of the associated graph $\mathcal{G}$, $e$ on the number of its edges connecting two different sites, and $g$ on the number of its tadpoles subgraphs.} is the hyperplane in $\mathbb{P}^{n_s+n_e-1}$ where the facet lives such that, compatibly with the constraints on the midpoints of the generating triangles and segments,  $\mathbf{\tilde{x}}_{sI}\mathbf{x}_{s'}^I\,=\,\delta_{ss'}$, $\mathbf{\tilde{y}}_{eI}\mathbf{y}_{e'}^I\,=\,\delta_{e e'}$, $\mathbf{\tilde{h}}_{gI}\mathbf{h}_{g'}^I\,=\,\delta_{g g'}$, and $\mathbf{\tilde{y}}_{eI}\mathbf{y}_{e'}^I\,=\,\delta_{e e'}$ with all the other scalar products between vectos and co-vectors vanishing. All the other vertices of $\mathcal{P}$ which are not on the facet identified by the hyperplane $\mathcal{W}_i$ are such that $\mathcal{W}_IZ^{I}$. Each of these hyperplanes is in a $1-1$ correspondence with a subgraph $\mathfrak{g}\,\subseteq\,\mathcal{G}$, so that given any subgraph $\mathfrak{g}\subseteq\mathcal{G}$, it can be written as $\mathcal{W}_I\:=\:\sum_{s\in\mathfrak{g}}\tilde{x}_{s}\mathbf{\tilde{x}}_s+\sum_{e\in\mathcal{E}^{\mbox{\tiny ext}}_{\mathfrak{g}}}\tilde{y}_e\mathbf{\tilde{y}}_e+\sum_{g\in\mathcal{H}^{\mbox{\tiny ext}}_{\mathfrak{g}}}\tilde{h}_g\mathbf{\tilde{h}}_g$, with $\mathcal{E}^{\mbox{\tiny ext}}_{\mathfrak{g}}$ and $\mathcal{H}^{\mbox{\tiny ext}}_{\mathfrak{g}}$ being the sets of edges and tadpoles respectively which are external to the subgraph $\mathfrak{g}$ and depart from the sites of $\mathfrak{g}$.

The correspondence between cosmological polytopes and graphs allows to extract all the information about the polytope from the associated graph. For example, it allows to know all the vertices belonging to a certain face identified by an hyperplane $\mathcal{W}$, by introducing a marking on the graphs that identifies those which {\it do not} live on $\mathcal{W}$
\begin{equation*}
 \begin{tikzpicture}[ball/.style = {circle, draw, align=center, anchor=north, inner sep=0}, cross/.style={cross out, draw, minimum size=2*(#1-\pgflinewidth), inner sep=0pt, outer sep=0pt}, scale={1.125}, transform shape]
  \begin{scope}
   \node[ball,text width=.18cm,fill,color=black,label=below:{\footnotesize $x_{s\phantom{'}}$}] at (0,0) (v1) {};
   \node[ball,text width=.18cm,fill,color=black,label=below:{\footnotesize $x_{s'}$},right=1.5cm of v1.east] (v2) {};  
   \draw[-,thick,color=black] (v1.east) edge node [text width=.18cm,below=.1cm,midway] {\footnotesize $y_e$} (v2.west);
   \node[very thick, cross=4pt, rotate=0, color=blue, right=.7cm of v1.east]{};
   \coordinate (x) at ($(v1)!0.5!(v2)$);
   \node[below=.375cm of x, scale=.85] (lb1) {$\mathcal{W}\cdot({\bf x}_s+{\bf x}_{s'}-{\bf y}_e)>\,0$};  
  \end{scope}
  \begin{scope}[shift={(4.5,0)}]
   \node[ball,text width=.18cm,fill,color=black,label=below:{\footnotesize $x_{s\phantom{'}}$}] at (0,0) (v1) {};
   \node[ball,text width=.18cm,fill,color=black,label=below:{\footnotesize $x_{s'}$},right=1.5cm of v1.east] (v2) {};  
   \draw[-,thick,color=black] (v1.east) edge node [text width=.18cm,below=.1cm,midway] {\footnotesize $y_e$} (v2.west);
   \node[very thick, cross=4pt, rotate=0, color=blue, left=.1cm of v2.west]{};
   \coordinate (x) at ($(v1)!0.5!(v2)$);
   \node[below=.375cm of x, scale=.85] (lb1) {$\mathcal{W}\cdot({\bf x}_{s'}+{\bf y}_e-{\bf x}_s)>\,0$};  
  \end{scope}
  \begin{scope}[shift={(9,0)}]
   \node[ball,text width=.18cm,fill,color=black,label=below:{\footnotesize $x_{s\phantom{'}}$}] at (0,0) (v1) {};
   \node[ball,text width=.18cm,fill,color=black,label=below:{\footnotesize $x_{s'}$},right=1.5cm of v1.east] (v2) {};  
   \draw[-,thick,color=black] (v1.east) edge node [text width=.18cm,below=.1cm,midway] {\footnotesize $y_e$} (v2.west);
   \node[very thick, cross=4pt, rotate=0, color=blue, right=.1cm of v1.east]{};
   \coordinate (x) at ($(v1)!0.5!(v2)$);
   \node[below=.375cm of x, scale=.85] (lb1) {$\mathcal{W}\cdot({\bf x}_s+{\bf y}_e-{\bf x}_{s'})>\,0$};  
  \end{scope}
  \begin{scope}[shift={(2.25,-2.25)}, scale={.75}]
   \coordinate (LC) at (,0); 
   \coordinate [label=left:{$\displaystyle x_s$}] (x) at ($(LC)+(-1.25cm,0)$);
   \coordinate [label=right:{$\displaystyle h_g$}] (h) at ($(LC)+(1.25cm,0)$);    
   \tikzset{point/.style={insert path={ node[scale=2.5*sqrt(\pgflinewidth)]{.} }}} 

   \draw[very thick] (LC) circle (1.25cm); 
   \draw[fill] (x) circle (3pt);
   \node[very thick, cross=4pt, rotate=0, color=blue] (X2) at (h) {};   
  
   \coordinate [label={$\displaystyle\mathcal{W}\cdot(2\mathbf{x}_s-\mathbf{h}_g)>0$}] (hyp1) at ($(LC)-(0,2)$);
  \end{scope}
  \begin{scope}[shift={(7,-2.25)}, scale={.75}]
   \coordinate (LC) at (,0); 
   \coordinate [label=left:{$\displaystyle x_s$}] (x) at ($(LC)+(-1.25cm,0)$);
   \coordinate [label=right:{$\displaystyle h_g$}] (h) at ($(LC)+(1.25cm,0)$);    
   \tikzset{point/.style={insert path={ node[scale=2.5*sqrt(\pgflinewidth)]{.} }}} 

   \draw[very thick] (LC) circle (1.25cm); 
   \draw[fill] (x) circle (3pt);
   \coordinate (va) at ($(x)+(.02,.375)$);
   \coordinate (vb) at ($(x)+(.02,-.375)$);
   \node[very thick, cross=4pt, rotate=0, color=blue] (X1a) at (va) {};      
   \node[very thick, cross=4pt, rotate=0, color=blue] (X1b) at (vb) {};  
  
   \coordinate [label={$\displaystyle\mathcal{W}\cdot\mathbf{h}_g>0$}] (hyp1) at ($(LC)-(0,2)$);
  \end{scope}
 \end{tikzpicture}
\end{equation*}
where the two vertices indicated by a marking close to the only site indicate the very same vertex $\mathbf{h}$. Hence considering a general face of a cosmological polytope, the associated graph $\mathcal{G}$ gets marked in the middle of its edges which are internal to the subgraph $\mathfrak{g}$, and in the extreme close to $\mathfrak{g}$ for those edges which are external to $\mathfrak{g}$:
\begin{equation*}
  \begin{tikzpicture}[ball/.style = {circle, draw, align=center, anchor=north, inner sep=0}, cross/.style={cross out, draw, minimum size=2*(#1-\pgflinewidth), inner sep=0pt, outer sep=0pt}, scale=1, transform shape]
   \begin{scope}[scale=.8, transform shape]
    \node[ball,text width=.18cm,fill,color=black] at (0,0) (x1) {};    
    \node[ball,text width=.18cm,fill,color=black,right=1.2cm of x1.east] (x2) {};    
    \node[ball,text width=.18cm,fill,color=black,right=1.2cm of x2.east] (x3) {};
    \node[ball,text width=.18cm,fill,color=black] at (-1,.8) (x4) {};    
    \node[ball,text width=.18cm,fill,color=black] at (-1,-.8) (x5) {};    
    \node[ball,text width=.18cm,fill,color=black] at (-1.7,-2) (x6) {};    
    \node[ball,text width=.18cm,fill,color=black] at (-.3,-2) (x7) {};

    \node[above=.35cm of x5.north] (ref2) {};
    \coordinate (Int2) at (intersection of x5--x1 and ref2--x2);  

    \def\r{.225}
    \pgfmathsetmacro\x{\r*cos{60}};
    \pgfmathsetmacro\y{\r*sin{60}};
    \coordinate (c1u) at ($(x1)+(\x,\y)$);
    \coordinate (c1l) at ($(x1)!-0.175!(x2)$);
    \coordinate (c1r) at ($(x1)+(\x,-\y)$);
    \coordinate (c2u) at ($(x2)+(0,.2cm)$);
    \coordinate (c2d) at ($(x2)-(0,.2cm)$);
    \coordinate (c3) at ($(x3)!-0.15!(x2)$);
    \coordinate (c4) at ($(x4)!-0.15!(x1)$);
    \coordinate (c5u) at ($(x5)+(-\x,\y)$);
    \coordinate (c5r) at ($(x5)+(.3cm, 0)$);
    \coordinate (c5d) at ($(x5)+(0,-.3cm)$);
    \coordinate (c6) at ($(x6)!-0.15!(x5)$);
    \coordinate (c7) at ($(x7)!-0.15!(x5)$);

    \coordinate (c1xu) at ($(c1u)+(\x,\y)$);
    \coordinate (c1xl) at ($(x1)!-0.35!(x2)$);
    \coordinate (c1xr) at ($(x1)+(\x,-\y)+(\x,-\y)$);
    \coordinate (c2xu) at ($(x2)+(0,.4cm)$);
    \coordinate (c2xd) at ($(x2)-(0,.4cm)$);
    \coordinate (c3x) at ($(x3)+(.4cm,0)$);
    \coordinate (c4x) at ($(x4)!-.3!(x1)$);
    \coordinate (c5xu) at ($(c5u)+(-\x,\y)$);
    \coordinate (c5xr) at ($(c5r)+(.25cm,0)$);
    \coordinate (c5xd) at ($(c5d)+(0,-.25cm)$);
    \coordinate (c6x) at ($(x6)!-0.3!(x5)$);
    \coordinate (c7x) at ($(x7)!-0.3!(x5)$);

    \draw[-,thick,color=black] (x1) -- (x2) -- (x3); 
    \draw[-,thick,color=black] (x1) -- (x4);
    \draw[-,thick,color=black] (x5) -- (x1);
    \draw[-,thick,color=black] (x5) -- (x7);   
    \draw[-,thick,color=black] (x5) -- (x6); 

    \draw[thick] (c1u) circle (.2cm);
    \draw[thick] (c1l) ellipse (.25cm and .15cm);
    \draw[thick] (c1r) circle (.2cm);
    \draw[thick] (c2u) circle (.2cm);    
    \draw[thick] (c2d) circle (.2cm);
    \draw[thick] (c3) circle (.2cm);
    \draw[thick] (c4) circle (.2cm);
    \draw[thick] (c5u) circle (.2cm);
    \draw[thick] (c5r) ellipse (.25cm and .15cm);
    \draw[thick] (c5d) ellipse (.15cm and .25cm);
    \draw[thick] (c6) circle (.2cm);    
    \draw[thick] (c7) circle (.2cm);

    \def\rr{.113}
    \pgfmathsetmacro\xx{\rr*cos{60}};
    \pgfmathsetmacro\yy{\rr*sin{60}};
    \coordinate (a1u) at ($(c1xu)+(\xx,\yy)$);
    \coordinate (a1l) at ($(x1)!-0.7!(x2)$);
    \coordinate (a1r) at ($(c1xr)+(\xx,-\yy)$);
    \coordinate (a2u) at ($(c2xu)+(0, .25cm)$);
    \coordinate (a2d) at ($(c2xd)-(0, .25cm)$);
    \coordinate (a3u) at ($(x3)+(0, .5cm)$);
    \coordinate (a3r) at ($(x3)!-0.5!(x2)$);
    \coordinate (a3d) at ($(x3)-(0,.5cm)$);
    \coordinate (a4ur) at ($(c4x)+(\x,\y)$);
    \coordinate (a4ul) at ($(x4)!-.6!(x1)$);
    \coordinate (a4dl) at ($(c4x)-(\x,\y)$);
    \coordinate (a5u) at ($(c5xu)+(-\x,\y)$);
    \coordinate (a5r) at ($(c5xr)+(.1cm,0)$);
    \coordinate (a6ul) at ($(c6x)+(-.25,0)$);
    \coordinate (a6d) at ($(x6)!-.45!(x5)$);
    \coordinate (a6r) at ($(c6x)+(.25,0)$);
    \coordinate (a67) at ($(x6)!0.5!(x7)$);
    \coordinate (a7l) at ($(c7x)-(.25,0)$);
    \coordinate (a7d) at ($(x7)!-0.45!(x5)$);
    \coordinate (a7r) at ($(c7x)+(.25,0)$);
     
    \draw[red!50!black, thick] plot [smooth cycle] coordinates {(a3r) (a3u) (a2u) (a1u) (a4ur) (a4ul) (a4dl) (a1l) (a5u) (a6ul) (a6d) (a6r) (a67) (a7l) (a7d) (a7r) (a5r) (a1r) (a2d) (a3d)};
    \node[color=red!50!black,] at ($(x5)+(3,0)$) {\large $\mathfrak{g}\,=\,\mathcal{G}$}; 

    \coordinate (m1) at ($(x1)!0.5!(x4)$);
    \coordinate (m2) at ($(x1)!0.5!(x2)$);
    \coordinate (m3) at ($(x2)!0.5!(x3)$);
    \coordinate (m4) at ($(x1)!0.5!(x5)$);
    \coordinate (m5) at ($(x5)!0.5!(x6)$);
    \coordinate (m6) at ($(x5)!0.5!(x7)$);

    \node[very thick, cross=4pt, rotate=0, color=blue] at (m1) {};
    \node[very thick, cross=4pt, rotate=0, color=blue] at (m2) {};
    \node[very thick, cross=4pt, rotate=0, color=blue] at (m3) {};  
    \node[very thick, cross=4pt, rotate=0, color=blue] at (m4) {};  
    \node[very thick, cross=4pt, rotate=0, color=blue] at (m5) {};  
    \node[very thick, cross=4pt, rotate=0, color=blue] at (m6) {};

    \node[very thick, cross=4pt, rotate=0, color=blue] at (c1xu) {};
    \node[very thick, cross=4pt, rotate=0, color=blue] at (c1xl) {};
    \node[very thick, cross=4pt, rotate=0, color=blue] at (c1xr) {};
    \node[very thick, cross=4pt, rotate=0, color=blue] at (c2xu) {};
    \node[very thick, cross=4pt, rotate=0, color=blue] at (c2xd) {};
    \node[very thick, cross=4pt, rotate=0, color=blue] at (c3x) {};
    \node[very thick, cross=4pt, rotate=0, color=blue] at (c4x) {};
    \node[very thick, cross=4pt, rotate=0, color=blue] at (c5xu) {};
    \node[very thick, cross=4pt, rotate=0, color=blue] at (c5xr) {};
    \node[very thick, cross=4pt, rotate=0, color=blue] at (c5xd) {};
    \node[very thick, cross=4pt, rotate=0, color=blue] at (c6x) {};
    \node[very thick, cross=4pt, rotate=0, color=blue] at (c7x) {};
 
   \end{scope}
   \begin{scope}[shift={(5,-1.75)}, scale={1.5}, transform shape]
   \coordinate (v1) at (0,0);
   \coordinate (v2) at ($(v1)+(0,1.25)$);
   \coordinate (v3) at ($(v2)+(1,0)$);
   \coordinate (v4) at ($(v3)+(1,0)$);
   \coordinate (v5) at ($(v4)-(0,.625)$);   
   \coordinate (v6) at ($(v5)-(0,.625)$);
   \coordinate (v7) at ($(v6)-(1,0)$);
   \draw[thick] (v1) -- (v2) -- (v3) -- (v4) -- (v5) -- (v6) -- (v7) -- cycle;
   \draw[thick] (v3) -- (v7);
   \draw[fill=black] (v1) circle (2pt);
   \draw[fill=black] (v2) circle (2pt);
   \draw[fill=black] (v3) circle (2pt);
   \draw[fill=black] (v4) circle (2pt);
   \draw[fill=black] (v5) circle (2pt);
   \draw[fill=black] (v6) circle (2pt);
   \draw[fill=black] (v7) circle (2pt);   
   \coordinate (v12) at ($(v1)!0.5!(v2)$);   
   \coordinate (v23) at ($(v2)!0.5!(v3)$);
   \coordinate (v34) at ($(v3)!0.5!(v4)$);
   \coordinate (v45) at ($(v4)!0.5!(v5)$);   
   \coordinate (v56) at ($(v5)!0.5!(v6)$);   
   \coordinate (v67) at ($(v6)!0.5!(v7)$);
   \coordinate (v71) at ($(v7)!0.5!(v1)$);   
   \coordinate (v37) at ($(v3)!0.5!(v7)$);   
   \node[very thick, cross=4pt, rotate=0, color=blue, scale=.625] at (v34) {};
   \node[very thick, cross=4pt, rotate=0, color=blue, scale=.625] at (v45) {};
   \node[very thick, cross=4pt, rotate=0, color=blue, scale=.625, left=.15cm of v3] (v3l) {};
   \node[very thick, cross=4pt, rotate=0, color=blue, scale=.625, below=.15cm of v3] (v3b) {};   
   \node[very thick, cross=4pt, rotate=0, color=blue, scale=.625, below=.1cm of v5] (v5b){};
   \coordinate (a) at ($(v3l)!0.5!(v3)$);
   \coordinate (b) at ($(v3)+(0,.125)$);
   \coordinate (c) at ($(v34)+(0,.175)$);
   \coordinate (d) at ($(v4)+(0,.125)$);
   \coordinate (e) at ($(v4)+(.125,0)$);
   \coordinate (f) at ($(v45)+(.175,0)$);
   \coordinate (g) at ($(v5)+(.125,0)$);
   \coordinate (h) at ($(v5b)!0.5!(v5)$);
   \coordinate (i) at ($(v5)-(.125,0)$);
   \coordinate (j) at ($(v45)-(.175,0)$);
   \coordinate (k) at ($(v34)-(0,.175)$);
   \coordinate (l) at ($(v3)-(0,.125)$);
   \draw [thick, red!50!black] plot [smooth cycle] coordinates {(a) (b) (c) (d) (e) (f) (g) (h) (i) (j) (k) (l)};
   \node[below=.05cm of k, color=red!50!black] {\footnotesize $\displaystyle\mathfrak{g}$};   
  \end{scope}
 \end{tikzpicture}  
\end{equation*}
where the subgraph is encircled. Summarising, the marking in the middle of an edge of $\mathcal{G}$ indicates that the corresponding vertex $\mathbf{x}_s-\mathbf{y}_e+\mathbf{x}_{s'}$  does not belong to the face, while the marking in the extreme of the edge close to the graph site with weight $x_{s'}$ indicates that the vertex $-\mathbf{x}_s+\mathbf{y}_e+\mathbf{x}_{s'}$ does not belong to the relevant face.


\section{Projective Polytopes and Covariant Forms}\label{sec:DFHP}

Projective polytopes, as well as more generally positive geometries, are in $1-1$ correspondence with canonical forms, which are meromorphic forms with simple poles only. In this section, we show that:
\begin{enumerate}[label=(\alph*)]
    \item given a projective polytope $(\mathbb{P}^N,\,\mathcal{P})$, it is possible to associate a class of differential forms $\omega^{(k)}(\mathcal{Y},\mathcal{P})$ to it, which we call \emph{covariant forms}. These are
          meromorphic forms with poles of higher multiplicity on the boundaries of $\mathcal{P}$, and are distinguished by a $GL(1)$-scaling of degree $k$. For a fixed $GL(1)$ scaling and fixed multiplicities $m_j$'s of the poles, the covariant meromorphic form associated to a given polytope $(\mathbb{P}^N,\,\mathcal{P})$ is {\it not} unique;
    \item given a projective polytope $(\mathbb{P}^N,\,\mathcal{P})$, we define a more general way of associating differential forms with $GL(1)$ scaling to it. In 
          particular, we introduce the notion of {\it covariant pairing} $(\mathcal{P},\,\omega^{\mbox{\tiny $(k)$}})$ as the association of a differential meromorphic form with $GL(1)$ scaling of degree $k$ whose poles are along the boundary components of a signed triangulation of $(\mathbb{P}^N,\,\mathcal{P})$, including the collection of subsets of boundary components which triangulates the empty set. Moreover, the cancellation of spurious poles along such subsets, is such that that the order of the associated poles is lowered in the sum, but in general remains non-zero; 
    \item it is possible to complete the geometric-combinatorial characterisation of covariant forms and covariant pairings by relating them to higher dimensional projective
          polytopes whose restrictions onto certain hyperplanes return the polytope they are associated to. In particular, we introduce the notion {\it covariant restriction} of a canonical form of a polytope onto a given hyperplane, which maps the canonical form of the polytope into a covariant form associated to the restricion of the polytope on the hyperlpane, or into a differential form in covariant pairing with it.
\end{enumerate}


\subsection{Covariant Forms}\label{subsec:CovFrm}

Let us begin with defining a covariant form. Let $(\mathbb{P}^N,\,\mathcal{P})$ be a projective polytope with canonical form \eqref{eq:CFP} with $\mathcal{P}$ defined via the set of inequalities $\{q_j(\mathcal{Y})\, :=\,\mathcal{Y}^I\mathcal{W}_I^{\mbox{\tiny $(j)$}}\,\ge\,0,\,j\,=\,1,\ldots,\,\tilde{\nu}\}$, and let $\{m_j\in\,\mathbb{N},\,j\,=\,1,\ldots,\tilde{\nu}\}$ be a set of strictly positive integers, then a \emph{covariant form} of degree $k\,\in\,\mathbb{N}_0$ is defined as
\begin{equation}\label{eq:CFdk} 
 \omega^{\mbox{\tiny $(k)$}}(\mathcal{Y},\,\mathcal{P})\:=\:\frac{\mathfrak{n}_{\delta}(\mathcal{Y})\langle\mathcal{Y}d^N\mathcal{Y}\rangle}{\prod_{j=1}^{\tilde{\nu}}q^{m_j}_j(\mathcal{Y})},\hspace{1cm}
 \mbox{deg}\{\mathfrak{n}_{\delta}\}\: :=\:\delta,
\end{equation}
such that
\begin{enumerate}[label=(\roman*)]
 \item under the action of a $GL(1)$ transformation $\mathcal{Y}\,\longrightarrow\,\lambda\,\mathcal{Y}$, $\lambda\,\in\,\mathbb{R}_{+}$, the covariant form $\omega^{\mbox{\tiny $(k)$}}(\mathcal{\mathcal{Y},\mathcal{P}})$ transforms 
       as
       \begin{equation}\label{eq:CFdkGL1}
        \omega^{\mbox{\tiny $(k)$}}(\mathcal{\mathcal{Y},\mathcal{P}})\:\longrightarrow\:\lambda^{-k}\omega^{\mbox{\tiny $(k)$}}(\mathcal{\mathcal{Y},\mathcal{P}}),\qquad k\:\in\:\mathbb{N}_0,
       \end{equation}
       with $k$ being the {\it covariant degree} of the differential form\footnote{In most of the text, when there is no ambiguity, we will refer to the covariant degree as simply as {\it degree} of the differential form with a little abuse of terminology.}.
       Such a property fixes the degree $\delta$ of the numerator $\mathfrak{n}_{\delta}(\mathcal{Y})$ of $\omega^{\mbox{\tiny $(k)$}}$ to be $\delta\:=\:\sum_{j=1}^{\tilde{\nu}}m_j-N-1-k$. The forms of degree $k$ that differ from each other by \eqref{eq:CFdkGL1} belong to the same equivalence class: 
       $$\omega^{\mbox{\tiny $(k)$}}(\mathcal{Y},\,\mathcal{P})\,\sim\,\lambda^{-k}\omega^{\mbox{\tiny $(k)$}}(\mathcal{Y},\,\mathcal{P});$$
 \item its leading Laurent  coefficient along any of the boundary components $(\mathbb{P}^{N-1},\,\partial\mathcal{P}^{\mbox{\tiny $(j)$}})$ is a covariant form of degree $k-m_j+1$ of the polytope constituted       by the boundary component $(\mathbb{P}^{N-1},\,\partial\mathcal{P}^{\mbox{\tiny $(j)$}})$ itself:
       \begin{equation}\label{eq:LorOp}
           \Lor_{\mbox{\tiny $\partial\mathcal{P}^{\mbox{\tiny $(j)$}}$}}{\hspace{-.25cm}}^{\mbox{\tiny $(m_j)$}}\left\{\omega^{\mbox{\tiny $(k)$}}(\mathcal{Y},\,\mathcal{P})\right\}\:=\:
            \omega^{\mbox{\tiny $(k-m_j+1)$}}(\mathcal{Y}',\,\partial\mathcal{P}^{\mbox{\tiny $(j)$}}),
       \end{equation}
       where $\mathcal{Y}'\,\in\,\mathbb{P}^{N-1}$, $\mathcal{L}^{\mbox{\tiny $(m_j)$}}$ is the \emph{Laurent operator} (of order $m_j$) applied to the covariant form $\omega(\mathcal{Y},\,\mathcal{P})$ along the boundary component $(\mathbb{P}^{N-1},\,\partial\mathcal{P}^{\mbox{\tiny $(j)$}})$.
    \end{enumerate}
    
  The Laurent operator in Eq.\eqref{eq:LorOp} is defined as follows. Let us parametrise $\mathbb{P}^N$ with a set of local holomorphic coordinates ($y_j,\,h_j$) such that the locus $h_j\,=\,0$ locally identifies the facet    
       $(\mathbb{P}^{N-1},\,\partial\mathcal{P}^{\mbox{\tiny $(j)$}})$, while $y_j$ collectively indicates the remaining local coordinates. Then the covariant form $\omega^{\mbox{\tiny $(k)$}}(\mathcal{\mathcal{Y},\mathcal{P}})$ shows a multiple pole in $h_j\,=\,0$ with multiplicity $m_j$ such that
       \begin{equation}\label{eq:CFdkLoc}
        \omega^{\mbox{\tiny $(k)$}}(\mathcal{\mathcal{Y},\mathcal{P}})\:=\:\omega^{\mbox{\tiny $(k-m_j+1)$}}(y_j)\,\wedge\,\frac{dh_j}{h_j^{m_j}} + \tilde{\omega}^{\mbox{\tiny $(k)$}},
       \end{equation}
       with $\tilde{\omega}^{\mbox{\tiny $(k)$}}$ being the part of the covariant form which at most shows poles in $h_j\,=\,0$ with multiplicity lower than $m_j$, {\it i.e.} it does not contribute to the leading coefficient in the Laurent expansion around $h_j\,=\,0$, and $\omega^{\mbox{\tiny $(k-m_j+1)$}}(y_j)$ is a covariant form of degree $k-m_j+1\,\in\,\mathbb{N}_0$ with poles in any of the other local variables included in the collective one $y_j$ whose multiplicity can be \emph{lower or equal} to $m_l$ ($l\,\neq\,j$):
       \begin{equation}\label{eq:LorOp2}
        \begin{split}
           \Lor_{\partial\mathcal{P}^{\mbox{\tiny $(j)$}}}{\hspace{-.25cm}}^{\mbox{\tiny $(m_j)$}}\left\{\omega^{\mbox{\tiny $(k)$}}(\mathcal{Y},\,\mathcal{P})\right\}\:&=\:
           \Lor_{\mbox{\tiny $h_j=0$}}{\hspace{-.25cm}}^{\mbox{\tiny $(m_j)$}}\left\{\omega^{\mbox{\tiny $(k)$}}(\mathcal{Y},\,\mathcal{P})\right\}\:=\:
           \omega^{\mbox{\tiny $(k-m_j+1)$}}(y_j)\:=\\
          &=\:\omega^{\mbox{\tiny $(k-m_j+1)$}}(\mathcal{Y}',\,\partial\mathcal{P}^{\mbox{\tiny $(j)$}}),
        \end{split}
       \end{equation}
       with equalities being valid locally.

Importantly, the requirement (ii) implies the existence of an upper bound for the multiplicity $m_j$ of a given pole for fixed covariance degree-$k$: $m_j\,\in\,]0,\,k+1],\;\forall\,j\,=\,1,\,\ldots,\,\tilde{\nu}$. Furthermore the conditions (i) and (ii) {\it do not} fix univocally the covariant form $\omega^{\mbox{\tiny $(k)$}}$ for a given polytope $(\mathbb{P}^N,\,\mathcal{P})$, except for the case $k\,=\,0$.

\begin{prop}\label{thm:wk0un}
 Given a polytope $(\mathbb{P}^N,\,\mathcal{P})$, there is a unique covariant form $\omega^{(0)}(\mathcal{Y},\mathcal{P})$ of degree $k\,=\,0$, and it is given by its canonical form $\omega(\mathcal{Y},\,\mathcal{P})$.
\end{prop}

\begin{proof}
 Let us consider  the most generic form \eqref{eq:CFdk} for a covariant form of degree $k\,=\,0$. The scaling property (i) -- in this case the invariance under $GL(1)$ transformations -- fixes the degree of the homogeneous polynomial $\mathfrak{n}_{\delta}(\mathcal{Y})$ constituting the numerator to be $\delta\,=\,\sum_{j=1}^{\tilde{\nu}}m_j-N-1$:
 \begin{equation}\label{eq:CFdk0} 
  \omega^{\mbox{\tiny $(0)$}}(\mathcal{Y},\,\mathcal{P})\:=\:\frac{\mathfrak{n}_{\delta}(\mathcal{Y})\langle\mathcal{Y}d^N\mathcal{Y}\rangle}{\prod_{j=1}^{\tilde{\nu}}q^{m_j}_j(\mathcal{Y})},\hspace{1cm}
  \delta\:=\:\sum_{j=1}^{\tilde{\nu}}m_j-N-1.
\end{equation}
 Let us now parametrise $\mathbb{P}^N$ via the local holomorphic coordinates $(y_j,\,h_j)$ such that the locus $h_j\,=\,0$ identifies one of the facets of the polytope, while $y_j$ collectively indicate the other coordinates. Then, property (ii) implies that
 \begin{equation}\label{eq:CFdkLoc0}
  \omega^{\mbox{\tiny $(0)$}}(\mathcal{\mathcal{Y},\mathcal{P}})\:=\:\omega^{\mbox{\tiny $(0-m_j+1)$}}(y_j)\,\wedge\,\frac{dh_j}{h_j^{m_j}} + \tilde{\omega}^{\mbox{\tiny $(0)$}},
 \end{equation}
 However, by definition the degree of the covariant forms is non-negative and all the multiplicities $m_j$'s are strictly positive. Hence $m_j\,=\,1$. Iterating this argument for all the singularities of \eqref{eq:CFdk0}, then
 \begin{equation}\label{mj1}
  \forall\:j\,=\,1,\,  \ldots,\,\tilde{\nu}\: :\: m_j\,=\,1.
 \end{equation}
 Hence, a covariant form of degree $0$ satisfying the property (ii) has simple poles only and \eqref{eq:CFdkLoc0} reduces to \eqref{eq:CFX}, so that the residues of $\omega^{\mbox{\tiny $(0)$}}(\mathcal{\mathcal{Y},\mathcal{P}})$ along any of the facets returns a degree-$0$ form with simple poles only associated to the facets itself. Thus, the covariant form of degree-$0$ satisfying property (ii) is the canonical form associated to the polytope $(\mathbb{P}^N,\,\mathcal{P})$:
 \begin{equation}\label{eq:CFdk0can} 
  \omega^{\mbox{\tiny $(0)$}}(\mathcal{\mathcal{Y},\mathcal{P}})\,\equiv\,\omega(\mathcal{\mathcal{Y},\mathcal{P}})
 \end{equation}
\end{proof}

For $k\,>\,0$, the conditions (i) and (ii) are not sufficient to fix the covariant form with fixed degree $k$ and fixed multiplicities $m_j$ ($j\,=\,1,\ldots,\tilde{\nu}$): one could imagine it to be defined up to an overall constant only (it is a $GL(1)$ covariant form), however the defining conditions (i) and (ii) are not sufficient to fix the numerator $\mathfrak{n}_{\delta}$ up to an overall constant.
\\

\noindent
{\it Example: } Let us consider a simple visualisable example. Let us take $(\mathbb{P}^1,\,\mathcal{P})$ with $\mathcal{P}$ being a segment and let us try to fix a covariant form of degree $1$. Because of the bound $m_j\,\le\,k+1$ for the multiplicities of the poles, our degree-one covariant form can have simple and double poles only. Let us take both poles to have multiplicity two. Then the scaling condition (i) fixes the degree $\delta$ of the numerator $\mathfrak{n}_{\delta}$ to be $1$. Taking $\mathcal{Y}\,=\,(y_1,\,y_2)$  as homogeneous local coordinates, then the most generic form for such a degree-$1$ covariant form is given by
\begin{equation}\label{eq:ex1}
 \omega^{\mbox{\tiny $(1)$}}\:=\:\frac{a_1 y_1+a_2 y_2}{y_1^2 y_2^2}\frac{dy_1\wedge dy_2}{\mbox{Vol}\{GL(1)\}},\qquad a_1,\,a_2\,\in\,\mathbb{R}.
\end{equation}
Let us now check whether the condition (ii) fixes one of the coefficients $a_j$ ($j\,=\,1,\,2$)  in function of the other one. Notice that in correspondence of any of the facets $y_j\,=\,0$  ($j\,=\,1,\,2$) we get
\begin{equation}\label{eq:ex1ii}
 \omega^{\mbox{\tiny $(1)$}}\:=\:\underbrace{\frac{(-1)^j}{\mbox{Vol}\{GL(1)\}}\,\frac{dy_l}{y_l}}_{\omega^{\mbox{\tiny $(0)$}}(y_l)}\wedge\,a_l\frac{dy_j}{y_j^2}\:+\:\ldots,\hspace{1.5cm}
 j,\,l\:=\:1,\,2\quad(l\,\neq\,j).
\end{equation}
Taking the patch $y_l\,=\,1$, the covariant form \eqref{eq:ex1ii} acquires the form
\begin{equation}\label{eq:ex1iii}
    \omega^{\mbox{\tiny $(1)$}}\:=\:(-1)^j\,a_l\frac{dy_j}{y_j^2}\,+\,\ldots
\end{equation}
and the leading Laurent coefficient of this double pole is an arbitrary constant.

Notice that the constant $a_l$ is not fixed by requiring that the canonical form $\omega^{\mbox{\tiny $(0)$}}$ is $\pm\,1$ because it reflects the covariance degree of the differential $dy_j/y_j^2$\footnote{This is the important point which marks the difference between the cases $k\,=\,0$, for which the uniqueness theorem \ref{thm:wk0un} holds, and $k\,>\,0$: while in the first case all the poles are forced to be simple so that requiring the leading singularity to be $\pm1$ fixes the form completely, in the case of covariant forms we still have an equivalence class of forms because of \eqref{eq:ex1ii} albeit along a boundary degree-$0$ form is singled out.}. Hence, the leading part of the covariant form along a given boundary is still defined up to a constant. Consequently, the expression \eqref{eq:ex1ii} does not fix $a_1$ and $a_2$ to be proportional to each other, neither a $GL(1)$ transformation does it, given that it can allow to fix one of the two to one, but leaving the other arbitrary.
\\

 Differently from the canonical forms, that are in principle defined up to an overall constant $a$ which can be fixed by requiring that the leading singularities are $\pm1$ rather than $\pm a$, the covariant forms of degree $k$ define equivalence classes as a consequence of the property (i).


\subsection{Unions, Triangulations and Covariant Pairings}\label{subsec:UTCP}

Let $(\mathbb{P}^N,\,\mathcal{P}^{\mbox{\tiny $(1)$}}\cup\mathcal{P}^{\mbox{\tiny $(2)$}})$ the disjoint union of two projective polytopes $(\mathbb{P}^N\,\mathcal{P}^{\mbox{\tiny $(1)$}})$ and $(\mathbb{P}^N\,\mathcal{P}^{\mbox{\tiny $(2)$}})$. Then, the equivalence class of covariant forms of degree $k$ associated to such a disjoint union is defined by the sum of any representative of the covariant forms of each element of the union:
\begin{equation}\label{eq:cfwdu}
    \omega^{\mbox{\tiny $(k)$}}(\mathcal{Y},\,\mathcal{P}^{\mbox{\tiny $(1)$}}\cup\mathcal{P}^{\mbox{\tiny $(2)$}})\:=\:
        \omega^{\mbox{\tiny $(k)$}}(\mathcal{Y},\,\mathcal{P}^{\mbox{\tiny $(1)$}})+\omega^{\mbox{\tiny $(k)$}}(\mathcal{Y},\,\mathcal{P}^{\mbox{\tiny $(2)$}}).
\end{equation}
Because the boundaries of $(\mathbb{P}^N,\,\mathcal{P}^{\mbox{\tiny $(1)$}}\cup\mathcal{P}^{\mbox{\tiny $(2)$}})$ are either boundaries of one of the $(\mathbb{P}^N\,\mathcal{P}^{\mbox{\tiny $(j)$}})$'s or the union of their boundaries, the property (ii) is guaranteed for  $\omega^{\mbox{\tiny $(k)$}}(\mathcal{Y},\,\mathcal{P}^{\mbox{\tiny $(1)$}}\cup\mathcal{P}^{\mbox{\tiny $(2)$}})$. Equation \eqref{eq:cfwdu} is just the statement that the sum of {\it any} representative of the covariant forms of degree $k$ for $(\mathbb{P}^N\,\mathcal{P}^{\mbox{\tiny $(1)$}})$ and $(\mathbb{P}^N\,\mathcal{P}^{\mbox{\tiny $(2)$}})$ such that $\mathcal{P}^{\mbox{\tiny $(1)$}}\cap\mathcal{P}^{\mbox{\tiny $(2)$}}\,=\,\varnothing$ returns a representative of the covariant forms for their disjoint union $(\mathbb{P}^N,\,\mathcal{P}^{\mbox{\tiny $(1)$}}\cup\mathcal{P}^{\mbox{\tiny $(2)$}})$. 
\\

We would like now to generalise the notion of (signed) triangulation reviewed in Section \ref{subsubsec:DuTr} to the case of covariant forms. Recall that, given a polytope $(\mathbb{P}^N,\,\mathcal{P})$ and a collection of polytopes $\lbrace (\mathbb{P}^N,\,\mathcal{P}^{(j)})\rbrace$ which sign-triangulates it, then the canonical form of $(\mathbb{P}^N,\,\mathcal{P})$ can be expressed as a sum of the canonical forms of the elements of the collection $\lbrace (\mathbb{P}^N,\,\mathcal{P}^{(j)})\rbrace$:
\begin{equation} \label{eq:triangex}
    \omega(\mathcal{Y},\,\mathcal{P})=\sum_{j} \omega(\mathcal{Y},\,\mathcal{P}^{(j)}).
\end{equation}
In particular, for any collection $\lbrace \mathcal{Q}^{(i)} \rbrace$ of faces of some of the polytopes $\lbrace (\mathbb{P}^N,\,\mathcal{P}^{(j))}\rbrace$ such they triangulate the empty set, the canonical form $\omega(\mathcal{Y},\,\mathcal{P})$ does not have poles on them. Therefore, the simple poles $\lbrace \omega(\mathcal{Y},\,\mathcal{P}^{(i)}) \rbrace$ have on $\lbrace \mathcal{Q}^{(i)} \rbrace$ all cancel in the sum \eqref{eq:triangex}: they are called \emph{spurious}.
In the case of covariant forms, they have in general poles of higher multiplicity and the poles related to those faces which triangulate the empty set might no longer be spurious, but their order could be lowered.

Let $(\mathbb{P}^N,\,\mathcal{P})$ be a projective polytope, $\{(\mathbb{P}^N,\,\mathcal{P}^{\mbox{\tiny $(j)$}})\}_{j=1}^{n}$ a collection of projective polytopes and $\omega^{(k)}$ a differential form of covariant degree $k$ such that:
\begin{enumerate}[label=(\roman*)]
    \item $\lbrace (\mathbb{P}^N,\,\mathcal{P}^{(j)}) \rbrace$ is a signed triangulation of $(\mathbb{P}^N,\,\mathcal{P})$;
    \item the form $\omega^{\mbox{\tiny $(k)$}}$  can be written as a sum of (representatives of) covariant forms $\omega^{(k)}(\mathcal{Y},\,\mathcal{P}^{(j)})$ of 
          degree $k$ associated to the projective polytopes $(\mathbb{P}^N,\,\mathcal{P}^{\mbox{\tiny $(j)$}})$:
          \begin{equation}\label{eq:cfwcp}
           \omega^{(k)}\:=\:\sum_{j=1}^n\omega^{(k)}(\mathcal{Y},\,\mathcal{P}^{(j)});
           \end{equation}
    \item for every collection of faces $\lbrace \mathcal{Q}^{(i)} \rbrace$ of some of the polytopes $\lbrace (\mathbb{P}^N,\,\mathcal{P}^{(j)})\rbrace$ such they triangulate the empty
          set, the order of spurious poles $\lbrace \omega(\mathcal{Y},\,\mathcal{P}^{(i)}) \rbrace$ have on $\lbrace \mathcal{Q}^{(i)} \rbrace$ are lowered in the sum;
\end{enumerate}
then the association $(\mathcal{P}, \omega^{(k)})$ is called a \emph{covariant pairing}. Moreover, $\lbrace (\mathcal{P}^{(j)},\omega^{(k)}(\mathcal{P}^{(j)})) \rbrace$ will be referred to as a \emph{covariant triangulation} of $(\mathcal{P},\omega^{(k)})$.

As the collection $\{(\mathbb{P}^N,\,\mathcal{P}^{\mbox{\tiny $(j)$}})\}_{j=1}^{n}$ provides a signed triangulation for $(\mathbb{P}^N,\,\mathcal{P})$, there exist a common pole in a subset of the collection of covariant forms $\{\omega^{\mbox{\tiny $(k)$}}(\mathcal{Y},\,\mathcal{P}^{\mbox{\tiny $(j)$}})\}_{j=1}^n$ such that the boundary components of the relevant elements of $\{(\mathbb{P}^N,\,\mathcal{P}^{\mbox{\tiny $(j)$}})\}_{j=1}^{n}$ triangulate the empty set. If such class of poles have multiplicity higher than $1$, then the covariant form $\omega^{\mbox{\tiny $(k)$}}$ in covariant pairing $(\mathcal{P},\,\omega^{\mbox{\tiny $(k)$}}(\mathcal{Y},\,\mathcal{P}))$ with $(\mathbb{P}^N,\,\mathcal{P})$ shows a pole of lower multiplicity:  the covariant form $\omega^{\mbox{\tiny $(k)$}}$ has poles in correspondence of both the boundary components of $(\mathbb{P}^N,\,\mathcal{P})$ and of the boundary components of $(\mathbb{P}^N,\,\mathcal{P}^{\mbox{\tiny $(j)$}})$ which are not boundaries of $(\mathbb{P}^N,\,\mathcal{P})$. If instead such a pole is a simple, it becomes spurious upon the summation \eqref{eq:cfwcp} and we recover the covariant form has only poles along the boundary components of $(\mathbb{P}^N,\,\mathcal{P})$.

Hence, the covariant pairing generalises the association between a covariant form and a projective polytope originally defined in Section \ref{subsec:CovFrm}. With a little abuse of notation, in what follows we will indicate with $\omega^{\mbox{\tiny $(k)$}}(\mathcal{Y},\,\mathcal{P})$ (a representative of) a covariant form with poles only along the boundary components of $(\mathbb{P}^N,\,\mathcal{P})$, as well as a covariant form in covariant pairing with the projective polytope $(\mathbb{P}^N,\,\mathcal{P})$, which have (multiple) poles both along the boundary components of $(\mathbb{P}^N,\,\mathcal{P})$ and along the empty-set-triangulating boundary components of a collection of projective polytopes providing a signed triangulation of $(\mathbb{P}^N,\,\mathcal{P})$. 
\\

\noindent
{\it Example}: Let us consider two segments $(\mathbb{P}^1,\,\mathcal{P}^{\mbox{\tiny $(32)$}})$ and $(\mathbb{P}^1,\,\mathcal{P}^{\mbox{\tiny $(31)$}})$ such that they provide a signed triangulation of $(\mathbb{P}^1,\,\mathcal{P}^{\mbox{\tiny $(12)$}})$ \footnote{The apex $(ij)$ in $\mathcal{P}^{\mbox{\tiny $(ij)$}}$ indicates that the vertices of that segment are $i$ and $j$.}. 

\begin{wrapfigure}{l}{4.5cm}
 \centering
 \begin{tikzpicture}[line join = round, line cap = round, ball/.style = {circle, draw, align=center, anchor=north, inner sep=0}, 
                     axis/.style={very thick, ->, >=stealth'}, pile/.style={thick, ->, >=stealth', shorten <=2pt, shorten>=2pt}, every node/.style={color=black}]
  \coordinate [label=below:{\footnotesize $\displaystyle \mathbf{3}$}] (Z3) at (0,0);
  \coordinate [label=below:{\footnotesize $\displaystyle \mathbf{2}$}] (Z2) at (3,0);
  \coordinate [label=below:{\footnotesize $\displaystyle \mathbf{1}$}] (Z1) at ($(Z3)!0.625!(Z2)$);
  \draw[-,thick] ($(Z3)!-.5cm!(Z2)$) -- ($(Z2)!-.5cm!(Z3)$);
  \draw[-,color=red, very thick] (Z3) -- (Z1);
  \draw[-,color=blue, very thick] (Z1) -- (Z2);
  \draw[draw=none, fill=red] (Z3) circle (2pt);
  \draw[draw=none, fill=blue] (Z2) circle (2pt);
  \draw[draw=none, fill=blue] (Z1) circle (2pt);
 \end{tikzpicture}
\end{wrapfigure}

Let $(Z_3,\,Z_1)$, $(Z_3,\,Z_2)$ and $(Z_1,\,Z_2)$ be the pair of boundary components of $(\mathbb{P}^1,\,\mathcal{P}^{\mbox{\tiny $(32)$}})$,  $(\mathbb{P}^1,\,\mathcal{P}^{\mbox{\tiny $(31)$}})$ and $(\mathbb{P}^1,\,\mathcal{P}^{\mbox{\tiny $(12)$}})$ respectively, with $Z_j\,\in\,\mathbb{P}^1$. Let us consider a covariant form of degree $1$ for with a double pole along one of the boundary components of $(\mathbb{P}^1,\,\mathcal{P}^{\mbox{\tiny $(3j)$}})$'s, namely:
\begin{equation}\label{eq:cpex}
    \omega^{\mbox{\tiny $(1)$}}(\mathcal{Y},\,\mathcal{P}^{\mbox{\tiny $(32)$}})\:=\:\frac{\langle23\rangle\langle\mathcal{Y}d\mathcal{Y}\rangle}{\langle\mathcal{Y}3\rangle^2\langle\mathcal{Y}2\rangle},
    \qquad
    \omega^{\mbox{\tiny $(1)$}}(\mathcal{Y},\,\mathcal{P}^{\mbox{\tiny $(31)$}})\:=\:\frac{\langle31\rangle\langle\mathcal{Y}d\mathcal{Y}\rangle}{\langle\mathcal{Y}3\rangle^2\langle\mathcal{Y}1\rangle}.
\end{equation}
Then the covariant form $\omega^{\mbox{\tiny $(k)$}}$ in covariant pairing with $(\mathbb{P}^1,\,\mathcal{P})$ is:
\begin{equation}\label{eq:cpex2}
    \omega^{\mbox{\tiny $(1)$}}(\mathcal{Y},\,\mathcal{P}^{\mbox{\tiny $(12)$}})\:=\:\sum_{j=1}^2\omega^{\mbox{\tiny $(1)$}}(\mathcal{Y},\,\mathcal{P}^{\mbox{\tiny $(3j)$}})\:=\:
    \frac{\langle21\rangle\langle\mathcal{Y}d\mathcal{Y}\rangle}{\langle\mathcal{Y}1\rangle\langle\mathcal{Y}2\rangle\langle\mathcal{Y}3\rangle},
\end{equation}
which shows a pole in each boundary component of the collection $\{(\mathbb{P}^1,\,\mathcal{P}^{\mbox{\tiny $(3j)$}})\}_{j=1}^2$, with the pole along the common boundary of lower multiplicity. 


We can also consider the following covariant forms associated to $\{(\mathbb{P}^1,\,\mathcal{P}^{\mbox{\tiny $(3j)$}})\}_{j=1}^2$
\begin{equation}\label{eq:cpex3}
    \omega^{\mbox{\tiny $(1)$}}(\mathcal{Y},\,\mathcal{P}^{\mbox{\tiny $(32)$}})\:=\:\frac{\langle23\rangle^2\langle\mathcal{Y}d\mathcal{Y}\rangle}{\langle\mathcal{Y}3\rangle\langle\mathcal{Y}2\rangle^2},
    \qquad
    \omega^{\mbox{\tiny $(1)$}}(\mathcal{Y},\,\mathcal{P}^{\mbox{\tiny $(31)$}})\:=-\:\frac{\langle31\rangle^2\langle\mathcal{Y}d\mathcal{Y}\rangle}{\langle\mathcal{Y}3\rangle\langle\mathcal{Y}1\rangle^2}.
\end{equation}
Then the covariant form $\omega^{\mbox{\tiny $(k)$}}$ in covariant pairing with $(\mathbb{P}^1,\,\mathcal{P}^{\mbox{\tiny $(12)$}})$ is:
\begin{equation}\label{eq:cpex4}
    \omega^{\mbox{\tiny $(1)$}}(\mathcal{Y},\,\mathcal{P}^{\mbox{\tiny $(12)$}})\:=\:
     \sum_{j=1}^2\omega^{\mbox{\tiny $(1)$}}(\mathcal{Y},\,\mathcal{P}^{\mbox{\tiny $(3j)$}})\:=\:
     \frac{\langle21\rangle\left(\langle32\rangle\langle\mathcal{Y}1\rangle+\langle31\rangle\langle\mathcal{Y}2\rangle\right)\langle\mathcal{Y}d\mathcal{Y}\rangle}{\langle\mathcal{Y}1\rangle^2\langle\mathcal{Y}2\rangle^2}.
\end{equation}
Notice the this covariant form has poles only along the boundary components of $(\mathbb{P}^1,\,\mathcal{P}^{\mbox{\tiny $(12)$}})$ and it is one of the covariant forms of degree-$1$ naturally associated to the segment $(\mathbb{P}^1,\,\mathcal{P}^{\mbox{\tiny $(12)$}})$.

Summarising, in the previous two subsection we have introduced a natural way of associating the subclass of differential forms with non-logarithmic singularities constituted by forms whose coefficients are meromorphic homogeneous functions, to projective polytopes, through the notions of {\it covariant forms} and differential forms in {\it covariant pairing} with polytopes, with the latter generalising the former. Neither covariant forms nor covariant pairings are in $1-1$ correspondence with a polytope, not even fixing the multiplicity of the poles in the covariant form: the defining conditions for the covariant forms constrain as well as the requirement that spurious higher codimension singularities  cancel, constrain the numerator of the covariant forms but they do not fix it uniquely.

In the next subsection we will see how it is possible to complete the geometric-combinatorial characterisation of covariant forms and covariant pairings by associating them to higher dimensional polytopes and their canonical forms.


\subsection{Parent Polytopes, Child Polytopes and Covariant Forms}\label{subsec:PCP}

Let $(\mathbb{P}^N,\,\mathcal{P})$ be a projective polytope and let $\mathcal{F}_{\mathcal{P}}\, := \,\{\mathcal{W}_I^{\mbox{\tiny $(j)$}}\,\in\,\mathbb{P}^N(\mathbb{R}),\;j\,=\,1,\ldots,\tilde{\nu}\}$ be the set of dual vectors identifying its facets. Let $\mathcal{H}\, :=\,\{\mathcal{Y}\,\in\,\mathbb{P}^{N}(\mathbb{R})\,|\,h_{l}(\mathcal{Y})\, :=\, \mathcal{Y}^I\mathfrak{H}_I^{\mbox{\tiny $(l)$}}\,=\,0,\;\mathfrak{H}_I^{\mbox{\tiny $(l)$}}\,\nsubseteq\,\mathcal{F}_{\mathcal{P}},\;\forall\,l\,=\,1,\,\ldots,\,N-M\}$ be an hyperplane of codimension $N-M$ in $\mathbb{P}^N$ -- {\it i.e.} it lives in $\mathbb{P}^{M}\,\subset\,\mathbb{P}^N$, with $M\,<\,N$ -- such that it intersects the convex hull $\mathcal{P}$. Let $\mathcal{P}_{\mathcal{H}}\, :=\,\mathcal{P}\,\cap\,\mathcal{H}$ be the restriction\footnote{The term `restriction' is just equivalent to \emph{section} of the polytope, on the geometric side. In our case, it will also carry extra information about an operation on differential forms, as in \eqref{eq:CPcf}.} of $\mathcal{P}$ on $\mathcal{H}$. We will refer to the projective polytope $(\mathbb{P}^N,\,\mathcal{P})$ as {\it parent polytope}, and to its restriction $(\mathbb{P}^M,\,\mathcal{P}_{\mbox{\tiny $\mathcal{H}$}})$ on the hyperplane $\mathcal{H}$ as its {\it child polytope} with respect to $\mathcal{H}$.
\\

If $\omega(\mathcal{Y},\,\mathcal{P})$ is the canonical form associated to $(\mathbb{P}^N,\,\mathcal{P})$, then it is possible to define the {\it covariant restriction} of $\omega(\mathcal{Y},\,\mathcal{P})$ onto $\mathcal{H}$ as the  differential form
\begin{equation}\label{eq:CPcf}
 \omega^{\mbox{\tiny $(N-M)$}}(\mathcal{Y}_{\mathcal{H}})\: :=\: \frac{1}{(2\pi i)^{N-M}}\oint_{\mathcal{H}}\,  
  \frac{\omega(\mathcal{Y},\,\mathcal{P})}{\prod_{l=1}^{\mbox{\tiny $N-M$}}h_l(\mathcal{Y})}.
\end{equation}
The differential form \eqref{eq:CPcf} can be equivalently defined as
\begin{equation}\label{eq:LorOpH}
    \omega^{\mbox{\tiny $(N-M)$}}(\mathcal{Y}_{\mathcal{H}})\: :=\:
     \Lor_{\mbox{\tiny $\mathcal{H}$}}{}^{\mbox{\tiny $(0)$}}\left\{\omega(\mathcal{Y},\,\mathcal{P})\right\},
\end{equation}
where $\mathcal{L}^{\mbox{\tiny $(0)$}}$ is the Laurent operator defined in \eqref{eq:CFdkLoc} but now acting along a codimension $N-M$ hyperplane and extracting the zero-th order coefficient. More explicitly, let us parametrise $\mathbb{P}^N$ with a set of local holomorphic coordinates $(y,\,h)$, where $h\,:=\:\{h_1,\,\ldots,\,h_{\mbox{\tiny $N-M$}}\}$ collectively indicates the coordinates such that the locus $h\,=\,0$ locally identifies the hyperplane $\mathcal{H}\,\subset\,\mathbb{P}^N$, while $y$ collectively indicates the remaining local coordinates. Then, the canonical form $\omega(\mathcal{Y},\,\mathcal{P})$ can be written as
\begin{equation}\label{eq:LorOpLoc}
    \omega(\mathcal{Y},\,\mathcal{P})\:=\:\omega^{\mbox{\tiny $(N-M)$}}(y)\wedge dh \:+\: \tilde{\omega},
\end{equation}
with $\tilde{\omega}$ being the part of the canonical form which depends polynomially on $h$ (with degree equal or greater than $1$), and which does not contribute to the leading Laurent coefficient of the canonical form, which is now of order zero because the locus $h\,=\,0$ does not identify neither poles nor zeroes of the canonical form. Hence, locally:
\begin{equation}\label{eq:LorOpLoc2}
 \begin{split}
    \mathcal{L}_{\mbox{\tiny $\mathcal{H}$}}{}^{\mbox{\tiny $(0)$}} \left\{\omega(\mathcal{Y},\,\mathcal{P})\right\}\: 
     &=\:\mathcal{L}_{\mbox{\tiny $h=0$}}{}^{\mbox{\tiny $(0)$}}\left\{\omega(\mathcal{Y},\,\mathcal{P})\right\}\:=\:
     \omega^{\mbox{\tiny $(N-M)$}}(y)\:=\:\omega^{\mbox{\tiny $(N-M)$}}(\mathcal{Y}_{\mathcal{H}}).
 \end{split}
\end{equation}

Notice that $\omega^{\mbox{\tiny $(N-M)$}}(\mathcal{Y}_{\mathcal{H}})$ is a differential form of covariant degree $N-M$. This property is manifest in both \eqref{eq:CPcf} and \eqref{eq:LorOpH}: the canonical form $\omega(\mathcal{Y},\,\mathcal{P})$ of $(\mathbb{P}^N,\,\mathcal{P})$ is invariant under a $GL(1)$-transformation $\mathcal{Y}\,\longrightarrow\,\lambda\,\mathcal{Y}$ (with $\lambda\,\in\,\mathbb{R}_+$), while each homogeneous polynomial $h_l(\mathcal{Y})$ in the definition of the hypersurface $\mathcal{H}$ transform as $h_l(\mathcal{Y})\,\longrightarrow\,\lambda h_l(\mathcal{Y})$ being linear. Hence the {\it integrand} differential form in \eqref{eq:CPcf} transforms as $\lambda^{-(N-M)}$. Finally, the contour integration computes the residue of the {\it integrand} differential form at all the {\it simple} poles $h_l(\mathcal{Y})\,=\,0$, leaving the $GL(1)$-scaling behaviour unchanged.
 
 Properties of the covariant restriction $\omega^{\mbox{\tiny $(N-M)$}}$ are inherited from the the property of the canonical  $\omega(\mathcal{Y},\,\mathcal{P})$  associated to $(\mathbb{P}^N,\,\mathcal{P})$ that its residue along any of the boundary components $(\mathbb{P}^{N-1},\,\partial\mathcal{P}^{\mbox{\tiny $(j)$}})$ is the canonical form of the projective polytope $(\mathbb{P}^{N-1},\,\partial\mathcal{P}^{\mbox{\tiny $(j)$}})$ itself. 
 First, let $\mathcal{W}^{\mbox{\tiny $(j_1\ldots j_{m_j})$}}\, :=\, \bigcap_{r=1}^{m_j}\mathcal{W}^{(j_r)}$ be the intersection of $m_j$ facets, each of which is identified by a dual vector $\mathcal{W}^{\mbox{\tiny $(j_r)$}}$. If $\mathcal{H}\,\bigcap\,\mathcal{W}^{\mbox{\tiny $(j_1\ldots j_{m_j})$}}\,\neq\,\varnothing$, then the linear homogeneous polynomials $q_{j_r}(\mathcal{Y})\,=\,\mathcal{Y}\cdot\mathcal{W}^{\mbox{\tiny $j_r$}}$ ($r\,=\,1,\ldots,\,m_j$) providing a subset of poles of the canonical form $\omega(\mathcal{Y,\,\mathcal{P}})$ become equal to each other on the {\it covariant restriction} on the hypersurface $\mathcal{H}$ -- {\it i.e.} when the residues of the integrand \eqref{eq:CPcf} at all the poles $h_l(\mathcal{Y})\,=\,0$  are taken --, generating a multiple pole of multiplicity $m_j$. 
 Let us now parametrise $\mathbb{P}^N$ via a set of local holomorphic coordinates $(y_j,\,q_j)$ such that the locus $q_j\,=\,0$ locally identifies a particular boundary $\partial \mathcal{P}^{(i)}$, with $y_j$ collectively indicating the remaining local coordinates. As we already saw in \eqref{eq:CFX}, it shows a simple pole in $q_j\,=\,0$ and it can be locally written as
 \begin{equation}\label{eq:CFXb}
  \omega(\mathcal{Y,\,\mathcal{P}})\:=\:\omega(y_j)\wedge\frac{dq_j}{q_j}+\tilde{\omega}.
 \end{equation}
 Considering now \eqref{eq:CPcf}, the covariant restriction of the canonical form $\omega(\mathcal{Y},\,\mathcal{P})$ generates multiple poles and, hence, in the local holomorphic coordinates $(y_j,\,q_j)$, the differential form \eqref{eq:CPcf} can be written as
  \begin{equation}\label{eq:CPloc}
  \omega^{\mbox{\tiny $(N-M)$}}(\mathcal{Y}_{\mathcal{H}})\:=\:\omega^{\mbox{\tiny $(N-M-m_j)$}}(y_j)\wedge\frac{dq_j}{q_j^{m_j}}+\tilde{\omega}^{\mbox{\tiny $(N-M)$}},
 \end{equation}
which is exactly the very same structure as \eqref{eq:CFdkLoc0}, with $\tilde{\omega}^{\mbox{\tiny $(N-M)$}}$ having a lower order pole it $q_j\,=\,0$.
Hence, the differential form satisfies also the property \eqref{eq:CFdkLoc} in the definition of the covariant forms. 
\\


Let us now analyse the structure of these covariant restriction in detail as well as the covariant forms obtained from the canonical form of $(\mathbb{P}^N,\,\mathcal{P})$. As we will discuss in detail later on, the Laurent coefficients of the covariant form $\omega^{\mbox{\tiny $(N-M)$}}(\mathcal{Y}_{\mathcal{H}})$ are related to the residues of the canonical form $\omega(\mathcal{Y},\,\mathcal{P})$ (a manifestation of this fact was first observed in the context of the cosmological polytopes \cite{Benincasa:2019vqr}), which is a consequence of the property \eqref{eq:CPloc}. Interestingly, as we will prove shortly afterwards, the covariant form $\omega^{\mbox{\tiny $(N-M)$}}(\mathcal{Y}_{\mathcal{H}})$ turns out to be in covariant pairing with the child polytope $(\mathbb{P}^M,\,\mathcal{P}_{\mathcal{H}})$, with poles reflecting boundaries both inside and outside $\mathcal{P}_{\mathcal{H}}$, which occurs when the intersections between $\mathcal{H}$ and the facets of $\mathcal{P}$ lie outside of $\mathcal{P}_{\mathcal{H}}$, or just poles along the boundary components of $(\mathbb{P}^M,\,\mathcal{P}_{\mathcal{H}})$ which occur when the intersections between $\mathcal{H}$ and the facets of $\mathcal{P}$ are boundaries of $\mathcal{P}_{\mathcal{H}}$. From \eqref{eq:CFXb} and \eqref{eq:CPloc}, it is possible to see that in general the multiplicity of the poles of the covariant form $\omega^{\mbox{\tiny $(N-M)$}}(\mathcal{Y_{\mathcal{H}}})$ is given by the number of facets of the parent polytope which have a common intersection inside the polytope and on the hyperplane $\mathcal{H}$.
There are two exceptions. The first one is when the subspace where the facets of the parent polytope and $\mathcal{H}$ intersect {\it is not} on the hypersurface which determines the zeroes of the canonical form of the parent polytope itself. In this latter case, the multiplicity of the pole is lower. The second exception occurs  when the number of facets on the common intersection with the hyperplane $\mathcal{H}$ is higher than the codimension of such intersection. Because of the properties of the canonical form of the parent polytope, if the child polytope has dimension $M$, then its poles with order great than one are on faces of dimension $M-1$ of the parent polytope. Therefore, the maximal order of these poles equals $N-(M-1)=k+1$, where $k$ is the covariant degree of $\omega^{(N-M)}(\mathcal{Y}_{\mathcal{H}})$, consistently with what discussed in Section \ref{subsec:CovFrm}.

In order to prove that statement that the covariant form $\omega(\mathcal{Y}_{\mathcal{H}})$ is in covariant pairing with $(\mathbb{P}^M,\,\mathcal{P}_{\mathcal{H}})$, let us first consider the case of simplices as parent polytopes and then generalise to arbitrary projective polytopes.

\begin{lemma}\label{lem:triangulation}
 Let $(\mathbb{P}^N, \Delta)$ be a simplex and $\omega(\mathcal{Y},\Delta)$ its canonical form. Given an hyperplane $\mathcal{H}$ of codimension $N-M$ in $\mathbb{P}^N$, let $\omega^{\mbox{\tiny $(N-M)$}}(\mathcal{Y}_{\mathcal{H}})$ be the covariant restriction of $\omega(\mathcal{Y},\Delta)$ onto $\mathcal{H}$, and $\Delta_{\mathcal{H}}\,:=\,\Delta\cap\mathcal{H}$ so that $(\mathbb{P}^M,\,\Delta_{\mathcal{H}})$ is the restriction of $(\mathbb{P}^N, \Delta)$ onto $\mathcal{H}$. Then $(\Delta_{\mathcal{H}},\,\omega^{\mbox{\tiny $(N-M)$}}(\mathcal{Y}_{\mathcal{H}}))$ is a covariant pairing. In particular, 
 there exist a collection of simplices $\lbrace (\mathbb{P}^{M},\Delta^{(\sigma)}_{\mathcal{H}})\rbrace$ which is a signed triangulation of $(\mathbb{P}^{M}, \Delta_{\mathcal{H}})$ and 
 \begin{equation}
     \omega^{\mbox{\tiny $(N-M)$}}(\mathcal{Y}_{\mathcal{H}})\:\equiv\:\omega^{\mbox{\tiny $(N-M)$}}(\mathcal{Y}_{\mathcal{H}},\Delta_{\mathcal{H}})\:=\:\sum_{\sigma}  \omega^{\mbox{\tiny $(N-M)$}}(\mathcal{Y}_{\mathcal{H}},\Delta^{(\sigma)}_{\mathcal{H}}),
 \end{equation}
 where $\omega^{\mbox{\tiny $(N-M)$}}(\mathcal{Y}_{\mathcal{H}},\Delta^{(\sigma)}_{\mathcal{H}})$ are covariant forms of degree $N-M$ associate to $(\mathbb{P}^{M},\Delta^{(\sigma)}_{\mathcal{H}})$.
\end{lemma}
\begin{proof}
 Let $(\mathbb{P}^N,\,\Delta)$ be a simplex and let $\mathcal{F}_{\Delta}\,:=\,\lbrace\mathcal{W}_I^{\mbox{\tiny $(j)$}}\,\in\,\mathbb{P}^N(\mathbb{R}),\:j\,=\,1,\ldots,\,N+1\rbrace$ be the set of dual vectors identifying its facets $F_j\,:=\,\lbrace\mathcal{Y}\,\in\,\mathbb{P}^N(\mathbb{R})\,|\,F_j(\mathcal{Y})\,:=\,\mathcal{Y}^I\mathcal{W}^{\mbox{\tiny $(j)$}}_I\,=\,0\rbrace$.
 
 Let $Z_1, \ldots Z_{N+1}$ be the vertices of $\Delta$, with $Z_i$ being the only vertex which does not belong to $F_i$. Let us denote consider the $M-1$ dimensional intersection $F_{\mbox{\tiny $N+1$}}\cap\mathcal{H}$ lies outside $\mathcal{P}$. 
 
 Let us now consider the $M$ dimensional hyperplane\footnote{Without loss of generality $\mathcal{H}$ does not pass through $Z_{N+1}$, otherwise we choose another facet.} $B=\cap_{a \in [N-M]}B_a$, with $B_a=\lbrace\mathcal{Y}\,\in\,\mathbb{P}^N(\mathbb{R})\,|\,\mathcal{Y}^I\mathcal{X}^{\mbox{\tiny $(a)$}}_I=0 \rbrace$, which includes the $M-1$ hyperplane $F_{\mbox{\tiny $N+1$}} \cap \mathcal{H}$ and the vertex $Z_{\mbox{\tiny $N+1$}}$. 
 
 Furthermore, the linear space of vectors dual to hyperplanes passing by $Z_{\mbox{\tiny $N+1$}}$ is $N$ dimensional and that $\lbrace \mathcal{W}^{\mbox{\tiny $(1)$}}\ldots \mathcal{W}^{\mbox{\tiny $(N)$}} \rbrace$ provides a basis for such a space. Therefore, since $Z_{\mbox{\tiny $N+1$}} \subset B_a$, then
 \begin{equation} \label{rel:dualvectors}
 \mathcal{X}^{\mbox{\tiny $(a)$}}=\sum_{i=1}^{N} c_{ai} \mathcal{W}^{\mbox{\tiny $(i)$}}, \quad a \in [N-M].
 \end{equation}
 
 Let us now consider the canonical form of the simplex and re-write it as:
 \begin{equation} \label{rewrite}
    \omega(\mathcal{Y},\Delta)\:\sim\:\frac{\displaystyle\prod_{a=1}^{N-M} B_a \, \langle \mathcal{Y} \mbox{d}^N \mathcal{Y} \rangle}{
      \displaystyle\prod_{a=1}^{N-M} B_a  \prod_{i=1}^{N+1} F_i}\:=\:
     \sum_{i_1,\ldots, i_{N-M} \in [N]} \frac{\displaystyle c_{1 i_1} \ldots c_{N-M i_{N-M}}}{\displaystyle \prod_{a=1}^{N-M} B_a F_{N+1}\prod_{s=1}^{M} F_{\bar{i}_s}} 
      \langle \mathcal{Y} \mbox{d}^N \mathcal{Y} \rangle
 \end{equation}
 where we used \eqref{rel:dualvectors} and we denoted $\lbrace \bar{i}_1, \ldots, \bar{i}_M \rbrace=[N]\setminus \lbrace i_1,\ldots, i_{N-M} \rbrace$. If we denote as $\Delta^{(\sigma)}$ the simplices whose facets are $B_1,\ldots, B_{\mbox{\tiny $N-M$}}, F_{\mbox{\tiny $N+1$}}, F_{\sigma_1}, \ldots, F_{\sigma_M}$, with $\sigma \in {[N] \choose M}$, then one can show that \eqref{rewrite} produces the oriented triangulation of $\Delta$ into $\lbrace \Delta^{(\sigma)} \rbrace$, in particular:
 \begin{equation} \label{cansimplexdec}
 \omega(\mathcal{Y},\Delta)\:\sim\:\sum_{\sigma \in {[N] \choose M}} \omega(\mathcal{Y},\Delta^{(\sigma)}).
 \end{equation}
 We now consider the covariant restriction of $\omega(\mathcal{Y},\Delta)$ onto $\mathcal{H}$ and use \eqref{cansimplexdec}. Let us choose a set of local holomorphic coordinates $(y,\,\tilde{y}_a)$ such that the locus $\tilde{y}_a \,=\,0$ locally identifies the hyperplane $\mathcal{H}$, with $y$ collectively indicating the remaining local coordinates. Then the covariant restriction of $\omega(\mathcal{Y},\Delta^{(\sigma)})$ to $\mathcal{H}$ is:
 \begin{equation}
 \omega^{\mbox{\tiny $(N-M)$}}(y,\,\Delta^{(\sigma)}_{\mathcal{H}})\:\sim\:\frac{\langle y\,\mbox{d}^M y \rangle}{f_{N+1}(y)^{N-M+1}\prod_{s=1}^{M} f_{\sigma_s}(y)}  
 \end{equation}
 which is a covariant form of degree $N-M$ of the simplex $(\mathbb{P}^M,\Delta^{(\sigma)}_{\mathcal{H}})$, whose facets are $\lbrace f_{\mbox{\tiny $N+1$}}, f_{\sigma_1}, \ldots, f_{\sigma_M} \rbrace$, where $f_i(y):= F_i(\mathcal{Y})|_{\tilde{y}_a=0}$ for $i \in [N+1]$.
 Notice that the covariant form has a pole of order $N-M+1$ in $f_{N+1}$ since $\Delta^{(\sigma)}$ has $N-M+1$ facets which intersect in $f_{\mbox{\tiny $N+1$}}$: 
 \begin{equation}
     f_{\mbox{\tiny $N+1$}}=F_{\mbox{\tiny $N+1$}}\cap \mathcal{H}=B_a \cap \mathcal{H}, \quad a \in [N-M].
 \end{equation}
 Then, by \eqref{cansimplexdec}, we have:
 \begin{equation}
        \omega^{\mbox{\tiny $(N-M)$}}(\mathcal{Y}_{\mathcal{H}}) \sim 
  \sum_{\sigma \in {[N] \choose M}}\omega^{(N-M)}(\mathcal{Y}, \Delta^{(\sigma)}_{\mathcal{H}})
 \end{equation}
 Since $\lbrace \Delta^{(\sigma)}\rbrace$ provides a signed triangulation of $\Delta$, then $\lbrace \Delta^{(\sigma)} \cap \mathcal{H}= \Delta^{(\sigma)}_{\mathcal{H}}\rbrace$ is a signed triangulation of $\Delta_{\mathcal{H}}$. Therefore $(\Delta_{\mathcal{H}},\,\omega^{\mbox{\tiny $(N-M)$}}(\mathcal{Y}_{\mathcal{H}},\,\Delta_{\mathcal{H}}))$ is a covariant form triangulation of $\lbrace (\Delta^{(\sigma)}_{\mathcal{H}},\omega^{\mbox{\tiny $(N-M)$}}(\mathcal{Y}_{\mathcal{H}},\,\Delta^{(\sigma)}_{\mathcal{H}}) \rbrace$.
 We comment on  why $\Delta^{(N)}_{\sigma} \cap \mathcal{H}$ is a simplex.
\end{proof}

\begin{theorem}\label{th:covpair}
  Let $(\mathbb{P}^N, \mathcal{P})$ be a projective polytope and $\omega(\mathcal{Y},\mathcal{P})$ its canonical form. Given a hyperplane $\mathcal{H}$ of codimension $N-M$ in $\mathbb{P}^N$, let $\omega^{\mbox{\tiny $(N-M)$}}(\mathcal{Y}_{\mathcal{H}})$ be the covariant form of degree $N-M$ of the restriction as in \eqref{eq:CPcf}. Then ${(\mathcal{P}_{\mathcal{H}}, \, \omega^{\mbox{\tiny $(N-M)$}}(\mathcal{Y}_{\mathcal{H}}))}$ is a covariant pairing.
\end{theorem}

\begin{proof}
 Given a projective polytope $(\mathbb{P}^N,\,\mathcal{P})$, let us consider its triangulation via the simplices $\lbrace \Delta^{\mbox{\tiny $(j)$}} \rbrace$. Then, it is possible to triangulate each $\Delta^{\mbox{\tiny $(j)$}}$ using the signed triangulations defined in Lemma \ref{lem:triangulation} as $\lbrace \Delta^{\mbox{\tiny $(j\sigma_j)$}}\rbrace$. Of course $\lbrace \Delta^{\mbox{\tiny $(j\sigma_j)$}} \rbrace$ is a signed triangulation of $(\mathbb{P}^N,\,\mathcal{P})$ as well. By Lemma \ref{lem:triangulation}, the covariant restriction of $\omega(\mathcal{Y},\,\Delta^{\mbox{\tiny $j\sigma_j$}})$ on $\mathcal{H}$ is a covariant form of degree $N-M$ of the simplex $\Delta^{\mbox{\tiny $(j\sigma_j)$}}_{\mathcal{H}}=\Delta^{\mbox{\tiny $(j\sigma_j)$}} \cap \mathcal{H}$. Therefore:
 \begin{equation}
     \omega^{\mbox{\tiny $(N-M)$}}(\mathcal{Y}_{\mathcal{H}})=\sum_{j,\sigma_j} \omega^{\mbox{\tiny $(N-M)$}}(\mathcal{Y}_{\mathcal{H}},\Delta^{\mbox{\tiny $(j \sigma_j)$}}_{\mathcal{H}}).
 \end{equation}
\end{proof}

In Lemma \ref{lem:triangulation} we encountered restrictions of a simplices on hyperplanes, and we will see them again in Section \ref{subsec:jkex} as well. In general, every polytope can be realised as a restriction from a simplex of suitable dimension.
Therefore, restrictions of arbitrary polytopes are subsumed under the study of restrictions of simplices. Surprisingly, despite the simplicity of simplices, little is known about the geometric and combinatorial properties of restrictions of simplices in full generality. We refer to \cite{Prabhu1999, Bezdek1990} for related questions and answers on such properties. In particular, in \cite{Prabhu1999} it is shown that, given a simplex $\Delta$ in $\mathbb{P}^N$, there exists an hyperplane $\mathcal{H}$ of even dimension $M$ such that it intersects the interior of \emph{all} the faces of $\Delta$ of dimension $N-M/2$. The only visualisable example is a $2$-plane which intersects all facets of a tetrahedron: the restriction on such plane gives a quadrilateral. Therefore, curiously enough, we can always intersect the interior of \emph{all} $N+1$ facets of a simplex in $\mathbb{P}^N$ with a $2$-plane. 
If we perform a covariant restriction of the canonical form $\omega(\mathcal{Y},\Delta)$ of $\Delta$ onto such $2$-dimensional hyperplane $\mathcal{H}$, we get a covariant form $\omega^{(N-2)}(\mathcal{Y}_{\mathcal{H}},\Delta_{\mathcal{H}})$ of degree $N-2$ in covariant pairing with $\Delta_{\mathcal{H}}$, i.e. a polygon with $N+1$ edges. This form has \emph{all} simple poles on the edges of the polygon, since $\mathcal{H}$ intersects the all facets of $\Delta$ on the simplex. However, it is not its canonical form: it has poles outside, where non-adjacent edges intersect.

There is an analogous statement \cite{Bezdek1990} for polytopes in $\mathbb{P}^N$, with $N\geq3$: if the polytope has at most $2N$ facets\footnote{Under the condition that the polytope has at a simple vertex (i.e. it belongs to exactly $N$ facets of the polytope), then the result is true also if the polytope has $2N+1$ facets, $N\geq4$.}, then there is always an hyperplane which intersects the interior of \emph{all} its facets.

\begin{theorem}
 Let $(\mathbb{P}^N,\mathcal{P})$ be a projective polytope and $\omega(\mathcal{Y},\mathcal{P})$ its canonical form. Then, given an hyperplane $\mathcal{H}$ of codimension $N-M$ in $\mathbb{P}^N$, let $\omega^{\mbox{\tiny $(N-M)$}}(\mathcal{Y}_{\mathcal{H}},\,\mathcal{P}_{\mathcal{H}})$ be the covariant restriction of $\omega(\mathcal{Y},\mathcal{P})$ onto $\mathcal{H}$, which is in covariant pairing with $(\mathbb{P}^M,\,\mathcal{P}_{\mathcal{H}})$.
 Let $(\mathbb{P}^{M-1},\,\partial\mathcal{P}^{\mbox{\tiny $(j)$}}_{\mathcal{H}})$ be a boundary component of the restriction $(\mathbb{P}^M,\,\mathcal{P}_{\mathcal{H}})$ of $(\mathbb{P}^N,\mathcal{P})$ onto $\mathcal{H}$, corresponding to a pole with multiplicity $N-M+1$ in $\omega^{\mbox{\tiny $(N-M)$}}(\mathcal{Y}_{\mathcal{H}},\,\mathcal{P}_{\mathcal{H}})$. Then, if $\{(\mathbb{P}^{N-1},\,\partial\mathcal{P}^{(\alpha)})\}$ is the collection of boundary components of $(\mathbb{P}^N,\,\mathcal{P})$ such that $\partial\mathcal{P}^{\mbox{\tiny $(\alpha)$}}\cap \mathcal{H} \,=\,\partial\mathcal{P}^{\mbox{\tiny $(j)$}}_{\mathcal{H}}=\cap_{\alpha} \partial \mathcal{P}^{(\alpha)}$, then
 \begin{equation}
     \ResT_{\mbox{\tiny $\bigcap_\alpha \partial\mathcal{P}^{\mbox{\tiny $(\alpha)$}}$}}\left\{\omega(\mathcal{Y},\,\mathcal{P})\right\}\:=\:
      \Lor_{\partial\mathcal{P}^{\mbox{\tiny $(j)$}}_{\mathcal{H}}}{}^{\mbox{\tiny $(N-M+1)$}}\left\{\omega^{\mbox{\tiny $(N-M)$}}(\mathcal{Y}_{\mathcal{H}},\,\mathcal{P}_{\mathcal{H}})\right\}.
 \end{equation}
\end{theorem}
\begin{proof}
 Let $(\mathbb{P}^N,\mathcal{P})$ be a projective polytope and $\omega(\mathcal{Y},\mathcal{P})$ its canonical form. If $\tilde{n}$ is the total number of its facets, then 
 its canonical form can be decomposed as\footnote{This corresponds to picking a triangulation of the polytope.}:
 \begin{equation} \label{eq:candec}
    \omega(\mathcal{Y},\mathcal{P})= \sum_{\sigma \in { [\tilde{n}]\choose N+1} } \frac{a_{\sigma}}{F_{\sigma_1} \ldots F_{\sigma_{N+1}}} \langle \mathcal{Y} \, \mbox{d}^N \mathcal{Y}\rangle
 \end{equation}
 where $ F_i\,=\,F_i(\mathcal{Y})$ is the linear homogeneous polynomial identifying the facet $\mathcal{F}_i$, with $i \in [\tilde{n}]$.
 Among these, let us denote as $Q_\alpha$ the linear homogeneous polynomial identifying the facet $\partial \mathcal{P}^{(\alpha)}$, with $\alpha=1,\ldots,m$. 
 Then the residue operator receives contributions only from terms of the following type:
 \begin{equation} \label{eq:canresterm}
    \frac{1}{\prod_{\alpha \in [m]} Q_{\alpha}}\sum_{\tilde{\sigma} \in I} \frac{a_{\tilde{\sigma}}}{F_{\tilde{\sigma}_1} \ldots F_{\tilde{\sigma}_{M}}} \langle \mathcal{Y} \, \mbox{d}^N \mathcal{Y}\rangle 
 \end{equation}
 where $I \subseteq {[\tilde{\nu}] \choose M}$ such that $\lbrace F_{\tilde{\sigma}_1} \ldots F_{\tilde{\sigma}_{M}} \rbrace$ does not contain any of the $Q_{\alpha}$, and we used the fact that $N+1-m =M$.  Let us now parametrise $\mathbb{P}^N$ via a set of local holomorphic coordinates $(x_j,\,h_\alpha)$ such that the locus $h_{\alpha}\,=\,0$ locally identifies with the facet $Q_{\alpha}$, with $x_j$ collectively indicating the remaining local coordinates. Then:
 \begin{equation} \label{res:can}
     \Res_{\mbox{\tiny $\lbrace Q_{\alpha}=0 \rbrace$}} \omega(\mathcal{Y},\mathcal{P}) \sim  \sum_{\tilde{\sigma} \in I} \left. \frac{a_{\tilde{\sigma}}}{F_{\tilde{\sigma}_1} \ldots F_{\tilde{\sigma}_{M}}} \right|_{h_{\alpha}=0} \langle x \, \mbox{d}^{M-1} x\rangle,
 \end{equation}
 where $\sim$ expresses the result up to an overall constant.
 
 We now focus on the covariant form $\omega^{\mbox{\tiny $(N-M)$}}(\mathcal{Y},\mathcal{P}_{\mathcal{H}})$. Using the expansion \eqref{eq:candec}, we can see that, the only term contributing to its leading Laurent coefficient around the pole corresponding to the boundary $\partial \mathcal{P}^{(j)}_{\mathcal{H}}$ is the restriction on $\mathcal{H}$ of the form in \eqref{eq:canresterm}. If we parametrise $\mathbb{P}^N$ via a set of local holomorphic coordinates $(y,\,\tilde{y}_a)$ such that the locus $\tilde{y}_a \,=\,0$ locally identifies with the hyperplane $\mathcal{H}$, with $y$ collectively indicating the remaining local coordinates, then this restriction reads:
 \begin{equation}
    \frac{1}{q^m}\sum_{\tilde{\sigma} \in I} \left. \frac{a_{\tilde{\sigma}}}{F_{\tilde{\sigma}_1} \ldots F_{\tilde{\sigma}_{M}}} \right|_{\tilde{y}_a=0} \langle y \, \mbox{d}^M y \rangle,
 \end{equation}
 where we denoted $q$ the linear homogeneous polynomial corresponding to the boundary $\partial \mathcal{P}^{(j)}_{\mathcal{H}}$ and used the fact that $Q_{\alpha}|_{\tilde{y}_a=0}=q$, since $\partial \mathcal{P}^{(\alpha)} \cap \mathcal{H}=\partial \mathcal{P}^{(j)}_{\mathcal{H}}$.
  Furthermore, let us choose coordinates $(\tilde{x},\,x_j)$ such that the locus $\tilde{x} \,=\,0$ locally identifies with the locus of the pole $q=0$, and $x$ collectively indicating the remaining local coordinates. Then:
  \begin{equation} \label{llc:can}
     \Lor_{\mbox{\tiny $\lbrace q=0 \rbrace$}} \omega^{(N-M)}(\mathcal{Y},\mathcal{P}_{\mathcal{H}}) \sim  \sum_{\tilde{\sigma} \in I} \left. \frac{a_{\tilde{\sigma}}}{F_{\tilde{\sigma}_1} \ldots F_{\tilde{\sigma}_{M}}} \right|_{\tilde{y}_a,\tilde{x}=0} \langle x \, \mbox{d}^{M-1} x\rangle.
 \end{equation}
 By hypotheses, $(\mathcal{F}_i \cap \mathcal{H})\cap \partial \mathcal{P}^{(j)}_{\mathcal{H}}=\mathcal{F}_i \cap_{\alpha \in [m]} \partial \mathcal{P}^{(\alpha)}$, then the restriction to $h_{\alpha}=0$ in \eqref{res:can} and the restriction to $\tilde{y}_a,\tilde{x}=0$ in \eqref{llc:can} coincide. The statement of theorem follows immediately, once we consider a representative of \eqref{llc:can} such that it has unit leading singularities. 
\end{proof}

\subsection{Visualisable Examples: Polygons and Polyhedra}\label{subsec:Ex}

In order to illustrate the covariant restriction map between the canonical form of a parent polytope and the covariant pairing of its child polytope, and how their structures are tied to each other, we will discuss some non-trivial example in the two visualisable cases, {\it i.e.} polytopes in $\mathbb{P}^2$ and $\mathbb{P}^3$, distinguishing between the cases in which the restriction is with respect to a hyperplane intersecting the parent polytope inside only, and when the hyperplane can intersect its facets outside.


\subsubsection{Polygons and Internal Intersections}

Let us consider the simplest non-trivial examples of polytopes in $\mathbb{P}^2$, and let us indicate the $n$-gons as $\mathcal{P}_n$. First, notice that just for $n\,=\,3,\,4$ there exist hyper-planes $\mathcal{H}$ which intersect $\mathcal{P}_n$ inside or on its boundaries only:
\begin{figure}[H]
 \centering
 \begin{tikzpicture}[line join = round, line cap = round, ball/.style = {circle, draw, align=center, anchor=north, inner sep=0}, 
                     axis/.style={very thick, ->, >=stealth'}, pile/.style={thick, ->, >=stealth', shorten <=2pt, shorten>=2pt}, every node/.style={color=black}]
 \begin{scope}
  \coordinate [label=above:{\footnotesize $\displaystyle 1$}] (A) at (0,0);
  \coordinate [label=left:{\footnotesize $\displaystyle 2$}] (B) at (-1.75,-2.25);
  \coordinate [label=right:{\footnotesize $\displaystyle 3$}] (C) at (+1.75,-2.25);
  
  \draw[-, fill=blue!30, opacity=.7] (A) -- (B) -- (C) -- cycle;  
  
  \coordinate [label=below:{\footnotesize $\displaystyle 4$}] (Z) at ($(B)!0.25!(C)$);
  \coordinate [label=right:{\small $\displaystyle \mathcal{H}$}] (ZU) at ($(A)!-0.75cm!(Z)$);
  \coordinate [label=left:{\small $\displaystyle \mathcal{P}_3$}] (P3) at ($(A)!0.375!(B)$);
  \draw[color=red, thick] ($(A)!-0.75cm!(Z)$) -- ($(Z)!-0.75cm!(A)$);
  \draw[color=red, very thick] (A) -- (Z);
  \draw[fill=red, color=red] (A) circle (1pt);
- \draw[fill=red, color=red] (Z) circle (1pt);
 \end{scope}
 \begin{scope}[shift={(7,0)}, transform shape]
  \coordinate [label=above:{\footnotesize $\displaystyle 1$}] (A) at (0,0);
  \coordinate [label=left:{\footnotesize $\displaystyle 2$}] (B) at (-.5,-2.25);
  \coordinate [label=right:{\footnotesize $\displaystyle 3$}] (C) at (2,-2);
  \coordinate [label=above:{\footnotesize $\displaystyle 4$}] (D) at (1.75,-.25);

  \draw[-, fill=blue!30, opacity=.7] (A) -- (B) -- (C) -- (D) -- cycle;
  \coordinate [label=right:{\small $\displaystyle \mathcal{H}$}] (H) at ($(D)!-0.75cm!(B)$);
  \coordinate [label=left:{\small $\displaystyle \mathcal{P}_4$}] (P4) at ($(A)!0.5!(B)$);
  \draw[color=red, thick] ($(D)!-0.75cm!(B)$) -- ($(B)!-0.75cm!(D)$);
  \draw[color=red, very thick] (D) -- (B);
  \draw[fill=red, color=red] (D) circle (1pt);
- \draw[fill=red, color=red] (B) circle (1pt);
 \end{scope}
 \end{tikzpicture}
 %
\end{figure}
For all $n\,\ge\,5$ such hyperplanes do not exist, and any hyperplane intersects $\mathcal{P}_{n\ge5}$ both inside {\it and} outside (see Figure \ref{fig:PolygH}). Let us begin with discussing the two examples in which the intersection $\mathcal{H}\,\bigcap\,\mathcal{P}_n$ lies completely inside the convex hull $\mathcal{P}_n$, {\it i.e.} for the triangle and the square depicted above. Let us choose the local coordinates $\mathcal{Y}\,=\,(y_1,\,y_2,\,y_3)$. For the triangle $\mathcal{P}_3$, let us take its vertices to be $Z_1\,=\,(1,0,0)$, $Z_2\,=\,(0,1,0)$, $Z_3\,=\,(0,0,1)$, then its canonical form is given by
\begin{equation}\label{eq:P3cf}
    \omega(\mathcal{Y},\,\mathcal{P}_3)\:=\:\frac{\langle123\rangle^2}{\langle\mathcal{Y}12\rangle\langle\mathcal{Y}23\rangle\langle\mathcal{Y}31\rangle}\langle\mathcal{Y}d^2\mathcal{Y}\rangle\:=\:
   \bigwedge_{j=1}^3\frac{dy_j}{y_j}\frac{1}{\mbox{Vol}\{GL(1)\}}.
\end{equation}
Let us now consider the hyperplane $\mathcal{H}$ defined as
\begin{equation}
    \mathcal{H}\:=\:\left\{\mathcal{Y}\,\in\,\mathbb{P}^2\,|\,\langle\mathcal{Y}14\rangle\,=\,0\,=\,\alpha y_3-(1-\alpha) y_2\right\},
\end{equation}
where $Z_4\,=\,\alpha Z_2+(1-\alpha)Z_3\,=\,(0,\alpha,\,1-\alpha)$, with $\alpha\,\in\,]0,1[$ so to guarantee that it lies inside the boundary $(2,3)$, and the last equality is just the representation of the hyperplane in our local coordinates. The restriction $\mathcal{P}_{\mathcal{H}}\, :=\,\mathcal{P}_3\cap\mathcal{H}$ is just the segment with boundaries in $Z_1$ and $Z_4$. Then, the covariant restriction \eqref{eq:CPcf} of \eqref{eq:P3cf} yields:
\begin{equation}\label{eq:P3cfR}
    \omega^{\mbox{\tiny $(1)$}}(\mathcal{Y}_{\mathcal{H}}\,\mathcal{P}_{\mathcal{H}})
      \:=\:\frac{\langle Z_{\star}14\rangle}{\langle\mathcal{Y}_{\mathcal{H}}Z_{\star}1\rangle^2\langle\mathcal{Y}_{\mathcal{H}}Z_{\star}4\rangle}\langle Z_{\star}\mathcal{Y}_{\mathcal{H}}d\mathcal{Y}_{\mathcal{H}}\rangle
      \:\sim\:\frac{1}{y_2^2\,y_1}\frac{dy_1\wedge\,dy_2}{\mbox{Vol$\{GL(1)\}$}},
\end{equation}
where $Z_{\star}\, :=\,(0,\,-(1-\alpha),\,\alpha)$ identifies the restriction on $\mathcal{H}$, and the symbol $\sim$ indicates that the form is defined up to an overall constant, {\it i.e.} there is an equivalence class of degree-$1$ covariant forms, and \eqref{eq:P3cf} is a representative. In this case the boundary components of the parent polytope are mapped to boundary components of the child polytope, and the covariant form \eqref{eq:P3cfR} in covariant pairing with the segment $(\mathbb{P}^1,\,\mathcal{P}_{\mathcal{H}})$ has poles only on the boundary components of the segment, {\it i.e.} the vertices $(1,\,4)$. Notice that the double pole in the facet of the child polytope is the manifestation of the fact that there are two boundaries of the parent polytope (the triangle) which are projected onto it. while there is a single pole in correspondence of the facet of the child polytope encoding just one facet of the parent polytope. Notice also that the covariant form \eqref{eq:P3cfR} does not depend on $\alpha$, which parametrises the intersection between the hyperplane $\mathcal{H}$ and the facet $(2,\,3)$ of $\mathcal{P}_3$, or, more precisely, such a dependence results in an overall coeffcient. Hence, the form structure is not changed and all the forms differing by the $\alpha$-dependent scale factor belongs to the same equivalence class.

Finally notice that the leading Laurent coefficients of the covariant form \eqref{eq:P3cfR} of the child polytope of each of the poles -- for the simple pole it is just its residue -- return the canonical form of a lower codimension boundary of the parent polytope, which for the double pole is simply the canonical form of the vertex $1$ of the parent triangle.

We can repeat the same analysis for a square $\mathcal{P}_4$ intersected by the hyperplane $\mathcal{H}\,=\,\{\mathcal{Y}\,\in\,\mathbb{P}^2\,|\,\langle\mathcal{Y}24\rangle\,=\,0\}$ in the figure above. For the sake of concreteness, let us take the vertices of the square to be $Z_1\,=\,(1,0,0)$, $Z_2\,=\,(0,1,0)$, $Z_3\,=\,(0,0,1)$, $Z_4\,=\,\alpha\,Z_1\,-\,(\alpha+\beta-1)Z_2\,+\,\beta\,Z_3$ (with $\alpha+\beta-1\,>\,0$). The the canonical form associated the square is given by
\begin{equation}\label{eq:P4cf}
 \begin{split}
    \omega(\mathcal{Y},\,\mathcal{P}_4)\:&=\:\frac{\langle\mathcal{Y}Z_{13}Z_{24}\rangle}{\langle\mathcal{Y}12\rangle\langle\mathcal{Y}23\rangle\langle\mathcal{Y}34\rangle\langle\mathcal{Y}41\rangle}
                                              \langle\mathcal{Y}d^2\mathcal{Y}\rangle\:=\\
                                         &=\:\frac{\beta(\alpha+\beta-1)y_1+\alpha\beta y_2 + \alpha(\alpha+\beta-1)y_3}{y_3 y_1[(\alpha+\beta-1)y_1+\alpha y_2][\beta y_2+(\alpha+\beta-1)y_3]}\,\bigwedge_{j=1}^3\,dy_j\frac{1}{\mbox{Vol}\{GL(1)\}}
 \end{split}
\end{equation}
where $Z_{ij}\,:=\,(i,i+1)\bigcap(j,j+1)\,:=\,Z_i\langle i+1,j,j+1\rangle-Z_{i+1}\langle i,j,j+1\rangle$ represents the intersection between the two facets $(i,i+1)$ and $(j,j+1)$. The line identified by the two points $Z_{13}$ and $Z_{24}$ provide a zero of the canonical form \eqref{eq:P4cf}. In such local coordinates, the line $\mathcal{H}$ is identified by the equation $\alpha\,y_3-\beta\,y_1\,=\,0$, and the restriction of $\mathcal{P}_4$ onto it is simply the segment with boundaries $Z_2$ and $Z_4$. The covariant restriction of the canonical form \eqref{eq:P4cf} is therefore
\begin{equation}
    \omega^{\mbox{\tiny $(1)$}}(\mathcal{Y}_{\mathcal{H}},\mathcal{P}_{\mathcal{H}})\:=\:
        \frac{\langle\mathcal{Y}_{\mathcal{H}}Z_{\star}Z_{\circ}\rangle}{\langle\mathcal{Y}_{\mathcal{H}}Z_{\star}2\rangle^2\langle\mathcal{Y}_{\mathcal{H}}Z_{\star}4\rangle^2}
        \langle Z_{\star}\mathcal{Y}_{\mathcal{H}}d\mathcal{Y}_{\mathcal{H}}\rangle\:\sim\:
        \frac{2(\alpha+\beta-1)y_1+\alpha y_2}{y_1^2[(\alpha+\beta-1)y_1+\alpha y_2]^2}\,\frac{dy_1\wedge\,dy_2}{\mbox{Vol}\{GL(1)\}},
\end{equation}
with $Z_{\star}\,:=\,(\beta,0,-\alpha)$ identifying the restriction on $\mathcal{H}$, and $Z_{\circ}\,:=\,(24)\bigcap (Z_{13}Z_{24})$ being the projection of the locus identifying the zero of the canonical form of the parent polytope onto $\mathcal{H}$. Here we can see how a covariant form of a polytope inherits the zero of the canonical form of the parent polytope, which is now a point outside the segment $(2,4)$ in $\mathbb{P}^1$. If we were to start from the segment $(\mathbb{P}^1,\,\mathcal{P})$ and associate to it a covariant form of degree-$1$ with both poles of second order, the homogeneiy condition would fix the numerator to be linear, but then, as we already saw in the previous section, no other defining property of a covariant form, would fix the coefficients up to an overall constant. We would need some extra information, but we have no reason to choice any special point outside the segment as a zero given that does not arise from any geometrical feature of the segment itself. 


\subsubsection{Polygons with Outer Intersections}\label{subsubsec:POI}

Let us now consider the case of $n$-gons $(\mathbb{P}^2,\,\mathcal{P}_n)$ and an hyperline $\mathcal{H}$ intersecting their facets both inside and outside the convex hull $\mathcal{P}_n$ (see Figure \ref{fig:PolygH}). We begin with the simplest example of the triangle $\mathcal{P}_3$. For any hyper-line $\mathcal{H}$, its intersection with the facets of the triangle $\mathcal{P}_3$ occurs on three points, which we label $Z_4,\,Z_5,\,Z_6$ following the notation of Figure \ref{fig:PolygH}, with two of them inside the polytope $Z_4,\,Z_5$ and the third one $Z_6$ outside. Hence, the hyper-line $\mathcal{H}$ is identified by
\begin{equation}\label{eq:H45s}
    \mathcal{H}\:=\:\left\{\mathcal{Y}\,\in\,\mathbb{P}^2\,\big|\,\langle\mathcal{Y}45\rangle\,=\,0\right\},
\end{equation}
with 
\begin{equation}\label{eq:Z456}
    Z_4\:\sim\:\alpha Z_2 + (1-\alpha)Z_3,\qquad
    Z_5\:\sim\:\beta Z_3 + (1-\beta)Z_1,\qquad
    Z_6\:=\:(12)\cap(45)
\end{equation}
where $\alpha,\,\beta\,\in\,]0,1[$.

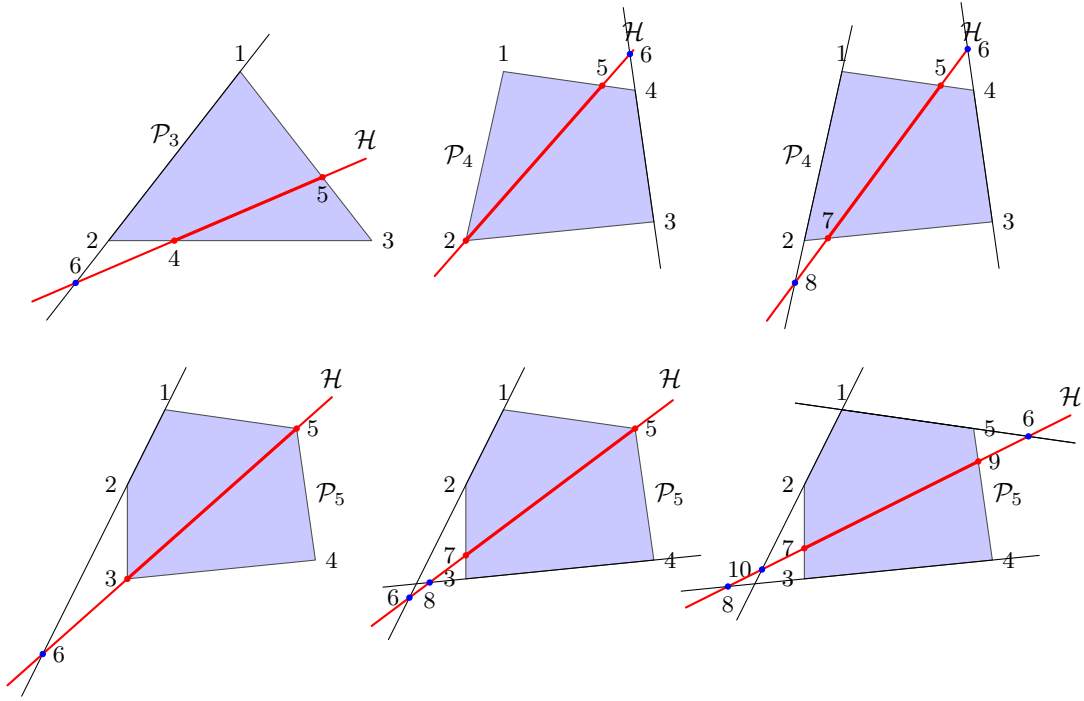
\begin{figure}[h]
 \centering
 \begin{tikzpicture}[line join = round, line cap = round, ball/.style = {circle, draw, align=center, anchor=north, inner sep=0}, 
                     axis/.style={very thick, ->, >=stealth'}, pile/.style={thick, ->, >=stealth', shorten <=2pt, shorten>=2pt}, every node/.style={color=black}]
 \begin{scope}
  \coordinate [label=above:{\footnotesize $\displaystyle 1$}] (A) at (0,0);
  \coordinate [label=left:{\footnotesize $\displaystyle 2$}] (B) at (-1.75,-2.25);
  \coordinate [label=right:{\footnotesize $\displaystyle 3$}] (C) at (+1.75,-2.25);
  
  \draw[-, fill=blue!30, opacity=.7] (A) -- (B) -- (C) -- cycle;  
  
  \coordinate [label=below:{\footnotesize $\displaystyle 4$}] (Z4) at ($(B)!0.25!(C)$);
  \coordinate [label=below:{\footnotesize $\displaystyle 5$}] (Z5) at ($(C)!0.375!(A)$);
  \coordinate [label=above:{\footnotesize $\displaystyle 6$}] (Z6) at (intersection of A--B and Z4--Z5);

  \coordinate [label=above:{\small $\displaystyle \mathcal{H}$}] (H) at ($(Z5)!-0.625cm!(Z4)$);
  \coordinate [label=left:{\small $\displaystyle \mathcal{P}_3$}] (P3) at ($(A)!0.375!(B)$);
  \draw[color=red, thick] ($(Z5)!-0.625cm!(Z6)$) -- ($(Z6)!-0.625cm!(Z5)$);
  \draw ($(A)!-0.625cm!(Z6)$) -- ($(Z6)!-0.625cm!(A)$);
  \draw[color=red, very thick] (Z4) -- (Z5);
  \draw[fill=red, color=red] (Z4) circle (1pt);
- \draw[fill=red, color=red] (Z5) circle (1pt);
  \draw[fill=red, color=blue] (Z6) circle (1pt);
 \end{scope}
 \begin{scope}[shift={(3.5,0)}, transform shape]
  \coordinate [label=above:{\footnotesize $\displaystyle 1$}] (A) at (0,0);
  \coordinate [label=left:{\footnotesize $\displaystyle 2$}] (B) at (-.5,-2.25);
  \coordinate [label=right:{\footnotesize $\displaystyle 3$}] (C) at (2,-2);
  \coordinate [label=right:{\footnotesize $\displaystyle 4$}] (D) at (1.75,-.25);

  \draw[-, fill=blue!30, opacity=.7] (A) -- (B) -- (C) -- (D) -- cycle;
  \coordinate [label=above:{\footnotesize $\displaystyle 5$}] (Z5) at ($(D)!0.25!(A)$);
  \coordinate [label=right:{\footnotesize $\displaystyle 6$}] (Z6) at (intersection of D--C and B--Z5);
  \coordinate [label=above:{\small $\displaystyle \mathcal{H}$}] (H) at ($(Z5)!-0.625cm!(B)$);
  \coordinate [label=left:{\small $\displaystyle \mathcal{P}_4$}] (P4) at ($(A)!0.5!(B)$);
  \draw[color=red, thick] ($(Z5)!-0.625cm!(B)$) -- ($(B)!-0.625cm!(Z5)$);
  \draw[color=red, very thick] (Z5) -- (B);
  \draw ($(C)!-0.625cm!(Z6)$) -- ($(Z6)!-0.625cm!(C)$);
  \draw[fill=red, color=red] (Z5) circle (1pt);
- \draw[fill=red, color=red] (B) circle (1pt);
 \draw[fill=red, color=blue] (Z6) circle (1pt);
 \end{scope}
 \begin{scope}[shift={(8,0)}, transform shape]
  \coordinate [label=above:{\footnotesize $\displaystyle 1$}] (A) at (0,0);
  \coordinate [label=left:{\footnotesize $\displaystyle 2$}] (B) at (-.5,-2.25);
  \coordinate [label=right:{\footnotesize $\displaystyle 3$}] (C) at (2,-2);
  \coordinate [label=right:{\footnotesize $\displaystyle 4$}] (D) at (1.75,-.25);

  \draw[-, fill=blue!30, opacity=.7] (A) -- (B) -- (C) -- (D) -- cycle;
  \coordinate [label=above:{\footnotesize $\displaystyle 5$}] (Z5) at ($(D)!0.25!(A)$);
  \coordinate [label=above:{\footnotesize $\displaystyle 7$}] (Z7) at ($(B)!0.125!(C)$);
  \coordinate [label=right:{\footnotesize $\displaystyle 6$}] (Z6) at (intersection of D--C and Z7--Z5);
  \coordinate [label=right:{\footnotesize $\displaystyle 8$}] (Z8) at (intersection of A--B and Z7--Z5);
  \coordinate [label=above:{\small $\displaystyle \mathcal{H}$}] (H) at ($(Z5)!-0.625cm!(B)$);
  \coordinate [label=left:{\small $\displaystyle \mathcal{P}_4$}] (P4) at ($(A)!0.5!(B)$);
  \draw[color=red, thick] ($(Z5)!-0.625cm!(Z8)$) -- ($(Z8)!-0.625cm!(Z5)$);
  \draw[color=red, very thick] (Z5) -- (Z7);
  \draw ($(C)!-0.625cm!(Z6)$) -- ($(Z6)!-0.625cm!(C)$);
  \draw ($(A)!-0.625cm!(Z8)$) -- ($(Z8)!-0.625cm!(A)$);
  \draw[fill=red, color=red] (Z5) circle (1pt);
- \draw[fill=red, color=red] (Z7) circle (1pt);
 \draw[fill=red, color=blue] (Z6) circle (1pt);
 \draw[fill=red, color=blue] (Z8) circle (1pt);
 \end{scope}
 \begin{scope}[shift={(-1,-4.5)}, transform shape]
  \coordinate [label=above:{\footnotesize $\displaystyle 1$}] (A) at (0,0);
  \coordinate [label=left:{\footnotesize $\displaystyle 3$}] (B) at (-.5,-2.25);
  \coordinate [label=left:{\footnotesize $\displaystyle 2$}] (E) at ($(A)!0.5!(B)+(-.25,+.125)$);
  \coordinate [label=right:{\footnotesize $\displaystyle 4$}] (C) at (2,-2);
  \coordinate [label=right:{\footnotesize $\displaystyle 5$}] (D) at (1.75,-.25);

  \draw[-, fill=blue!30, opacity=.7] (A) -- (E) -- (B) -- (C) -- (D) -- cycle;
  \coordinate [label=right:{\footnotesize $\displaystyle 6$}] (Z6) at (intersection of A--E and D--B);
  \coordinate [label=above:{\small $\displaystyle \mathcal{H}$}] (H) at ($(D)!-0.625cm!(B)$);
  \coordinate [label=right:{\small $\displaystyle \mathcal{P}_5$}] (P5) at ($(C)!0.5!(D)$);
  \draw[color=red, thick] ($(D)!-0.625cm!(Z6)$) -- ($(Z6)!-0.625cm!(D)$);
  \draw[color=red, very thick] (D) -- (B);
  \draw ($(A)!-0.625cm!(Z6)$) -- ($(Z6)!-0.625cm!(A)$);
  \draw[fill=red, color=red] (D) circle (1pt);
- \draw[fill=red, color=red] (B) circle (1pt);
 \draw[fill=red, color=blue] (Z6) circle (1pt);
 \end{scope}
 \begin{scope}[shift={(3.5,-4.5)}, transform shape]
  \coordinate [label=above:{\footnotesize $\displaystyle 1$}] (A) at (0,0);
  \coordinate [label=left:{\footnotesize $\displaystyle 3$}] (B) at (-.5,-2.25);
  \coordinate [label=left:{\footnotesize $\displaystyle 2$}] (E) at ($(A)!0.5!(B)+(-.25,+.125)$);
  \coordinate [label=right:{\footnotesize $\displaystyle 4$}] (C) at (2,-2);
  \coordinate [label=right:{\footnotesize $\displaystyle 5$}] (D) at (1.75,-.25);

  \draw[-, fill=blue!30, opacity=.7] (A) -- (E) -- (B) -- (C) -- (D) -- cycle;
  \coordinate [label=left:{\footnotesize $\displaystyle 7$}] (Z7) at ($(E)!.75!(B)$);
  \coordinate [label=left:{\footnotesize $\displaystyle 6$}] (Z6) at (intersection of A--E and D--Z7);
  \coordinate [label=below:{\footnotesize $\displaystyle 8$}] (Z8) at (intersection of B--C and D--Z7);
  \coordinate [label=above:{\small $\displaystyle \mathcal{H}$}] (H) at ($(D)!-0.625cm!(B)$);
  \coordinate [label=right:{\small $\displaystyle \mathcal{P}_5$}] (P5) at ($(C)!0.5!(D)$);
  \draw[color=red, thick] ($(D)!-0.625cm!(Z6)$) -- ($(Z6)!-0.625cm!(D)$);
  \draw[color=red, very thick] (D) -- (Z7);
  \draw ($(A)!-0.625cm!(Z6)$) -- ($(Z6)!-0.625cm!(A)$);
  \draw ($(C)!-0.625cm!(Z8)$) -- ($(Z8)!-0.625cm!(C)$);
  \draw[fill=red, color=red] (D) circle (1pt);
- \draw[fill=red, color=red] (Z7) circle (1pt);
 \draw[fill=red, color=blue] (Z6) circle (1pt);
 \draw[fill=red, color=blue] (Z8) circle (1pt);
 \end{scope}
 \begin{scope}[shift={(8,-4.5)}, transform shape]
  \coordinate [label=above:{\footnotesize $\displaystyle 1$}] (A) at (0,0);
  \coordinate [label=left:{\footnotesize $\displaystyle 3$}] (B) at (-.5,-2.25);
  \coordinate [label=left:{\footnotesize $\displaystyle 2$}] (E) at ($(A)!0.5!(B)+(-.25,+.125)$);
  \coordinate [label=right:{\footnotesize $\displaystyle 4$}] (C) at (2,-2);
  \coordinate [label=right:{\footnotesize $\displaystyle 5$}] (D) at (1.75,-.25);

  \draw[-, fill=blue!30, opacity=.7] (A) -- (E) -- (B) -- (C) -- (D) -- cycle;
  \coordinate [label=left:{\footnotesize $\displaystyle 7$}] (Z7) at ($(E)!.675!(B)$); 
  \coordinate [label=right:{\footnotesize $\displaystyle 9$}] (Z9) at ($(C)!.75!(D)$);
  \coordinate [label=left:{\footnotesize $\displaystyle 10$}] (Z10) at (intersection of A--E and Z9--Z7);
  \coordinate [label=below:{\footnotesize $\displaystyle 8$}] (Z8) at (intersection of B--C and Z9--Z7);
  \coordinate [label=above:{\footnotesize $\displaystyle 6$}] (Z6) at (intersection of A--D and Z9--Z7);
  \coordinate [label=above:{\small $\displaystyle \mathcal{H}$}] (H) at ($(Z6)!-0.625cm!(Z9)$);
  \coordinate [label=right:{\small $\displaystyle \mathcal{P}_5$}] (P5) at ($(C)!0.5!(D)$);
  \draw[color=red, thick] ($(Z6)!-0.625cm!(Z8)$) -- ($(Z8)!-0.625cm!(Z6)$);
  \draw[color=red, very thick] (Z9) -- (Z7);
  \draw ($(A)!-0.625cm!(Z6)$) -- ($(Z6)!-0.625cm!(A)$);
  \draw ($(C)!-0.625cm!(Z8)$) -- ($(Z8)!-0.625cm!(C)$);
  \draw ($(A)!-0.625cm!(Z6)$) -- ($(Z6)!-0.625cm!(A)$);
  \draw ($(A)!-0.625cm!(Z10)$) -- ($(Z10)!-0.75cm!(A)$);
  \draw[fill=red, color=red] (Z9) circle (1pt);
- \draw[fill=red, color=red] (Z7) circle (1pt);
 \draw[fill=red, color=blue] (Z6) circle (1pt);
 \draw[fill=red, color=blue] (Z8) circle (1pt);
 \draw[fill=red, color=blue] (Z10) circle (1pt);
 \end{scope}
 \end{tikzpicture}
 \caption{Examples of projective polytopes $(\mathbb{P}^2,\,\mathcal{P}_n)$ intersected by a hyperplane on its facets both inside and outside $\mathcal{P}_{n}$. For polygons $\mathcal{P}_{n\,\ge\,5}$ the latter is the only choice for an hyperplane $\mathcal{H}$ such that $\mathcal{H}\,\bigcap\,\mathcal{P}_{n\,\ge\,5}\,\neq\,\varnothing$}
 \label{fig:PolygH}
\end{figure}
\noindent
The canonical form of the triangle \eqref{eq:P3cf} reduces to a covariant form of degree $1$ in $\mathbb{P}^1$ with {\it three} simple poles:
\begin{equation}\label{eq:P3chHo}
    \omega(\mathcal{Y}_{\mathcal{H}},\,\mathcal{P}_{\mathcal{H}})\:=\:
     \frac{\langle Z_{\star}45\rangle\langle Z_{\star}\mathcal{Y}_{\mathcal{H}}d\mathcal{Y}_{\mathcal{H}}\rangle}{\langle\mathcal{Y}_{\mathcal{H}}Z_{\star}6\rangle\langle\mathcal{Y}_{\mathcal{H}}Z_{\star}4\rangle\langle\mathcal{Y}_{\mathcal{H}}Z_{\star}5\rangle},
\end{equation}
where $Z_{\star}$ is the orthogonal complement of $\mathcal{W}_{\mbox{\tiny $I$}}^{\mbox{\tiny $(\mathcal{H})$}}\,:=\,\varepsilon_{\mbox{\tiny $IJK$}}Z_4^{\mbox{\tiny $J$}}Z_5^{\mbox{\tiny $K$}}$. Such a covariant form is in covariant pairing with $(\mathbb{P}^1,\,\mathcal{P}_{\mathcal{H}})$, where $\mathcal{P}_{\mathcal{H}}\,:=\,\mathcal{P}_2^{\mbox{\tiny $(45)$}}$\footnote{The superscript $(ab)$ in $\mathcal{P}_2^{\mbox{\tiny $(ab)$}}$ labels the boundaries (vertices) of the segment.} and it can be seen as the sum of the covariant forms associated to the two segments $(\mathbb{P}^1,\,\mathcal{P}_2^{\mbox{\tiny $(65)$}})$ and $(\mathbb{P}^1,\,\mathcal{P}_2^{\mbox{\tiny $(46)$}})$ which provide a signed triangulation of $(\mathbb{P}^1,\,\mathcal{P}_2^{\mbox{\tiny $(45)$}})$ through $Z_6$ and, consequently, the covariant pairing $(\mathcal{P}_{\mathcal{H}},\,\omega^{\mbox{\tiny $(1)$}}(\mathcal{Y}_{\mathcal{H}},\,\mathcal{P}_{\mathcal{H}}))$ is covariant triangulated by the covariant forms $\omega^{\mbox{\tiny $(1)$}}(\mathcal{Y}_{\mathcal{H}},\,\mathcal{P}^{\mbox{\tiny $(65)$}})$ and $\omega^{\mbox{\tiny $(1)$}}(\mathcal{Y}_{\mathcal{H}},\,\mathcal{P}_2^{\mbox{\tiny $(46)$}})$:
\begin{equation}\label{eq:P3chHo2}
 \begin{split}
    \omega(\mathcal{Y}_{\mathcal{H}},\,\mathcal{P}_\mathcal{H})\:&=\:\omega(\mathcal{Y}_{\mathcal{H}},\,\mathcal{P}_{2}^{\mbox{\tiny $(65)$}})+
                                                                     \omega(\mathcal{Y}_{\mathcal{H}},\,\mathcal{P}_{2}^{\mbox{\tiny $(46)$}})\:=\\
    &=\:\frac{\langle Z_{\star}65\rangle\langle Z_{\star}\mathcal{Y}_{\mathcal{H}}d\mathcal{Y}_{\mathcal{H}}\rangle}{
         \langle\mathcal{Y}_{\mathcal{H}}6\rangle^2\langle\mathcal{Y}_{\mathcal{H}}5\rangle}+
        \frac{\langle Z_{\star}46\rangle\langle Z_{\star}\mathcal{Y}_{\mathcal{H}}d\mathcal{Y}_{\mathcal{H}}\rangle}{
          \langle\mathcal{Y}_{\mathcal{H}}6\rangle^2\langle\mathcal{Y}_{\mathcal{H}}4\rangle},
 \end {split}                                                            
\end{equation}
{\it i.e.} $\mathcal{P}_{\mathcal{H}}$ is triangulated via an external point, and the covariant form in covariant pairing with it is the sum of the covariant forms associated with the two segments in the signed triangulation via the external point which now have poles only in the boundaries of the associated polytope, with a double pole in the common boundary which get lowered to simple pole upon summation. Here we see the phenomenon of how the covariant forms which are elements of a covariant form triangulation shows a certain multiplicity of a pole We can understand \eqref{eq:P3chHo2} also from the perspective of the parent polytope. The parent polytope is a triangle with vertices $(123)$ which can be triangulated via the external point $Z_6$ into $(123)\:=\:(163)+(236)$, and its canonical form can be written as sum of the canonical forms of $(163)$ and $(236)$. 
\begin{wrapfigure}{l}{4.75cm}
 \centering
 \begin{tikzpicture}[line join = round, line cap = round, ball/.style = {circle, draw, align=center, anchor=north, inner sep=0}, 
                     axis/.style={very thick, ->, >=stealth'}, pile/.style={thick, ->, >=stealth', shorten <=2pt, shorten>=2pt}, every node/.style={color=black}]
 
 \begin{scope}
  \coordinate [label=above:{\footnotesize $\displaystyle 1$}] (A) at (0,0);
  \coordinate [label=left:{\footnotesize $\displaystyle 2$}] (B) at (-1.75,-2.25);
  \coordinate [label=right:{\footnotesize $\displaystyle 3$}] (C) at (+1.75,-2.25);
  
  \draw[-, fill=blue!30, opacity=.4] (A) -- (B) -- (C) -- cycle;  
  
  \coordinate [label=below:{\footnotesize $\displaystyle 4$}] (Z4) at ($(B)!0.25!(C)$);
  \coordinate [label=below:{\footnotesize $\displaystyle 5$}] (Z5) at ($(C)!0.375!(A)$);
  \coordinate [label=above:{\footnotesize $\displaystyle 6$}] (Z6) at (intersection of A--B and Z4--Z5);
  
  \draw[-, fill=yellow!30, opacity=.7] (A) -- (Z6) -- (C) -- cycle;

  \coordinate [label=above:{\small $\displaystyle \mathcal{H}$}] (H) at ($(Z5)!-0.625cm!(Z4)$);
  \coordinate [label=left:{\small $\displaystyle \mathcal{P}_3$}] (P3) at ($(A)!0.375!(B)$);
  \draw[color=red, thick] ($(Z5)!-0.625cm!(Z6)$) -- ($(Z6)!-0.625cm!(Z5)$);
  \draw ($(A)!-0.625cm!(Z6)$) -- ($(Z6)!-0.625cm!(A)$);
  
  \draw[color=red, very thick] (Z4) -- (Z5);
  \draw[fill=red, color=red] (Z4) circle (1pt);
- \draw[fill=red, color=red] (Z5) circle (1pt);
  \draw[fill=red, color=blue] (Z6) circle (1pt);
 \end{scope}
 \begin{scope}[shift={(0,-3)}, transform shape]
  \coordinate [label=above:{\footnotesize $\displaystyle 1$}] (A) at (0,0);
  \coordinate [label=left:{\footnotesize $\displaystyle 2$}] (B) at (-1.75,-2.25);
  \coordinate [label=right:{\footnotesize $\displaystyle 3$}] (C) at (+1.75,-2.25);
  
  \draw[-, fill=blue!30, opacity=.4] (A) -- (B) -- (C) -- cycle;  
  
  \coordinate [label=below:{\footnotesize $\displaystyle 4$}] (Z4) at ($(B)!0.25!(C)$);
  \coordinate [label=below:{\footnotesize $\displaystyle 5$}] (Z5) at ($(C)!0.375!(A)$);
  \coordinate [label=above:{\footnotesize $\displaystyle 6$}] (Z6) at (intersection of A--B and Z4--Z5);
  
  \draw[-, fill=red!30, opacity=.7] (B) -- (Z6) -- (C) -- cycle;

  \coordinate [label=above:{\small $\displaystyle \mathcal{H}$}] (H) at ($(Z5)!-0.625cm!(Z4)$);
  \coordinate [label=left:{\small $\displaystyle \mathcal{P}_3$}] (P3) at ($(A)!0.375!(B)$);
  \draw[color=red, thick] ($(Z5)!-0.625cm!(Z6)$) -- ($(Z6)!-0.625cm!(Z5)$);
  \draw ($(A)!-0.625cm!(Z6)$) -- ($(Z6)!-0.625cm!(A)$);
  
  \draw[color=red, very thick] (Z4) -- (Z5);
  \draw[fill=red, color=red] (Z4) circle (1pt);
- \draw[fill=red, color=red] (Z5) circle (1pt);
  \draw[fill=red, color=blue] (Z6) circle (1pt);
 \end{scope}
 \end{tikzpicture}
\end{wrapfigure}

The hyperline $\mathcal{H}$ defined as \eqref{eq:H45s} intersects both such triangles inside only, and hence the covariant restriction of their canonical form onto it is as the example discussed in the previous subsection: upon the restriction, these triangles are mapped into segments and the related covariant forms show a double pole at the boundary of the segment where two facets of the parent polytope intersect (in other words, two codimension-$1$ boundary of the parent polytope reduce to the same codimension-$1$ boundary of the child polytope). The covariant forms in the second line of \eqref{eq:P3chHo2} are exactly the restriction of the canonical forms of the triangles $(163)$ and $(236)$, which are mapped to the segments $(65)$ and $(46)$ respectively. The fact that such segments share a boundary manifests itself in the lower multiplicity of the related pole.
\\

The same happens for any other polygon: its restriction on a hyper-line is still a segment which can be decomposed into a union (triangulations) of segments each of which is the restriction of the terms of the triangulation of the parent polytope. From the perspective of the covariant form in $\mathbb{P}^1$ obtained as covariant restriction of the canonical form of the parent polytope, each covariant form obtained from a single term in the triangulation of the parent polytope has simple poles at those boundaries in $\mathbb{P}^1$ identified by the intersection of $\mathcal{H}$ with a single facet of $\mathcal{P}_n$ and a double pole if the boundary in $\mathbb{P}^1$ is identified by the intersection between $\mathcal{H}$ and two of the facets of $\mathcal{P}_n$. If the double poles are related to facets which are common to two segments, then it will become a single pole upon summation of all the covariant forms, while if it is a simple pole, it will become spurious. Let us briefly discuss it for some of the cases depicted in Figure \ref{fig:PolygH}.

Let us consider a square and a hyperline which intersects its facets outside just in one point. The hyperplane $\mathcal{H}$ intersects the facets $(12)$ and $(23)$ in the same point $Z_2$, while intersects the facet $(41)$ in $Z_5\,\sim\,\alpha Z_4+(1-\alpha)Z_1$ which lies between the vertices $Z_4$ and $Z_1$ ({\it i.e.} $\alpha\,\in\,]0,\,1[$), and the facet $(34)$ outside, in $Z_6\,=\,(25)\cap(34)$. Hence, we can already expect that, upon the covariant restriction on $\mathcal{H}$, the canonical form of the square gets mapped into a covariant form of degree-$1$ with a double and two single poles. 

\begin{wrapfigure}{l}{4.75cm}
 \centering
 \begin{tikzpicture}[line join = round, line cap = round, ball/.style = {circle, draw, align=center, anchor=north, inner sep=0}, 
                     axis/.style={very thick, ->, >=stealth'}, pile/.style={thick, ->, >=stealth', shorten <=2pt, shorten>=2pt}, every node/.style={color=black}]
  \begin{scope}
   \coordinate [label=above:{\footnotesize $\displaystyle 1$}] (A) at (0,0);
   \coordinate [label=left:{\footnotesize $\displaystyle 2$}] (B) at (-.5,-2.25);
   \coordinate [label=right:{\footnotesize $\displaystyle 3$}] (C) at (2,-2);
   \coordinate [label=right:{\footnotesize $\displaystyle 4$}] (D) at (1.75,-.25); 
 
   \draw[-, fill=blue!30, opacity=.4] (A) -- (B) -- (C) -- (D) -- cycle;
   \coordinate [label=above:{\footnotesize $\displaystyle 5$}] (Z5) at ($(D)!0.25!(A)$);
   \coordinate [label=right:{\footnotesize $\displaystyle 6$}] (Z6) at (intersection of D--C and B--Z5);
   \draw[-, fill=yellow!30, opacity=.7] (A) -- (B) -- (C) -- (Z6) -- cycle;
   \coordinate [label=above:{\small $\displaystyle \mathcal{H}$}] (H) at ($(Z5)!-0.625cm!(B)$);
   \coordinate [label=left:{\small $\displaystyle \mathcal{P}_4$}] (P4) at ($(A)!0.5!(B)$);
   \draw[color=red, thick] ($(Z5)!-0.625cm!(B)$) -- ($(B)!-0.625cm!(Z5)$);
   \draw[color=red, very thick] (Z5) -- (B);
   \draw ($(C)!-0.625cm!(Z6)$) -- ($(Z6)!-0.625cm!(C)$);
   \draw[fill=red, color=red] (Z5) circle (1pt);
 - \draw[fill=red, color=red] (B) circle (1pt);
  \draw[fill=red, color=blue] (Z6) circle (1pt);
 \end{scope}
 \begin{scope}[shift={(0,-3.5)}, transform shape]
   \coordinate [label=above:{\footnotesize $\displaystyle 1$}] (A) at (0,0);
   \coordinate [label=left:{\footnotesize $\displaystyle 2$}] (B) at (-.5,-2.25);
   \coordinate [label=right:{\footnotesize $\displaystyle 3$}] (C) at (2,-2);
   \coordinate [label=right:{\footnotesize $\displaystyle 4$}] (D) at (1.75,-.25); 
 
   \draw[-, fill=blue!30, opacity=.4] (A) -- (B) -- (C) -- (D) -- cycle;
   \coordinate [label=above:{\footnotesize $\displaystyle 5$}] (Z5) at ($(D)!0.25!(A)$);
   \coordinate [label=right:{\footnotesize $\displaystyle 6$}] (Z6) at (intersection of D--C and B--Z5);
   \draw[-, fill=red!30, opacity=.7] (A) -- (D) -- (Z6) -- cycle;
   \coordinate [label=above:{\small $\displaystyle \mathcal{H}$}] (H) at ($(Z5)!-0.625cm!(B)$);
   \coordinate [label=left:{\small $\displaystyle \mathcal{P}_4$}] (P4) at ($(A)!0.5!(B)$);
   \draw[color=red, thick] ($(Z5)!-0.625cm!(B)$) -- ($(B)!-0.625cm!(Z5)$);
   \draw[color=red, very thick] (Z5) -- (B);
   \draw ($(C)!-0.625cm!(Z6)$) -- ($(Z6)!-0.625cm!(C)$);
   \draw[fill=red, color=red] (Z5) circle (1pt);
 - \draw[fill=red, color=red] (B) circle (1pt);
  \draw[fill=red, color=blue] (Z6) circle (1pt);
 \end{scope}
 \end{tikzpicture}
\end{wrapfigure}

 Such a covariant form is associated to the {\it tangent union} of the segments $(\mathbb{P}^1,\,\mathcal{P}_2^{\mbox{\tiny $(26)$}})$ and $(\mathbb{P}^1,\,\mathcal{P}_2^{\mbox{\tiny $(65)$}})$. Let us take again the perspective of the parent polytope. We can see it as a triangulation $(1234)\,=\,(1236)+(641)$ through the external point $Z_6\,=\,(25)\cap(34)$, which decompose it into a square and a triangle. The hyperplane $\mathcal{H}$ intersects both terms of this triangulation just on their facets, so that the covariant restriction of the canonical form of the square $(1235)$ generated a covariant form of degree-$1$ associated to the segment $(\mathbb{P}^1,\,\mathcal{P}_2^{\mbox{\tiny $(26)$}})$ with two double poles (both boundary components of the segment arise from the intersection of two facets of the parent polytope on the same point on $\mathcal{H}$), while the covariant restriction of the canonical form of the triangle $(641)$ generates a covariant form of degree-$1$ associated to the segment $(\mathbb{P}^1,\,\mathcal{P}_2^{\mbox{\tiny $(65)$}})$ with a double and a single pole. The common boundary component between $(\mathbb{P}^1,\,\mathcal{P}_2^{\mbox{\tiny $(26)$}})$ and $(\mathbb{P}^1,\,\mathcal{P}_2^{\mbox{\tiny $(65)$}})$ is identified by a double pole in both the covariant form and, because of the orientation inherited from the triangulation of the parent polytope, it becomes a single pole upon their summation. Explicitly
\begin{equation}\label{eq:P225cf}
 \begin{split}
    \omega^{\mbox{\tiny $(1)$}}(\mathcal{Y}_{\mathcal{H}},\,\mathcal{P}_{\mathcal{H}})\:&=\:
     \frac{\langle\mathcal{Y}_{\mathcal{H}}Z_{\star}Z_{\circ}\rangle\langle Z_{\star}\mathcal{Y}_{\mathcal{H}}d\mathcal{Y}_{\mathcal{H}}\rangle}{
      \langle\mathcal{Y}_{\mathcal{H}}Z_{\star}2\rangle^2\langle\mathcal{Y}_{\mathcal{H}}Z_{\star}5\rangle\langle\mathcal{Y}_{\mathcal{H}}Z_{\star}6\rangle}\:=\\
     &=\:\frac{\langle\mathcal{Y}_{\mathcal{H}}Z_{\star}\tilde{Z}_{\circ}\rangle\langle Z_{\star}\mathcal{Y}_{\mathcal{H}}d\mathcal{Y}_{\mathcal{H}}\rangle}{
          \langle\mathcal{Y}_{\mathcal{H}}Z_{\star}2\rangle^2\langle\mathcal{Y}_{\mathcal{H}}Z_{\star}6\rangle^2}+
         \frac{\langle Z_{\star}65\rangle\langle Z_{\star}\mathcal{Y}_{\mathcal{H}}d\mathcal{Y}_{\mathcal{H}}\rangle}{\langle\mathcal{Y}_{\mathcal{H}}Z_{\star}6\rangle^2\langle\mathcal{Y}_{\mathcal{H}}Z_{\star}5\rangle}\:=\\
     &=\:\omega^{\mbox{\tiny $(1)$}}(\mathcal{Y}_{\mathcal{H}},\,\mathcal{P}_2^{\mbox{\tiny $(26)$}}) + \omega^{\mbox{\tiny $(1)$}}(\mathcal{Y}_{\mathcal{H}},\,\mathcal{P}_2^{\mbox{\tiny $(65)$}}),
 \end{split}
\end{equation}
where $\mathcal{P}_{\mathcal{H}}\:=\:\mathcal{P}_2^{\mbox{\tiny $(26)$}}\,\cup\,\mathcal{P}_2^{\mbox{\tiny $(65)$}}$, $Z_{\circ}\,=\,(Z_{13}Z_{24})\cap\mathcal{H}$, and $\tilde{Z}_{\circ}\,=\,(Z_{13}Z_{25})\cap\mathcal{H}$.
\\



\begin{wrapfigure}{l}{5.75cm}
 \centering
 \begin{tikzpicture}[line join = round, line cap = round, ball/.style = {circle, draw, align=center, anchor=north, inner sep=0}, 
                     axis/.style={very thick, ->, >=stealth'}, pile/.style={thick, ->, >=stealth', shorten <=2pt, shorten>=2pt}, every node/.style={color=black}, scale=.85]
   \begin{scope}
    \coordinate [label=above:{\footnotesize $\displaystyle 1$}] (A) at (0,0);
    \coordinate [label=left:{\footnotesize $\displaystyle 3$}] (B) at (-.5,-2.25);
    \coordinate [label=left:{\footnotesize $\displaystyle 2$}] (E) at ($(A)!0.5!(B)+(-.25,+.125)$);
    \coordinate [label=right:{\footnotesize $\displaystyle 4$}] (C) at (2,-2);
    \coordinate [label=right:{\footnotesize $\displaystyle 5$}] (D) at (1.75,-.25);

    \draw[-, fill=blue!30, opacity=.6] (A) -- (E) -- (B) -- (C) -- (D) -- cycle; 
    \coordinate [label=left:{\footnotesize $\displaystyle 7$}] (Z7) at ($(E)!.675!(B)$); 
    \coordinate [label=right:{\footnotesize $\displaystyle 9$}] (Z9) at ($(C)!.75!(D)$);
    \coordinate [label=left:{\footnotesize $\displaystyle 10$}] (Z10) at (intersection of A--E and Z9--Z7);
    \coordinate [label=below:{\footnotesize $\displaystyle 8$}] (Z8) at (intersection of B--C and Z9--Z7);
    \coordinate [label=above:{\footnotesize $\displaystyle 6$}] (Z6) at (intersection of A--D and Z9--Z7);
    \draw[-, fill=yellow!30, opacity=.7] (A) -- (Z8) -- (C) -- (Z6) -- cycle;
    \coordinate [label=above:{\small $\displaystyle \mathcal{H}$}] (H) at ($(Z6)!-0.625cm!(Z9)$);
    \coordinate [label=right:{\small $\displaystyle \mathcal{P}_5$}] (P5) at ($(C)!0.5!(D)$);
    \draw[color=red, thick] ($(Z6)!-0.625cm!(Z8)$) -- ($(Z8)!-0.625cm!(Z6)$);
    \draw[color=red, very thick] (Z9) -- (Z7);
    \draw ($(A)!-0.625cm!(Z6)$) -- ($(Z6)!-0.625cm!(A)$);
    \draw ($(C)!-0.625cm!(Z8)$) -- ($(Z8)!-0.625cm!(C)$);
    \draw ($(A)!-0.625cm!(Z6)$) -- ($(Z6)!-0.625cm!(A)$);
    \draw ($(A)!-0.625cm!(Z10)$) -- ($(Z10)!-0.75cm!(A)$);
    \draw[fill=red, color=red] (Z9) circle (1pt);
-   \draw[fill=red, color=red] (Z7) circle (1pt);
    \draw[fill=red, color=blue] (Z6) circle (1pt);
    \draw[fill=red, color=blue] (Z8) circle (1pt);
    \draw[fill=red, color=blue] (Z10) circle (1pt);
  \end{scope}
  \begin{scope}[shift={(0,-4.5)}, transform shape]
    \coordinate [label=above:{\footnotesize $\displaystyle 1$}] (A) at (0,0);
    \coordinate [label=left:{\footnotesize $\displaystyle 3$}] (B) at (-.5,-2.25);
    \coordinate [label=left:{\footnotesize $\displaystyle 2$}] (E) at ($(A)!0.5!(B)+(-.25,+.125)$);
    \coordinate [label=right:{\footnotesize $\displaystyle 4$}] (C) at (2,-2);
    \coordinate [label=right:{\footnotesize $\displaystyle 5$}] (D) at (1.75,-.25);

    \draw[-, fill=blue!30, opacity=.6] (A) -- (E) -- (B) -- (C) -- (D) -- cycle; 
    \coordinate [label=left:{\footnotesize $\displaystyle 7$}] (Z7) at ($(E)!.675!(B)$); 
    \coordinate [label=right:{\footnotesize $\displaystyle 9$}] (Z9) at ($(C)!.75!(D)$);
    \coordinate [label=left:{\footnotesize $\displaystyle 10$}] (Z10) at (intersection of A--E and Z9--Z7);
    \coordinate [label=below:{\footnotesize $\displaystyle 8$}] (Z8) at (intersection of B--C and Z9--Z7);
    \coordinate [label=above:{\footnotesize $\displaystyle 6$}] (Z6) at (intersection of A--D and Z9--Z7);
    \draw[-, fill=red!30, opacity=.7] (D) -- (C) -- (Z6) -- cycle;
    \draw[-, fill=green!30, opacity=.7] (E) -- (B) -- (Z8) -- cycle;
    \draw[-, fill=orange!30, opacity=.7] (A) -- (E) -- (Z8) -- cycle;
    \coordinate [label=above:{\small $\displaystyle \mathcal{H}$}] (H) at ($(Z6)!-0.625cm!(Z9)$);
    \coordinate [label=right:{\small $\displaystyle \mathcal{P}_5$}] (P5) at ($(C)!0.5!(D)$);
    \draw[color=red, thick] ($(Z6)!-0.625cm!(Z8)$) -- ($(Z8)!-0.625cm!(Z6)$);
    \draw[color=red, very thick] (Z9) -- (Z7);
    \draw ($(A)!-0.625cm!(Z6)$) -- ($(Z6)!-0.625cm!(A)$);
    \draw ($(C)!-0.625cm!(Z8)$) -- ($(Z8)!-0.625cm!(C)$);
    \draw ($(A)!-0.625cm!(Z6)$) -- ($(Z6)!-0.625cm!(A)$);
    \draw ($(A)!-0.625cm!(Z10)$) -- ($(Z10)!-0.75cm!(A)$);
    \draw[fill=red, color=red] (Z9) circle (1pt);
-   \draw[fill=red, color=red] (Z7) circle (1pt);
    \draw[fill=red, color=blue] (Z6) circle (1pt);
    \draw[fill=red, color=blue] (Z8) circle (1pt);
    \draw[fill=red, color=blue] (Z10) circle (1pt);
  \end{scope}
 \end{tikzpicture}
\end{wrapfigure}

As a final illustrative example in $\mathbb{P}^2$, let us consider a pentagon intersected in his facets by $\mathcal{H}$ in three external points.

The hyperline $\mathcal{H}$ intersects the convex hull $\mathcal{P}_5$ in its facets inside on $Z_7\,\sim\,\alpha Z_2-(1-\alpha)Z_3$ and $Z_9\,\sim\,\beta Z_4+(1-\beta)Z_5$, with $\alpha,\,\beta\,\in\,]0,\,1[$, while it intersects the facets outside of $\mathcal{P}_5$ in
$Z_6\,=\,(51)\cap\mathcal{H}$, $Z_8\,=\,(34)\cap\mathcal{H}$, $Z_{10}\,=\,(12)\cap\mathcal{H}$: all the facets of $\mathcal{P}_5$ are projected on different points upon the restriction on the hyperline $\mathcal{H}$ and, consequently, they will be reflected on a single pole each in the covariant form of degree-$1$ $\omega^{\mbox{\tiny $(k)$}}(\mathcal{Y}_{\mathcal{H}},\,\mathcal{P}_{\mathcal{H}})$ obtained via \eqref{eq:CPcf} from the canonical form of $(\mathbb{P}^2,\,\mathcal{P}_5)$. Such a covariant form can be understood as a sum of the covariant forms associated to the segments $(\mathbb{P}^1,\,\mathcal{P}^{\mbox{\tiny $(86)$}})$, $(\mathbb{P}^1,\,\mathcal{P}^{\mbox{\tiny $(78)$}})$, $(\mathbb{P}^1,\,\mathcal{P}^{\mbox{\tiny $(8,10)$}})$, $(\mathbb{P}^1,\,\mathcal{P}^{\mbox{\tiny $(69)$}})$. 
\vspace{.1cm}

\noindent
From the perspective of the parent polytope, this sum comes from the (signed) triangulation of $\mathcal{P}_5$ as
\begin{equation}\label{eq:P5tr}
    \mathcal{P}_5\:=\:(1846)\,\cup\,(238)\,\cup\,(821)\,\cup\,(645).
\end{equation}
Upon the covariant restriction \eqref{eq:CPcf} of the canonical form of each of the terms in \eqref{eq:P5tr} on $\mathcal{H}$, one obtains the covariant form for the segments $(\mathbb{P}^1,\,\mathcal{P}^{\mbox{\tiny $(86)$}})$, $(\mathbb{P}^1,\,\mathcal{P}^{\mbox{\tiny $(78)$}})$, $(\mathbb{P}^1,\,\mathcal{P}^{\mbox{\tiny $(8,10)$}})$, $(\mathbb{P}^1,\,\mathcal{P}^{\mbox{\tiny $(69)$}})$ respectively, which are characterised by having a double pole in the two boundaries identified by the vertices $Z_8$ and $Z_6$: these are the only two certices on which two facets of the parent polytope are restricted. Upon the summation such double poles are lowered to simple poles.


\subsubsection{Polyhedra with Internal Intersections}\label{subsubsec:PhII}

Let us now discuss some example in $\mathbb{P}^3$: the general relation between high order poles in the covariant forms of the child polytope and the number of facets intersecting each other in the lower dimensional hypersurface does not change, but it is instructive to see the restriction at work for examples other than $\mathbb{P}^2$.

The simplest example is given by a tetrahedron $(\mathbb{P}^3,\,\mathcal{P})$ and a hyperplane such that $\mathcal{P}\cap\mathcal{H}$ is a triangle with two vertices being vertices of $\mathcal{P}$ and the third one lying on one of its edges

\begin{wrapfigure}{l}{5.5cm}
  \centering
  \begin{tikzpicture}[line join = round, line cap = round, ball/.style = {circle, draw, align=center, anchor=north, inner sep=0}, 
                     axis/.style={very thick, ->, >=stealth'}, pile/.style={thick, ->, >=stealth', shorten <=2pt, shorten>=2pt}, every node/.style={color=black}]
  \begin{scope}
   \pgfmathsetmacro{\factor}{1/sqrt(2)};  
   \coordinate[label=right:{$\mathbf{4}$}] (B2) at (1.5,-3,-1.5*\factor);
   \coordinate[label=left:{$\mathbf{1}$}] (A1) at (-1.5,-3,-1.5*\factor);
   \coordinate[label=right:{$\mathbf{3}$}] (B1) at (1.5,-3.75,1.5*\factor);
   \coordinate[label=above:{$\mathbf{2}$}] (C1) at (0.75,-.65,.75*\factor);
   \coordinate[label=left:{$\mathbf{5}$}] (Z5) at ($(A1)!.375!(B1)$);
   \coordinate (Z6) at ($(Z5)+(0,2.15,0)$);
   \coordinate (t1) at ($(Z5)!-.625cm!(B2)$);
   \coordinate (t2) at ($(B2)!-.625cm!(Z5)$);
   \coordinate[label=right:{$\mathcal{H}$}] (t3) at ($(C1)!-1.75cm!(Z6)$);
   \coordinate (t4) at ($(Z6)!-.75cm!(C1)$);
   
   \draw[draw=none,fill=green!80,opacity=.3]  (B1) -- (B2) -- (A1) -- cycle;
   \draw[draw=none,fill=blue!60, opacity=.45] (C1) -- (B2) -- (A1) -- cycle;
   \draw[draw=none,fill=blue!70, opacity=.3] (C1) -- (A1) -- (B1) -- cycle;
   \draw[draw=none,fill=blue!70, opacity=.5] (C1) -- (B2) -- (B1) -- cycle;  

  
  \draw[draw=none, fill=red!40, opacity=.2] (t1) -- (t2) -- (t3) -- (t4) -- cycle;
  \draw[draw=none,fill=red!40, opacity=.5] (C1) -- (B2) -- (Z5) -- cycle;

  \end{scope}
 \end{tikzpicture}
\end{wrapfigure}

Taking the labeling of the vertices of $\mathcal{P}$ as in the picture here on the left, the hyperplane $\mathcal{H}$ is identified by 
\begin{equation}\label{eq:Htetrd}
    \mathcal{H}\:=\:\left\{\mathcal{Y}\,\in\,\mathbb{P}^3\,|\,\langle\mathcal{Y}245\rangle\,=\,0\right\},
\end{equation}
with $Z_5\,\sim\,\alpha Z_1+(1-\alpha)Z_3$ ($\alpha\,\in\,]0,1[$), and the child polytope is $(\mathbb{P}^2,\,\mathcal{H}\cap\mathcal{P})$, {\it i.e.} the triangle $\mathcal{P}_{\mathcal{H}}$ identified by the vertices $245$. The covariant restriction of the canonical form of $(\mathbb{P}^3,\,\mathcal{P})$ onto $\mathcal{H}$ is a covariant form of degree-$1$ with a double pole and two simple poles: the facets $(124)$ and $(234)$ intersect the hyperplane $\mathcal{H}$ in the same segment $(24)$ which is a codimension-$1$ boundary of the child polytope, while the other two facets of the parent polytope intersect $\mathcal{H}$ alone. Hence

{\small{
\begin{equation}\label{eq:CCFt}
    \omega(\mathcal{Y},\,\mathcal{P})\:=\:\frac{\langle1234\rangle^3\langle\mathcal{Y}d^3\mathcal{Y}\rangle}{\langle\mathcal{Y}123\rangle\langle\mathcal{Y}124\rangle\langle\mathcal{Y}234\rangle\langle\mathcal{Y}134\rangle}
    \quad\longrightarrow\quad
    \omega^{\mbox{\tiny $(1)$}}(\mathcal{Y}_{\mathcal{H}},\,\mathcal{P}_{\mathcal{H}})\:\sim\:\frac{\langle Z_{\star}234\rangle^2\langle Z_{\star}\mathcal{Y}_{\mathcal{H}}d^2\mathcal{Y}_{\mathcal{H}}\rangle}{
     \langle\mathcal{Y}_{\mathcal{H}}Z_{\star}23\rangle\langle\mathcal{Y}_{\mathcal{H}}Z_{\star}34\rangle\langle\mathcal{Y}_{\mathcal{H}}Z_{\star}42\rangle^2}
\end{equation}
}}
with $Z_{\star}$ indicating the orthogonal complement of $\mathcal{H}$.
\\

Let us now look at a slightly different example, considering $(\mathbb{P}^3,\,\mathcal{P})$ as a square bipyramid with an hyperplane $\mathcal{H}$ intersecting the convex hull $\mathcal{P}$ along the common basis of the two pyramids.

\begin{wrapfigure}{l}{5.5cm}
 \centering
 \begin{tikzpicture}[line join = round, line cap = round, ball/.style = {circle, draw, align=center, anchor=north, inner sep=0}, 
                     axis/.style={very thick, ->, >=stealth'}, pile/.style={thick, ->, >=stealth', shorten <=2pt, shorten>=2pt}, every node/.style={color=black}]
  \begin{scope}[scale=.8, transform shape]
   \pgfmathsetmacro{\factor}{1/sqrt(2)};  
   \coordinate[label=right:{$\mathbf{4}$}] (B2) at (1.5,-3,-1.5*\factor);
   \coordinate[label=left:{$\mathbf{5}$}]  (A1) at (-1.5,-3,-1.5*\factor);
   \coordinate[label=right:{$\mathbf{3}$}] (B1) at (1.5,-3.75,1.5*\factor);
   \coordinate[label=left:{$\mathbf{2}$}]  (A2) at (-1.5,-3.75,1.5*\factor);  
   \coordinate[label=above:{$\mathbf{1}$}]  (C1) at (0.75,-.65,.75*\factor);
   \coordinate[label=below:{$\mathbf{6}$}]  (C2) at (0.4,-6.05,.75*\factor);
   
   \coordinate (t1) at ($(A2)!-.75cm!(B1)$);
   \coordinate (t2) at ($(B1)!-.75cm!(A2)$);
   \coordinate (t3) at ($(B2)!-.75cm!(A1)$);
   \coordinate[label=above:{$\mathcal{H}$}] (t4) at ($(A1)!-.75cm!(B2)$);
   
   \draw[draw=none, fill=red!60, opacity=.2] (t1) -- (t2) -- (t3) -- (t4) -- cycle;
   \draw[draw=none, fill=red!60, opacity=.9] (A2) -- (B1) -- (B2) -- (A1) -- cycle;

   \draw[draw=none,fill=blue!30, opacity=.7] (A1) -- (B2) -- (C1) -- cycle;
   \draw[draw=none,thick,fill=blue!20, opacity=.7] (A1) -- (A2) -- (C1) -- cycle;
   \draw[draw=none,thick,fill=blue!20, opacity=.7] (B1) -- (B2) -- (C1) -- cycle;
   \draw[draw=none,thick,fill=blue!35, opacity=.7] (A2) -- (B1) -- (C1) -- cycle;

   \draw[draw=none,fill=blue!30, opacity=.3] (A1) -- (B2) -- (C2) -- cycle;
   \draw[draw=none, thick, fill=blue!50, opacity=.5] (B2) -- (B1) -- (C2) -- cycle;
   \draw[draw=none,fill=blue!40, opacity=.3] (A1) -- (A2) -- (C2) -- cycle;
   \draw[draw=none, thick, fill=blue!45, opacity=.5] (A2) -- (B1) -- (C2) -- cycle;
   
  \end{scope}
 \end{tikzpicture}
\end{wrapfigure}


The child polytope obtained as a restriction of $\mathcal{P}$ onto $\mathcal{H}$ is the square in $\mathbb{P}^2$ identified by the vertices $2345$. In this case, the facets of the bipyramid intersect $\mathcal{H}$ in pairs in the same segment: the codimension-$1$ boundaries of the parent polytope are mapped in pairs to the same codimension-$1$ boundary of the child polytope. Consequently, the canonical form of the parent polytope is mapped to a covariant form of degree-$1$ of the child polytope with double poles only, which inherits the structure of its zeros as well

\begin{equation}\label{eq:CFdpsq}
    \omega^{\mbox{\tiny $(1)$}}(\mathcal{Y}_{\mathcal{H}},\,\mathcal{P}_{\mathcal{H}})\:\sim\:\frac{\langle\mathcal{Y}Z_{\star}Z_{24}Z_{35}\rangle\langle Z_{\star}\mathcal{Y}_{\mathcal{H}}d^2\mathcal{Y}_{\mathcal{H}}\rangle}{
        \langle\mathcal{Y}_{\mathcal{H}}Z_{\star}23\rangle^2\langle\mathcal{Y}_{\mathcal{H}}Z_{\star}34\rangle^2\langle\mathcal{Y}_{\mathcal{H}}Z_{\star}45\rangle^2\langle\mathcal{Y}_{\mathcal{H}}Z_{\star}52\rangle^2},
\end{equation}
with $Z_{\star}$ being the orthogonal complement of $\mathcal{H}$.


\section{Jeffrey-Kirwan Residue and Covariant Forms}\label{subsec:jkex}
In this section we explain how covariant restrictions are relvant for the Jeffrey-Kirwan computation introduced in section \ref{subsec:PrPol}.

Let us consider the map $\tilde{Z}$ from $\mathbb{P}^{\nu-1}$ to $\mathbb{P}^N$: 
\begin{equation}\label{eq:mappolytopesimplex}
  \tilde{Z}: C \mapsto C \cdot Z=:\mathcal{Y},  
\end{equation}
where $C=(c_1,\ldots,c_{\nu})$ are homogeneous coordinates in $\mathbb{P}^{\nu}$ and $Z$ is $\nu \times (N+1)$ matrix.
Then the projective polytope defined in \eqref{eq:Polyt}, with vertices $Z_k,k=1,\ldots,\nu$ which are rows of the matrix Z, is just the image of the simplex $\Delta$ in $\mathbb{P}^{\nu-1}$ via the map $\tilde{Z}$. Let us now fix a point $\mathcal{Y}$ inside the polytope $\mathcal{P}$, and let us consider the fiber over $\mathcal{Y}$:
\begin{equation}
  \tilde{Z}^{-1}(\mathcal{Y})= \lbrace C \in \mathbb{P}^{\nu-1}: C \cdot Z= \mathcal{Y} \rbrace.
\end{equation}
Then the differential form $\tilde{\omega}_{\mathcal{Y}}(C,\mathcal{P})$ defined in Eq. \eqref{def:formfiber} is a top (covariant) differential form on the fiber $\tilde{Z}^{-1}(\mathcal{Y})$. In particular, it is the covariant restriction of the canonical form of the simplex $\omega(C,\Delta)$ into the hyperplane $\mathcal{H} \equiv \tilde{Z}^{-1}(\mathcal{Y})$, {\it i.e.}
\begin{equation} \label{eq:covformfiber}
 \tilde{\omega}_{\mathcal{Y}}(C,\mathcal{P}) \equiv \omega^{\mbox{\tiny $(N+1)$}}(C_{\mathcal{H}},\Delta_{\mathcal{H}}),
\end{equation}
where $\Delta_{\mathcal{H}}=\Delta \cap \mathcal{H}$.
Therefore, $\tilde{\omega}_{\mathcal{Y}}(C,\mathcal{P})$ has poles on the $(N-1)$-dimensional hyperplanes $\mathcal{H}_1,\ldots,\mathcal{H}_{\nu}$ which are the intersections between the $\nu$ facets of $\Delta_{\nu-1}$ and $\mathcal{H}$. In general\footnote{The set of $\mathcal{Y}$ in the interior of the polytope for which $\mathcal{H}$ doesn't intersect the simplex in lower dimensional faces is dense. Therefore these covariant restrictions in general do not produce higher order poles.} $\tilde{\omega}_{\mathcal{Y}}(C,\mathcal{P})$ has only simple poles, however it is not the canonical form of $\Delta_{\mathcal{H}}$. Indeed, some of the poles are on the intersection between the hyperplanes corresponding facets of $\Delta$ and $\mathcal{H}$ which lie outside $\Delta_{\mathcal{H}}$. Nevertheless, thanks to Theorem \ref{th:covpair}, the covariant form $\tilde{\omega}_{\mathcal{Y}}(C,\mathcal{P})$ is in covariant pairing with the child polytope $\Delta_{\mathcal{H}}$.

In full generality, by Theorem \ref{th:jk} one can compute the canonical function $\Omega(\mathcal{Y},\mathcal{P})$ of a polytope $\mathcal{P}$ in $\mathbb{P}^{N}$ with $\nu$ vertices (or the volume of the dual polytope $\tilde{P}$, see \eqref{eq:DPolyt}) by applying the Jeffrey-Kirwan residue to a covariant differential form $\omega^{(N)}(C_{\mathcal{H}},\Delta_{\mathcal{H}})$ in covariant pairing with the restriction of the standard simplex $\Delta$ in $\mathbb{P}^{\nu-1}$ onto hyperplanes $\mathcal{H} \equiv \tilde{Z}^{-1}(\mathcal{Y})$ of dimension $N$. 

Let us consider an easy visualisable example. Let $\mathcal{P}$ be the pentagon with vertices $Z_1,\ldots,Z_5 \in \mathbb{P}^2$.
The pentagon can be obtained as the image in $\mathbb{P}^2$ of the simplex $\Delta$ in $\mathbb{P}^4$ under the map in Eq. \eqref{eq:mappolytopesimplex}. We would like to compute the restriction of the canonical form of the simplex
\begin{equation}
    \omega(C,\Delta)=\bigwedge_{k=1}^5 \frac{d c_k}{c_k}
\end{equation}
onto the $2$-dimensional hyperplane
\begin{equation}
  \tilde{Z}^{-1}(\mathcal{Y})= \lbrace C \in \mathbb{P}^{4}: \sum_{k=1}^5 c_k Z_k= \mathcal{Y} \rbrace \equiv \mathcal{H}
\end{equation}
We can choose a parametrisation of $\tilde{Z}^{-1}(\mathcal{Y})$ using local inhomogeneous coordinates $(x_1,x_2,1) \in \mathbb{P}^2$ as:
\begin{equation}
    c_k= x \cdot Z^\perp_k+\tilde{c}(\mathcal{Y})_k, \quad k=1,\ldots,5.
\end{equation}
We denoted as $Z^\perp_k$ the columns of a $2 \times 5$ matrix orthogonal $Z$, i.e. $Z^\perp \cdot Z=0$ and $\tilde{C}(\mathcal{Y})$ is a particular solution of $\mathcal{Y}=\tilde{C}(\mathcal{Y}) \cdot Z$. For example:
\begin{equation}
\arraycolsep=1.4pt\def\arraystretch{1.5}
 Z^\perp=\left(\begin{array}{c@{}c}
& \mathbb{I}_{2} \begin{matrix}  -\frac{\langle 145 \rangle}{\langle 345 \rangle} & \frac{\langle 135 \rangle}{\langle 345 \rangle}& -\frac{\langle 134 \rangle}{\langle 345 \rangle}\\ 
  -\frac{\langle 245 \rangle}{\langle 345 \rangle} & \frac{\langle 235 \rangle}{\langle 345 \rangle}& -\frac{\langle 234 \rangle}{\langle 345 \rangle} \end{matrix} 
\end{array}\right), \quad \tilde{C}(\mathcal{Y})=\left(0,0,\frac{\langle \mathcal{Y}45 \rangle}{\langle 345 \rangle},-\frac{\langle \mathcal{Y}35 \rangle}{\langle 345 \rangle},\frac{\langle \mathcal{Y}34 \rangle}{\langle 345 \rangle} \right)
\end{equation}
Then we have: 
\begin{equation}
    C \cdot Z= (x \cdot Z^\perp+\tilde{C}(\mathcal{Y})) \cdot Z=\tilde{c}(\mathcal{Y}) \cdot Z=\mathcal{Y},
\end{equation}
and
\begin{equation}\label{es:restriction}
    \tilde{\omega}_{\mathcal{Y}}(C,\mathcal{P}) \sim \frac{d^2 x}{\prod_{k=1}^{5} (x \cdot Z^\perp_k+\tilde{c}_k(\mathcal{Y}))} 
\end{equation}
We notice that the intersections between the facets of the simplex $\lbrace c_k=0\rbrace$ and the hyperplane $\tilde{Z}^{-1}(\mathcal{Y})$ appear in the factors in the denominator of \eqref{es:restriction}. This phenomenon is exactly the one described in section \ref{subsubsec:POI}, where we considered cases in which the hyperplane can intersect the facets of the parent polytope outside. The child polytope $\Delta_{\mathcal{H}} \equiv \tilde{Z}^{-1}(\mathcal{Y}) \cap \Delta$ can be a triangle, a quadrilater or a pentagon, according to where $\mathcal{Y}$ is located in the pentagon $\mathcal{P}$. Nevertheless, by Theorem \ref{th:covpair} in all cases the child polytope is in covariant pairing with the differential form \eqref{es:restriction}, i.e. using the notation in section \ref{subsec:PrPol} we can write: 
\begin{equation}
\tilde{\omega}_{\mathcal{Y}}(x,\mathcal{P}) \equiv \omega^{\mbox{\tiny $(3)$}}(x,\Delta_{\mathcal{H}}).
\end{equation}

For completeness, we briefly show how to apply Jeffrey-Kirwan to the covariant form $\omega^{(3)}(x,\Delta_{\mathcal{H}})
$ in order to obtain triangulations for the pentagon $\mathcal{P}$. We will refer to section \ref{subsec:PrPol} for the notations used in the following.
With our choice of our inhomogeneous coordinates $y$, the cones $\mathfrak{C}_{k_1 k_2}$ are spanned by positive linear combinations of $\lbrace Z^\perp_{k_1},Z^\perp_{k_2}\rbrace$. We depict them in $\mathbb{R}^2$ in Fig. \ref{fig:coneschambers}. Let us now fix a vector $\xi \in \mathbb{P}^2$ as in Fig. \ref{fig:coneschambers} such that $\xi$ is in the chamber $\mathfrak{c}_1$. Then by definition in \eqref{def:jk} the Jeffrey-Kirwan residue is computed as:
\begin{equation}
  \mathrm{JK}_{\xi} \, \omega^{(3)}(x,\Delta_{\mathcal{H}}) = \sum_{\mathfrak{C}_I \ni \xi} \mbox{Res}_{\mathfrak{C}_I} \omega^{(3)}(x,\Delta_{\mathcal{H}}) = \left(\mbox{Res}_{\mathfrak{C}_{25}} +\mbox{Res}_{\mathfrak{C}_{45}} +\mbox{Res}_{\mathfrak{C}_{23}} \right) \omega^{(3)}(x,\Delta_{\mathcal{H}}),
\end{equation}
since $\xi$ is contained in the cones $\mathfrak{C}_{25},\mathfrak{C}_{45},\mathfrak{C}_{23}$. In this example, $\mathfrak{C}_{k_1 k_2}$ is positively oriented if $\mbox{det}(Z^\perp_{k_1} Z^\perp_{k_2})>0$. 
\begin{figure}[ht] 
\centering{
\def\svgwidth{0.6\linewidth}{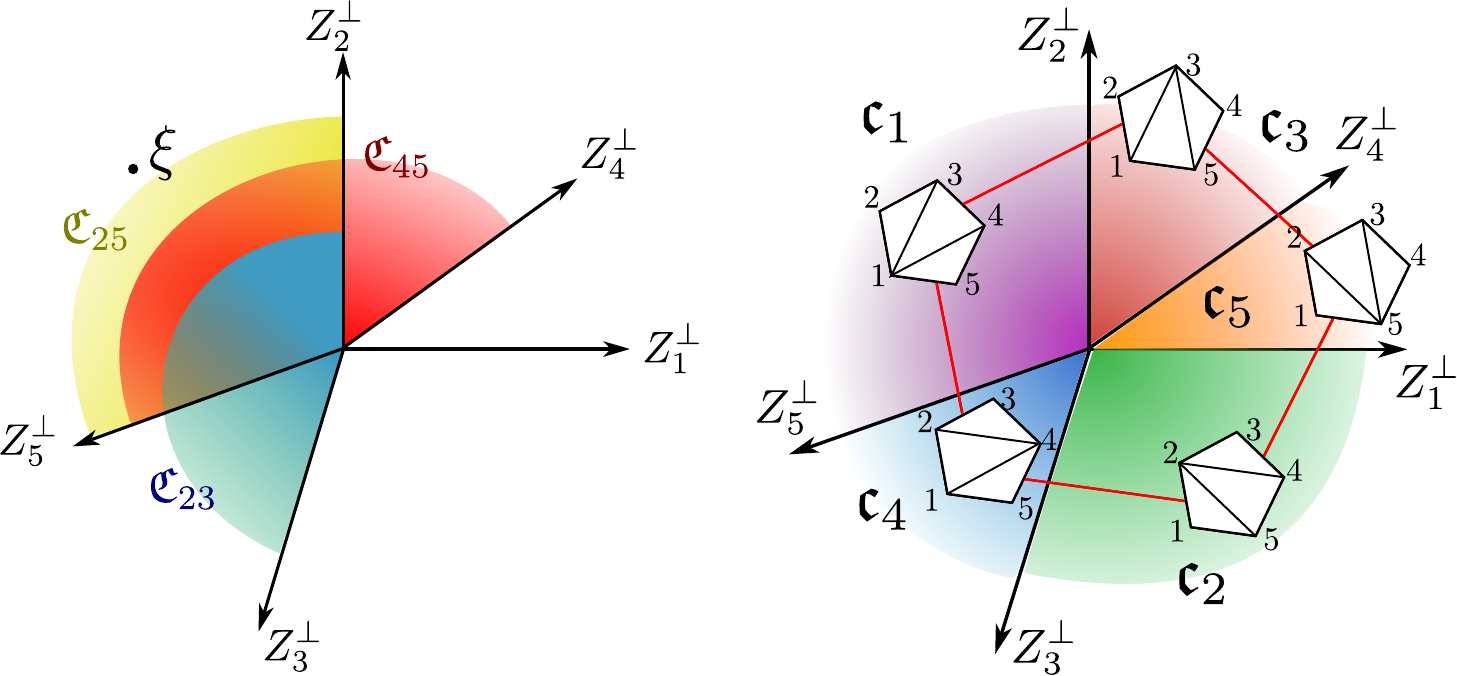}
}
\caption{Illustration of Cones and Chambers}
\label{fig:coneschambers}
\end{figure}

This produces the following representation of the canonical function of the pentagon:
\begin{equation}
  \mathrm{JK}_{\xi} \, \omega^{(3)}(x,\Delta_{\mathcal{H}}) = \Omega(\mathcal{Y},\Delta_{134})+  \Omega(\mathcal{Y},\Delta_{123})+\Omega(\mathcal{Y},\Delta_{145})=\Omega(\mathcal{Y},\mathcal{P}),
\end{equation}
where $\Delta_{k_1 k_2 k_3}$ are triangles in $\mathbb{P}^2$ with vertices $k_1,k_2,k_3$. Clearly, this corresponds to the triangulation of the pentagon into $\lbrace \Delta_{134},\Delta_{123},\Delta_{145}\rbrace$. All the other $4$ triangulations can be analogously obtained by choosing the reference vector $\xi$ in different chambers, see Fig.\ref{fig:coneschambers}.

\section{Cosmological Polytopes and Covariant Forms}\label{subsec:CPCF}

Let us now turn to the cosmological polytopes, which allows us to discuss higher dimensional examples. Recall that a cosmological polytope is constructed by taking a collection of triangles and segments, and intersecting them in the midpoints of their edges with the constraint that the triangles can be intersected on at most two out of its three sides. Using the notation introduced in Section \ref{subsec:CP}, we indicate with $\{\mathbf{x}_s\},\,\{\mathbf{y}_e\},\,\{\mathbf{h}_h\}$  the collection of vectors of the midpoints of the intersectable edges of the triangles and the segments, of the non-intersectable ones, and the non-intersectable vertex of the segments respectively and use it as a basis for the space where the cosmological polytope lives. Furthermore, they present natural hyperplanes on which the covariant restriction of their canonical form can be performed to produce covariant forms. 

\begin{prop}\label{prop:CPwk}
 Let $(\mathbb{P}^{n_e+n_s-1},\,\mathcal{P})$ be a cosmological polytope constructed from a collection of $n_t$ triangles and $n_h$ segments. Let $\mathcal{G}$ be the associated graph with $n_s$ sites and $n_e$ edges of which $n_h$ are tadpoles subgraphs. Let $k\,\in\,[1,\,n_h]$ be an integer, then if 
 \begin{equation}\label{eq:HK}
  \resizebox{0.9\hsize}{!}{$\displaystyle%
   \mathcal{H}^{\mbox{\tiny $(k)$}}\,=\,
     \left\{
       \mathcal{Y}\,\in\,\mathbb{P}^{n_e+n_s-1}\,\bigg|\,\mathcal{Y}\cdot\mathbf{\tilde{h}}_l,\,=0,\:
               \left\{
                \begin{array}{l}
                     \mathbf{h}_h\cdot\mathbf{\tilde{h}}_l\,=\,\delta_{hl},\hspace{.675cm}\forall\:l\,\in\,[1,k],\,h\,\in\,[1,\,n_h]  \\
                     (\mathbf{x}_s,\,\mathbf{y}_e)\cdot\mathbf{\tilde{h}}_l\,=\,0,\:\forall\:s\,\in\,[1,\,n_s],\,e\,\in\,[1,\,n_e]
                \end{array}
               \right.
       \right\}
   $
  }
 \end{equation}
 the restriction $(\mathbb{P}^{n_s+n_e-k-1},\,\mathcal{P}_{\mathcal{H}^{\mbox{\tiny $(k)$}}})$, with $\mathcal{P}_{\mathcal{H}^{\mbox{\tiny $(k)$}}}\, :=\,\mathcal{P}\,\cap\,\mathcal{H}^{\mbox{\tiny $(k)$}}$ of the cosmological polytope $(\mathbb{P}^{n_e+n_s-1},\,\mathcal{P})$ onto $\mathcal{H}^{\mbox{\tiny $(k)$}}$ is still a cosmological polytope whose associated $\mathcal{G}_{\mathcal{H}^{\mbox{\tiny $(k)$}}}$ is obtained from $\mathcal{G}$ suppressing $k$ tadpoles, and the covariant restriction of the canonical form of $(\mathbb{P}^{n_e+n_s-1},\,\mathcal{P})$ is a covariant form of degree-$k$ associated to $(\mathbb{P}^{n_s+n_e-k-1},\,\mathcal{P}_{\mathcal{H}^{\mbox{\tiny $(k)$}}})$.
\end{prop}

\begin{proof}
 Let us consider the cosmological polytope $(\mathbb{P}^{n_e+n_s-1},\,\mathcal{P})$ and let $\mathcal{G}$ be its associated graph. By definition, $\mathcal{P}$ is the covenx hull of the collection of vertices of $n_t$ triangles and $n_h$ segments with suitable identifications $\{\mathbf{x}_a\,=\,\mathbf{x}_b\}$ of the midpoints of triangles and segments, with the prescription that each triangle can be intersected on the midpoints at most two of its three sides. Hence, the vertices of $\mathcal{P}$ have the form (modulo midpoint identifications)
 \begin{equation*}
     \{\mathbf{x}_s-\mathbf{y}_e+\mathbf{x}_{s'},\,\mathbf{x}_s+\mathbf{y}_e-\mathbf{x}_{s'},\,-\mathbf{x}_s+\mathbf{y}_e+\mathbf{x}_{s'}\},\qquad
     \{2\mathbf{x}_{s''}-\mathbf{h}_{s''},\,\mathbf{h}_{s''}\}
 \end{equation*}
 with the two collections being the vertices of the triangles and segments respectively. Because of the definition \eqref{eq:HK} of $\mathcal{H}^{\mbox{\tiny $(k)$}}$, the vertices of the polytope which are on $\mathcal{H}^{\mbox{\tiny $(k)$}}$ are all vertices of the generating triangles, and all those vertices of the segments such that $s''\,\neq\,l,\;\forall\:l\,\in\,[1,\,k]$. In other words, the child polytope 
 $(\mathbb{P}^{n_e+n_s-k-1},\,\mathcal{P}_{\mathcal{H}^{\mbox{\tiny $(k)$}}})$ obtained as a restriction of $(\mathbb{P}^{n_e+n_s-1},\,\mathcal{P})$ onto $\mathcal{H}^{\mbox{\tiny $(k)$}}$ is such that $\mathcal{P}_{\mathcal{H}^{\mbox{\tiny $(k)$}}}$ is the convex hull of the vertices of all the generating triangles of $(\mathbb{P}^{n_e+n_s-1},\,\mathcal{P})$ and a subset of the vertices of its generating segments, with the very same intersections among triangles and segments as $(\mathbb{P}^{n_e+n_s-1},\,\mathcal{P})$. Thus, $(\mathbb{P}^{n_e+n_s-k-1},\,\mathcal{P}_{\mathcal{H}^{\mbox{\tiny $(k)$}}})$ is a cosmological polytope and its associated graph $\mathcal{G}_{\mathcal{H}^{\mbox{\tiny $(k)$}}}$ can be obtained from the graph $\mathcal{G}$ by suppressing the $k$ tadpoles related to the segments which are not on $\mathcal{H}^{\mbox{\tiny $(k)$}}$. Hence, given that the facets of a polytope are given by the subgraphs of the associated graphs, the subgraphs $\mathfrak{g}\,\subseteq\,\mathcal{G}$ are mapped into subgraphs $\mathfrak{g}_{\mathcal{H}^{\mbox{\tiny $(k)$}}}\,\subseteq\,\mathcal{G}_{\mathcal{H}^{\mbox{\tiny $(h)$}}}$ by excluding in $\mathfrak{g}$ the vertices corresponding to the tadpoles that one has to eliminate to map $\mathcal{G}$ into $\mathcal{G}_{\mathcal{G}_{\mathcal{H}^{\mbox{\tiny $(k)$}}}}$: all the facets of the parent polytope are mapped into facets of the child polytope. However, counting how many subgraphs of $\mathcal{G}$ (and therefore how many facets of the parent polytope) return the same subgraph of the child polytope {\it does not} provide the correct counting of the multiplicity of the poles of the covariant form on the child polytope: the configuration of vertices obtained from $\mathfrak{g}\,\subseteq\,\mathcal{G}$ by excluding the vertices of the tadpoles which are eliminated upon restriction, can be equivalently obtained by considering the common vertices among subgraphs of $\mathcal{G}$, {\it i.e. } the intersections of the facets of the parent polytope that give a higher codimension face. In order for $l$ facets to intersect, the number of $n_v$ common vertices must be such that they can span $\mathbb{P}^{n_s+n_e-l-1}$, {\it i.e. } $n_v\,\ge\,n_s+n_e-l$.  When $n_v\,<\,n_s+n_e-l$, the vertices cannot span $\mathbb{P}^{n_s+n_e-l-1}$ and the facets do not intersect. Thus, given $l$ facets of the parent polytope which intersecting provide the same vertex configuration of a facet of the child polytope and such that their common vertices $n_v$ span $\mathbb{P}^{n_s+n_e-l-1}$, these $l$ facets intersect each other on $\mathcal{H}^{\mbox{\tiny $(k)$}}$ and $l$ provides the multiplicity of the pole in the covariant form of the child polytope. Finally, notice that $l$ is also the codimension of the face of the parent polytope identified by the intersection of its $l$ facets and, consequently, it is possible to state that the multiplicity of a pole in the child polytope along a certain facet is given by the codimension of the face of the parent polytope with the same vertex configuration, and the Laurent coefficient of the covariant form along this facet of the child polytope is the residue of the canonical form of the parent polytope along such a face.
\end{proof}

We will show a realisation of the Preposition \ref{prop:CPwk} in an explicit example afterwords. For the time being, it is important to remark that, given a cosmological polytope, it is possible to systematically construct a full class of covariant forms associated to it.

\begin{prop}\label{prop_GenWk}
 Let $(\mathbb{P}^{n'_e+n_s-1},\,\mathcal{P}')$ a cosmological polytope constructed from a collection of $n_t$ triangles and $n_h$ segments. Let $\mathcal{G}'$ be the associated graph with $n_s$ sites and $n'_e$ edges of which $n'_h$ are tadpoles subgraphs. Let $\{2\mathbf{x}_a-\mathbf{h}_a,\,\mathbf{h}_a\}_{a=1}^k$ be a collection of segments, and $\mathcal{T}_a$ the corresponding tadpole graph. Then, it possible to generate a class of covariant forms of degree-$k$ associated to $(\mathbb{P}^{n'_e+n_s-1},\,\mathcal{P}')$ from the covariant restriction of the canonical form of the cosmological polytope $\{(\mathbb{P}^{n_e+n_s-1},\,\mathcal{P})\}$ ($n_e\,=\,n'_e+k$) that can be constucted by intersecting in all possible ways the $k$ segments with $\mathcal{P}'$. The graphs associated with such polytopes are obtained from the graph $\mathcal{G}'$ by attaching $k$ tadpoles according to the intersections of the segments with $\mathcal{P}'$ and the restriction is on the hyperplane which suppresses the additional tadpoles. The covariant forms of degree $k$ generated in this way all have poles along the boundaries of $(\mathbb{P}^{n'_e+n_s-1},\,\mathcal{P}')$ (all of them are associated to this polytope) but with different multiplicities.
\end{prop}

\begin{proof}
 Let $(\mathbb{P}^{n'_e+n_s-1},\,\mathcal{P}')$ a cosmological polytope, whose graph $\mathcal{G}'$ has $n_s$ sites and $n'_e$ edges, with $n'_e$ including both the edges connecting different sites and the edges in the tadpole subgraphs. Let $\{2\mathbf{x}_a-\mathbf{h}_a,\,\mathbf{h}_a\}_{a=1}^k$ be a collection of segments, and $\mathcal{T}_a$ the corresponding tadpole graph. Following the definition of the cosmological polytope, we can construct a new polytope by merging the site of each tadpole $\{\mathcal{T}_a\}_{a=1}^{k}$ with the sites of $\mathcal{G}'$ generating a graph $\mathcal{G}$ with the same number $n_s$ of sites and $n_e\,=\,n'_e+k$ edges which describes a cosmological polytope $(\mathbb{P}^{n_e+n_s-1},\,\mathcal{P})$. However, there are 
 $\begin{pmatrix} k+n_s-1 \\ n_s-1\end{pmatrix}$ ways of attaching $k$ tadpoles to a graph $\mathcal{G}'$ with $n_s$ sites. Thus, given  $\mathcal{G}'$ and the collection $\{\mathcal{T}_a\}$ of $k$ tadpoles, it is possible to construct $\begin{pmatrix} k+n_s-1 \\ n_s-1\end{pmatrix}$ inequivalent graphs and, hence, cosmological polytopes in $\mathbb{P}^{n_e+n_s-1}$. Let us label the convex hull of these polytopes as $\mathcal{P}_{\sigma}$, with $\sigma$ labeling the inequivalent configuration of the $k$ tadpoles. Each polytope generated in this way can be now restricted on a hyperplane \eqref{eq:HK} such that the resulting polytope has again $\mathcal{G}'$ as an associated graph. As from Proposition \ref{prop:CPwk}, the facets of $(\mathbb{P}^{n'_e+n_s-1},\,\mathcal{P}')$ are encoded in higher codimension faces of $\{(\mathbb{P}^{n_e+n_s-1},\,\mathcal{P}_{\sigma})\}$ which are given by intersection of their facets. However, for each $\sigma$ such intersections change and hence the codimension of the face corresponding to a given facet of $(\mathbb{P}^{n'_e+n_s-1},\,\mathcal{P}')$. Then, for each $\sigma$, the covariant restriction of the canonical form returns a covariant form of degree-$k$ whose poles along each facet has multiplicity given by the codimension of the face of $\{(\mathbb{P}^{n_e+n_s-1},\,\mathcal{P}_{\sigma})\}$ with the same vertex configuration, and such a codimension depends on $\sigma$. Hence we obtain $\begin{pmatrix} k+n_s-1 \\ n_s-1\end{pmatrix}$ covariant forms of degree $k$ associated to $(\mathbb{P}^{n'_e+n_s-1},\,\mathcal{P}')$ with different multiplicity for their poles.
\end{proof}

Let us illustrate both Propositions \ref{prop:CPwk} and \ref{prop_GenWk}, starting with the latter. As the simplest example let us consider a two-site line graph $\mathcal{G}$ (whose associated polytope is a triangle) and two tadpoles. There are three inequivalent ways of generating a new graph by attaching the tadpoles to the sites of $\mathcal{G}$:
\begin{equation*}
 \begin{tikzpicture}[ball/.style = {circle, draw, align=center, anchor=north, inner sep=0}, cross/.style={cross out, draw, minimum size=2*(#1-\pgflinewidth), inner sep=0pt, outer sep=0pt}]
  \begin{scope}
   \coordinate[label=left:{\footnotesize $x_1$}] (s1) at (0,0);
   \coordinate[label=right:{\footnotesize $x_2$}] (s2) at (2,0);
   \coordinate[label=above:{\footnotesize $y$}] (s12) at ($(s1)!0.5!(s2)$);
   \coordinate[label=below:{\footnotesize $x'_1$}] (s3) at (0,-2);
   \coordinate (c3) at ($(s3)+(0,.325cm)$);
   \coordinate[label=left:{\footnotesize $h_1$}] (l3) at ($(c3)+(-.325cm,0)$);
   \coordinate[label=below:{\footnotesize $x'_2$}] (s4) at (2,-2);
   \coordinate (c4) at ($(s4)+(0,.325cm)$);
   \coordinate[label=left:{\footnotesize $h_2$}] (l4) at ($(c4)+(-.25cm,0)$);
   
   \draw[-, very thick] (s1) -- (s2);
   \draw[fill] (s1) circle (2pt);
   \draw[fill] (s2) circle (2pt);
   
   \draw[very thick] (c3) circle (.325cm);
   \draw[fill] (s3) circle (2pt);
   \draw[very thick] (c4) circle (.325cm);
   \draw[fill] (s4) circle (2pt);
  \end{scope}
  \begin{scope}[shift={(6,0)}, transform shape]
   \coordinate[label=left:{\footnotesize $x_1$}] (s1) at (0,0);
   \coordinate[label=right:{\footnotesize $x_2$}] (s2) at (2,0);
   \coordinate[label=above:{\footnotesize $y$}] (s12) at ($(s1)!0.5!(s2)$);
   \coordinate (c1) at ($(s1)+(0,+.25cm)$);
   \coordinate[label=left:{\footnotesize $h_1$}] (l1) at ($(c1)+(0,+.325cm)$);
   \coordinate (c2) at ($(s1)+(0,-.25cm)$);
   \coordinate[label=right:{\footnotesize $h_2$}] (l2) at ($(c2)+(+.25cm,0)$);
   
   \draw[-, very thick] (s1) -- (s2);
   \draw[very thick] (c1) circle (.25cm);
   \draw[very thick] (c2) circle (.25cm);
   \draw[fill] (s1) circle (2pt);
   \draw[fill] (s2) circle (2pt);
  \end{scope}
  \begin{scope}[shift={(6,-1)}, transform shape]
   \coordinate[label=below:{\footnotesize $\hspace{.1cm}x_1$}] (s1) at (0,0);
   \coordinate[label=below:{\footnotesize $x_2\hspace{.1cm}$}] (s2) at (2,0);
   \coordinate[label=above:{\footnotesize $y$}] (s12) at ($(s1)!0.5!(s2)$);
   \coordinate (c1) at ($(s1)+(-.25cm,0)$);
   \coordinate[label=left:{\footnotesize $h_1$}] (l1) at ($(c1)+(-.25cm,0)$);
   \coordinate (c2) at ($(s2)+(+.25cm,0)$);
   \coordinate[label=right:{\footnotesize $h_1$}] (l2) at ($(c2)+(+.25cm,0)$);
   
   \draw[-, very thick] (s1) -- (s2);
   \draw[very thick] (c1) circle (.25cm);
   \draw[very thick] (c2) circle (.25cm);
   \draw[fill] (s1) circle (2pt);
   \draw[fill] (s2) circle (2pt);
  \end{scope}
  \begin{scope}[shift={(6,-2)}, transform shape]
   \coordinate[label=left:{\footnotesize $x_1$}] (s1) at (0,0);
   \coordinate[label=right:{\footnotesize $x_2$}] (s2) at (2,0);
   \coordinate[label=above:{\footnotesize $y$}] (s12) at ($(s1)!0.5!(s2)$);
   \coordinate (c1) at ($(s2)+(0,+.25cm)$);
   \coordinate[label=right:{\footnotesize $h_1$}] (l1) at ($(c1)+(+.25cm,0)$);
   \coordinate (c2) at ($(s2)+(0,-.25cm)$);
   \coordinate[label=right:{\footnotesize $h_2$}] (l2) at ($(c2)+(+.25cm,0)$);
   
   \draw[-, very thick] (s1) -- (s2);
   \draw[very thick] (c1) circle (.25cm);
   \draw[very thick] (c2) circle (.25cm);
   \draw[fill] (s1) circle (2pt);
   \draw[fill] (s2) circle (2pt);
  \end{scope}
 \end{tikzpicture}
\end{equation*}
Let us label $\mathcal{G}_{\mbox{\tiny $11$}},\:\mathcal{G}_{\mbox{\tiny $12$}},\,\mathcal{G}_{22}$ the three graphs appearing on the right, with the indices $ij$ indicating the site where each tadpole has been merged. All the polytopes associated to these graphs live in $\mathbb{P}^{4}$: the number of sites and edges is the same in all three cases, what changes is the way that the triangle associated to the two-site line graph has been intersected with the two segments associated to the two tadpoles. The polytopes associated to $\mathcal{G}_{\mbox{\tiny $11$}},\:\mathcal{G}_{\mbox{\tiny $12$}},\,\mathcal{G}_{22}$ are the convex hulls of, respectively, the following list of vertices
\begin{equation*}
  \begin{split}
    &\{\mathbf{x}_1-\mathbf{y}+\mathbf{x}_2,\:\mathbf{x}_1+\mathbf{y}-\mathbf{x}_2,\:-\mathbf{x}_1+\mathbf{y}+\mathbf{x}_2,\:
       2\mathbf{x}_1-\mathbf{h}_1,\;\mathbf{h}_1,\:2\mathbf{x}_1-\mathbf{h}_2,\;\mathbf{h}_2\},\\
    &\{\mathbf{x}_1-\mathbf{y}+\mathbf{x}_2,\:\mathbf{x}_1+\mathbf{y}-\mathbf{x}_2,\:-\mathbf{x}_1+\mathbf{y}+\mathbf{x}_2,\:
       2\mathbf{x}_1-\mathbf{h}_1,\;\mathbf{h}_1,\:2\mathbf{x}_2-\mathbf{h}_2,\;\mathbf{h}_2\},\\
    &\{\mathbf{x}_1-\mathbf{y}+\mathbf{x}_2,\:\mathbf{x}_1+\mathbf{y}-\mathbf{x}_2,\:-\mathbf{x}_1+\mathbf{y}+\mathbf{x}_2,\:
       2\mathbf{x}_2-\mathbf{h}_1,\;\mathbf{h}_1,\:2\mathbf{x}_2-\mathbf{h}_2,\;\mathbf{h}_2\},\\
  \end{split}
\end{equation*}
The canonical function can be readily written for all three polytopes
\begin{align*}
        \Omega_{\mbox{\tiny $(11)$}}\:=\:&\frac{1}{(x_1+x_2)(x_1+2h_1+2h_2)(y+x_2)}
         \left[
          \frac{2(x_1+y+x_2+h_1+h_2)}{(x_1+y)(x_1+y+2h_1)(x_1+y+2h_2)}+
         \right.\\
         &+
          \frac{2x_1+y+4h_1+2h_2}{(x_1+x_2+2h_1)(x_1+y+2h_1)(x_1+x_2+2h_1+2h_2)}+\\
         &\left.+
          \frac{2x_1+y+2h_1+4h_2}{(x_1+x_2+2h_2)(x_1+y+2h_2)(x_1+x_2+2h_1+2h_2)}
         \right],\\
        \Omega_{\mbox{\tiny $(12)$}}\:=\:&\frac{1}{(x_1+x_2)(x_1+y+2h_1)(y+x_2+2h_2)}
         \left[
          \frac{x_1+y+2x_2+2h_1+2h_2}{(x_1+x_2+2h_1)(x_1+x_2+2h_1+2h_2)(y+x_2)}
         \right.\\
         &\left.+
          \frac{2x_2+y+x_2+2h_1+2h_2}{(x_1+x_2+2h_2)(x_1+y)(x_1+x_2+2h_1+2h_2)}
         \right],
\end{align*}
with the canonical function $\Omega_{\mbox{\tiny $(22)$}}$ that can be obtained from $\Omega_{\mbox{\tiny $(11)$}}$ via the exchange $x_1\,\longleftrightarrow\,x_2$. In order to obtain covariant forms on the triangle (whose associated graph is the two-site line graph), the canonical forms of the polytopes associated to $\mathcal{G}_{11}$, $\mathcal{G}_{12}$ and $\mathcal{G}_{22}$  has to be restricted on the following hyperplane
\begin{equation*}
  \mathcal{H}\:=\:
   \left\{
    \mathcal{Y}\,\in\,\mathbb{P}^4\,\Bigg|\,
     \begin{array}{l}
          \mathcal{Y}\cdot\mathbf{\tilde{h}}_1\,=\,h_1\,=\,0  \\
          \mathcal{Y}\cdot\mathbf{\tilde{h}}_2\,=\,h_2\,=\,0 
     \end{array}
   \right\}.
\end{equation*}
It is easy so see that the only vertices which are on $\mathcal{H}$ are $\{\mathbf{x}_1-\mathbf{y}+\mathbf{x}_2,\:\mathbf{x}_1+\mathbf{y}-\mathbf{x}_2,\:-\mathbf{x}_1+\mathbf{y}+\mathbf{x}_2,\}$ in all three cases. The covariant restriction of the canonical forms produces covariant forms of degree-$2$ on the triangle, whose canonical functions are
\begin{equation}\label{eq:CFT}
  \begin{split}
    &\Omega_{\mbox{\tiny $(11)$}}^{\mbox{\tiny $(2)$}}\:=\:
      \frac{3x_1^2+3x_1y+3x_1x_2+y^2+yx_2+x_2^2}{(x_1+x_2)^3(x_1+y)^3(y+x_2)},\\
    &\Omega_{\mbox{\tiny $(12)$}}^{\mbox{\tiny $(2)$}}\:=\:
      \frac{x_1^2+2x_1y+3x_1x_2+y^2+2yx_2+x_2^2}{(x_1+x_2)^3(x_1+y)^2(y+x_2)^2}
  \end{split}
\end{equation}
and, again, $\Omega_{\mbox{\tiny $(22)$}}^{\mbox{\tiny $(2)$}}$ can be obtained from $\Omega_{\mbox{\tiny $(11)$}}^{\mbox{\tiny $(2)$}}$ via the exchange $x_1\,\longleftrightarrow\,x_2$. As for the Proposition \ref{prop_GenWk}, the poles of the three covariant forms are all along the facets of the triangle, with just different multiplicities, which is a reflection of the face structure of the different parent polytopes. The multiplicity of each pole in the covariant forms, whose canonical functions are given by \eqref{eq:CFT}, is the codimension $l$ of the face of the parent polytope which matches the relevant facet of the child polytope, recalling that $l$ facets intersect each other in a codimension-$l$ face if their common vertices span $\mathbb{P}^{4-l}$.

Let us now consider the cosmological polytope associated to a two-site graph with $4$ edges, two of which are tadpoles on the two different sites

\begin{wrapfigure}{l}{5.5cm}
 \centering
 \begin{tikzpicture}[line join = round, line cap = round, ball/.style = {circle, draw, align=center, anchor=north, inner sep=0}, 
                     axis/.style={very thick, ->, >=stealth'}, pile/.style={thick, ->, >=stealth', shorten <=2pt, shorten>=2pt}, every node/.style={color=black}]
 \begin{scope}
  \coordinate (LC) at (0,0);
  \coordinate[label=right:{$x_1$}] (x1) at (-1.25cm,0);
  \coordinate[label=left:{$x_2$}] (x2) at (+1.25cm,0);
  \coordinate[label=above:{$y_{12}$}] (y12) at ($(LC)+(0,1.25cm)$);
  \coordinate[label=below:{$y_{21}$}] (y21) at ($(LC)-(0,1.25cm)$);
  \coordinate (t1) at ($(x1)-(.325cm,0)$);
  \coordinate (t2) at ($(x2)+(.325cm,0)$);
  \coordinate[label=left:{$h_1$}] (h1) at ($(t1)-(.325cm,0)$);
  \coordinate[label=right:{$h_2$}] (h2) at ($(t2)+(.325cm,0)$);
  
  \draw[very thick] (LC) circle (1.25cm);
  \draw[very thick] (t1) circle (.325cm);
  \draw[very thick] (t2) circle (.325cm);
  \draw[fill] (x1) circle (2pt);
  \draw[fill] (x2) circle (2pt);
 \end{scope}
 \end{tikzpicture}
\end{wrapfigure}

This cosmological polytope lives in $\mathbb{P}^{5}$ and it is the convex hull of the following $10$ vertices:
\begin{equation}\label{eq:CPexV}
 \begin{split}
  &\left\{
    \mathbf{x}_1-\mathbf{y}_{12}+\mathbf{x}_2,\:\mathbf{x}_1+\mathbf{y}_{12}-\mathbf{x}_2,\:-\mathbf{x}_1+\mathbf{y}_{12}+\mathbf{x}_2,
   \right.\\
  &\hspace{.25cm}\mathbf{x}_1-\mathbf{y}_{21}+\mathbf{x}_2,\:\mathbf{x}_1+\mathbf{y}_{21}-\mathbf{x}_2,\:-\mathbf{x}_1+\mathbf{y}_{21}+\mathbf{x}_2,\\
  &\left.
    \hspace{.25cm}2\mathbf{x}_1-\mathbf{h}_1,\;\mathbf{h}_1,\qquad 2\mathbf{x}_2-\mathbf{h}_2,\;\mathbf{h}_2
   \right\}
 \end{split}
\end{equation}
where the first two lines are the vertices of the triangles and the last one the ones of the two segments which are intersected to generate it. We label them as $Z_{a}$ $(a\,=\,1,\dots,10)$ in the same order as they appear in \eqref{eq:CPexV}. The weights on the graph are the local coordinates $\mathcal{Y}\,:=\,(x_1,\,y_{12},\,y_{21},\,x_2,\,h_1,\,h_2)\,\in\,\mathbb{P}^5$ corresponding to the collection of vectors of midpoints for both the generating triangles and segments and the non-intersectable vertex for the segments, as a basis for $\mathbb{R}^6$. This cosmological polytope has $16$ facets, whose hyperplanes $\mathcal{W}$ are given as $\mathcal{Y}^{I}\mathcal{W}_I$ by taking all the possible subgraphs and associating to them the sum of the weights of the vertex plus the sum of the weights of those edges which depart from the vertices of the subgraph but are not contained in the subgraph:
\begin{align*}\label{eq:CPfct}
    &
    \begin{tikzpicture}[ball/.style = {circle, draw, align=center, anchor=north, inner sep=0}, cross/.style={cross out, draw, minimum size=2*(#1-\pgflinewidth), inner sep=0pt, outer sep=0pt}]
     \begin{scope}[scale={.675}, transform shape]
      \coordinate (LC) at (0,0);
      \coordinate[label=right:{$x_1$}] (x1) at (-1.25cm,0);
      \coordinate[label=left:{$x_2$}] (x2) at (+1.25cm,0);
      \coordinate[label=above:{$y_{12}$}] (y12) at ($(LC)+(0,1.25cm)$);
      \coordinate[label=below:{$y_{21}$}] (y21) at ($(LC)-(0,1.25cm)$);
      \coordinate (t1) at ($(x1)-(.325cm,0)$);
      \coordinate (t2) at ($(x2)+(.325cm,0)$);
      \coordinate[label=left:{$h_1$}] (h1) at ($(t1)-(.325cm,0)$);
      \coordinate[label=right:{$h_2$}] (h2) at ($(t2)+(.325cm,0)$);
      \draw[very thick] (LC) circle (1.25cm);
      \draw[very thick] (t1) circle (.325cm);
      \draw[very thick] (t2) circle (.325cm);
      \draw[fill] (x1) circle (2pt);
      \draw[fill] (x2) circle (2pt); 
      \coordinate (c1) at ($(LC)+(0,1.25cm)$);
      \coordinate (c2) at ($(LC)-(0,1.25cm)$);
      \coordinate (c3) at ($(t1)-(.325cm,0)$);
      \coordinate (c4) at ($(t2)+(.325cm,0)$);
      \node[ultra thick, cross=6pt, rotate=0, color=blue] at (c1) {};
      \node[ultra thick, cross=6pt, rotate=0, color=blue] at (c2) {};
      \node[ultra thick, cross=6pt, rotate=0, color=blue] at (c3) {};
      \node[ultra thick, cross=6pt, rotate=0, color=blue] at (c4) {};
      \coordinate (HP) at ($(LC)-(0,2.25cm)$);
      \node[align=center] (eqH) at (HP) {$\displaystyle\langle\mathcal{Y}2358(10)\rangle\,=\,0$ \\
                                         $\displaystyle(x_1+x_2\,=\,0)$};
     \end{scope}
     \begin{scope}[shift={(4,0)}, scale={.675}, transform shape]
      \coordinate (LC) at (0,0);
      \coordinate[label=right:{$x_1$}] (x1) at (-1.25cm,0);
      \coordinate[label=left:{$x_2$}] (x2) at (+1.25cm,0);
      \coordinate[label=above:{$y_{12}$}] (y12) at ($(LC)+(0,1.25cm)$);
      \coordinate[label=below:{$y_{21}$}] (y21) at ($(LC)-(0,1.25cm)$);
      \coordinate (t1) at ($(x1)-(.325cm,0)$);
      \coordinate (t2) at ($(x2)+(.325cm,0)$);
      \coordinate[label=left:{$h_1$}] (h1) at ($(t1)-(.325cm,0)$);
      \coordinate[label=right:{$h_2$}] (h2) at ($(t2)+(.325cm,0)$);
      \draw[very thick] (LC) circle (1.25cm);
      \draw[very thick] (t1) circle (.325cm);
      \draw[very thick] (t2) circle (.325cm);
      \draw[fill] (x1) circle (2pt);
      \draw[fill] (x2) circle (2pt); 
      \coordinate (c1) at ($(LC)+(0,1.25cm)$);
      \coordinate (c2) at ($(LC)-(0,1.25cm)$);
      \coordinate (c3) at ($(x1)+(-.125,+.175)$);
      \coordinate (c4) at ($(x1)+(-.125,-.175)$);
      \coordinate (c5) at ($(x2)+(+.125,+.175)$);
      \coordinate (c6) at ($(x2)+(+.125,-.175)$);
      \node[ultra thick, cross=6pt, rotate=0, color=blue] at (c1) {};
      \node[ultra thick, cross=6pt, rotate=0, color=blue] at (c2) {};
      \node[ultra thick, cross=6pt, rotate=0, color=blue] at (c3) {};
      \node[ultra thick, cross=6pt, rotate=0, color=blue] at (c4) {};
      \node[ultra thick, cross=6pt, rotate=0, color=blue] at (c5) {};
      \node[ultra thick, cross=6pt, rotate=0, color=blue] at (c6) {};
      \coordinate (HP) at ($(LC)-(0,2.25cm)$);
      \node[align=center] (eqH) at (HP) {$\displaystyle\langle\mathcal{Y}23579\rangle\,=\,0$ \\
                                         $\displaystyle(x_1+x_2+2h_1+2h_2\,=\,0)$};
     \end{scope}
     \begin{scope}[shift={(8,0)}, scale={.675}, transform shape]
      \coordinate (LC) at (0,0);
      \coordinate[label=right:{$x_1$}] (x1) at (-1.25cm,0);
      \coordinate[label=left:{$x_2$}] (x2) at (+1.25cm,0);
      \coordinate[label=above:{$y_{12}$}] (y12) at ($(LC)+(0,1.25cm)$);
      \coordinate[label=below:{$y_{21}$}] (y21) at ($(LC)-(0,1.25cm)$);
      \coordinate (t1) at ($(x1)-(.325cm,0)$);
      \coordinate (t2) at ($(x2)+(.325cm,0)$);
      \coordinate[label=left:{$h_1$}] (h1) at ($(t1)-(.325cm,0)$);
      \coordinate[label=right:{$h_2$}] (h2) at ($(t2)+(.325cm,0)$);
      \draw[very thick] (LC) circle (1.25cm);
      \draw[very thick] (t1) circle (.325cm);
      \draw[very thick] (t2) circle (.325cm);
      \draw[fill] (x1) circle (2pt);
      \draw[fill] (x2) circle (2pt); 
      \coordinate (c1) at ($(LC)+(0,1.25cm)$);
      \coordinate (c2) at ($(LC)-(0,1.25cm)$);
      \coordinate (c3) at ($(x1)+(-.125,+.175)$);
      \coordinate (c4) at ($(x1)+(-.125,-.175)$);
      \coordinate (c5) at ($(t2)+(.325cm,0)$);
      \node[ultra thick, cross=6pt, rotate=0, color=blue] at (c1) {};
      \node[ultra thick, cross=6pt, rotate=0, color=blue] at (c2) {};
      \node[ultra thick, cross=6pt, rotate=0, color=blue] at (c3) {};
      \node[ultra thick, cross=6pt, rotate=0, color=blue] at (c4) {};
      \node[ultra thick, cross=6pt, rotate=0, color=blue] at (c5) {};
      \coordinate (HP) at ($(LC)-(0,2.25cm)$);
      \node[align=center] (eqH) at (HP) {$\displaystyle\langle\mathcal{Y}2357(10)\rangle\,=\,0$ \\
                                         $\displaystyle(x_1+x_2+2h_1\,=\,0)$};
     \end{scope}
     \begin{scope}[shift={(12,0)}, scale={.675}, transform shape]
      \coordinate (LC) at (0,0);
      \coordinate[label=right:{$x_1$}] (x1) at (-1.25cm,0);
      \coordinate[label=left:{$x_2$}] (x2) at (+1.25cm,0);
      \coordinate[label=above:{$y_{12}$}] (y12) at ($(LC)+(0,1.25cm)$);
      \coordinate[label=below:{$y_{21}$}] (y21) at ($(LC)-(0,1.25cm)$);
      \coordinate (t1) at ($(x1)-(.325cm,0)$);
      \coordinate (t2) at ($(x2)+(.325cm,0)$);
      \coordinate[label=left:{$h_1$}] (h1) at ($(t1)-(.325cm,0)$);
      \coordinate[label=right:{$h_2$}] (h2) at ($(t2)+(.325cm,0)$);
      \draw[very thick] (LC) circle (1.25cm);
      \draw[very thick] (t1) circle (.325cm);
      \draw[very thick] (t2) circle (.325cm);
      \draw[fill] (x1) circle (2pt);
      \draw[fill] (x2) circle (2pt); 
      \coordinate (c1) at ($(LC)+(0,1.25cm)$);
      \coordinate (c2) at ($(LC)-(0,1.25cm)$);
      \coordinate (c3) at ($(t1)-(.325cm,0)$);
      \coordinate (c4) at ($(x2)+(+.125,+.175)$);
      \coordinate (c5) at ($(x2)+(+.125,-.175)$);
      \node[ultra thick, cross=6pt, rotate=0, color=blue] at (c1) {};
      \node[ultra thick, cross=6pt, rotate=0, color=blue] at (c2) {};
      \node[ultra thick, cross=6pt, rotate=0, color=blue] at (c3) {};
      \node[ultra thick, cross=6pt, rotate=0, color=blue] at (c4) {};
      \node[ultra thick, cross=6pt, rotate=0, color=blue] at (c5) {};
      \coordinate (HP) at ($(LC)-(0,2.25cm)$);
      \node[align=center] (eqH) at (HP) {$\displaystyle\langle\mathcal{Y}23589\rangle\,=\,0$ \\
                                         $\displaystyle(x_1+x_2+2h_2\,=\,0)$};
     \end{scope}
    \end{tikzpicture}
    \\
    &
    \begin{tikzpicture}[ball/.style = {circle, draw, align=center, anchor=north, inner sep=0}, cross/.style={cross out, draw, minimum size=2*(#1-\pgflinewidth), inner sep=0pt, outer sep=0pt}]
     \begin{scope}[scale={.675}, transform shape]
      \coordinate (LC) at (0,0);
      \coordinate[label=right:{$x_1$}] (x1) at (-1.25cm,0);
      \coordinate[label=left:{$x_2$}] (x2) at (+1.25cm,0);
      \coordinate[label=above:{$y_{12}$}] (y12) at ($(LC)+(0,1.25cm)$);
      \coordinate[label=below:{$y_{21}$}] (y21) at ($(LC)-(0,1.25cm)$);
      \coordinate (t1) at ($(x1)-(.325cm,0)$);
      \coordinate (t2) at ($(x2)+(.325cm,0)$);
      \coordinate[label=left:{$h_1$}] (h1) at ($(t1)-(.325cm,0)$);
      \coordinate[label=right:{$h_2$}] (h2) at ($(t2)+(.325cm,0)$);
      \draw[very thick] (LC) circle (1.25cm);
      \draw[very thick] (t1) circle (.325cm);
      \draw[very thick] (t2) circle (.325cm);
      \draw[fill] (x1) circle (2pt);
      \draw[fill] (x2) circle (2pt); 
      \coordinate (c1a) at ($(x1)+(+.1125,+.4)$);
      \coordinate (c1b) at ($(x2)+(-.1125,+.4)$);
      \coordinate (c2) at ($(LC)-(0,1.25cm)$);
      \coordinate (c3) at ($(t1)-(.325cm,0)$);
      \coordinate (c4) at ($(t2)+(.325cm,0)$);
      \node[ultra thick, cross=6pt, rotate=0, color=blue] at (c1a) {};
      \node[ultra thick, cross=6pt, rotate=0, color=blue] at (c1b) {};
      \node[ultra thick, cross=6pt, rotate=0, color=blue] at (c2) {};
      \node[ultra thick, cross=6pt, rotate=0, color=blue] at (c3) {};
      \node[ultra thick, cross=6pt, rotate=0, color=blue] at (c4) {};
      \coordinate (HP) at ($(LC)-(0,2.25cm)$);
      \node[align=center] (eqH) at (HP) {$\displaystyle\langle\mathcal{Y}1568(10)\rangle\,=\,0$ \\
                                         $\displaystyle(x_1+x_2+2y_{12}\,=\,0)$};
     \end{scope}
     \begin{scope}[shift={(4,0)}, scale={.675}, transform shape]
      \coordinate (LC) at (0,0);
      \coordinate[label=right:{$x_1$}] (x1) at (-1.25cm,0);
      \coordinate[label=left:{$x_2$}] (x2) at (+1.25cm,0);
      \coordinate[label=above:{$y_{12}$}] (y12) at ($(LC)+(0,1.25cm)$);
      \coordinate[label=below:{$y_{21}$}] (y21) at ($(LC)-(0,1.25cm)$);
      \coordinate (t1) at ($(x1)-(.325cm,0)$);
      \coordinate (t2) at ($(x2)+(.325cm,0)$);
      \coordinate[label=left:{$h_1$}] (h1) at ($(t1)-(.325cm,0)$);
      \coordinate[label=right:{$h_2$}] (h2) at ($(t2)+(.325cm,0)$);
      \draw[very thick] (LC) circle (1.25cm);
      \draw[very thick] (t1) circle (.325cm);
      \draw[very thick] (t2) circle (.325cm);
      \draw[fill] (x1) circle (2pt);
      \draw[fill] (x2) circle (2pt); 
      \coordinate (c1a) at ($(x1)+(+.1125,+.4)$);
      \coordinate (c1b) at ($(x2)+(-.1125,+.4)$);
      \coordinate (c2) at ($(LC)-(0,1.25cm)$);
      \coordinate (c3) at ($(x1)+(-.125,+.175)$);
      \coordinate (c4) at ($(x1)+(-.125,-.175)$);
      \coordinate (c5) at ($(x2)+(+.125,+.175)$);
      \coordinate (c6) at ($(x2)+(+.125,-.175)$);
      \node[ultra thick, cross=6pt, rotate=0, color=blue] at (c1a) {};
      \node[ultra thick, cross=6pt, rotate=0, color=blue] at (c1b) {};
      \node[ultra thick, cross=6pt, rotate=0, color=blue] at (c2) {};
      \node[ultra thick, cross=6pt, rotate=0, color=blue] at (c3) {};
      \node[ultra thick, cross=6pt, rotate=0, color=blue] at (c4) {};
      \node[ultra thick, cross=6pt, rotate=0, color=blue] at (c5) {};
      \node[ultra thick, cross=6pt, rotate=0, color=blue] at (c6) {};
      \coordinate (HP) at ($(LC)-(0,2.25cm)$);
      \node[align=center] (eqH) at (HP) {$\displaystyle\langle\mathcal{Y}15679\rangle\,=\,0$ \\
                                         $\displaystyle(x_1++2y_{12}+x_2+2h_1+2h_2\,=\,0)$};
     \end{scope}
     \begin{scope}[shift={(8,0)}, scale={.675}, transform shape]
      \coordinate (LC) at (0,0);
      \coordinate[label=right:{$x_1$}] (x1) at (-1.25cm,0);
      \coordinate[label=left:{$x_2$}] (x2) at (+1.25cm,0);
      \coordinate[label=above:{$y_{12}$}] (y12) at ($(LC)+(0,1.25cm)$);
      \coordinate[label=below:{$y_{21}$}] (y21) at ($(LC)-(0,1.25cm)$);
      \coordinate (t1) at ($(x1)-(.325cm,0)$);
      \coordinate (t2) at ($(x2)+(.325cm,0)$);
      \coordinate[label=left:{$h_1$}] (h1) at ($(t1)-(.325cm,0)$);
      \coordinate[label=right:{$h_2$}] (h2) at ($(t2)+(.325cm,0)$);
      \draw[very thick] (LC) circle (1.25cm);
      \draw[very thick] (t1) circle (.325cm);
      \draw[very thick] (t2) circle (.325cm);
      \draw[fill] (x1) circle (2pt);
      \draw[fill] (x2) circle (2pt); 
      \coordinate (c1a) at ($(x1)+(+.1125,+.4)$);
      \coordinate (c1b) at ($(x2)+(-.1125,+.4)$);
      \coordinate (c2) at ($(LC)-(0,1.25cm)$);
      \coordinate (c3) at ($(x1)+(-.125,+.175)$);
      \coordinate (c4) at ($(x1)+(-.125,-.175)$);
      \coordinate (c5) at ($(t2)+(.325cm,0)$);
      \node[ultra thick, cross=6pt, rotate=0, color=blue] at (c1a) {};
      \node[ultra thick, cross=6pt, rotate=0, color=blue] at (c1b) {};
      \node[ultra thick, cross=6pt, rotate=0, color=blue] at (c2) {};
      \node[ultra thick, cross=6pt, rotate=0, color=blue] at (c3) {};
      \node[ultra thick, cross=6pt, rotate=0, color=blue] at (c4) {};
      \node[ultra thick, cross=6pt, rotate=0, color=blue] at (c5) {};
      \coordinate (HP) at ($(LC)-(0,2.25cm)$);
      \node[align=center] (eqH) at (HP) {$\displaystyle\langle\mathcal{Y}1567(10)\rangle\,=\,0$ \\
                                         $\displaystyle(x_1+2y_{12}+x_2+2h_1\,=\,0)$};
     \end{scope}
     \begin{scope}[shift={(12,0)}, scale={.675}, transform shape]
      \coordinate (LC) at (0,0);
      \coordinate[label=right:{$x_1$}] (x1) at (-1.25cm,0);
      \coordinate[label=left:{$x_2$}] (x2) at (+1.25cm,0);
      \coordinate[label=above:{$y_{12}$}] (y12) at ($(LC)+(0,1.25cm)$);
      \coordinate[label=below:{$y_{21}$}] (y21) at ($(LC)-(0,1.25cm)$);
      \coordinate (t1) at ($(x1)-(.325cm,0)$);
      \coordinate (t2) at ($(x2)+(.325cm,0)$);
      \coordinate[label=left:{$h_1$}] (h1) at ($(t1)-(.325cm,0)$);
      \coordinate[label=right:{$h_2$}] (h2) at ($(t2)+(.325cm,0)$);
      \draw[very thick] (LC) circle (1.25cm);
      \draw[very thick] (t1) circle (.325cm);
      \draw[very thick] (t2) circle (.325cm);
      \draw[fill] (x1) circle (2pt);
      \draw[fill] (x2) circle (2pt); 
      \coordinate (c1a) at ($(x1)+(+.1125,+.4)$);
      \coordinate (c1b) at ($(x2)+(-.1125,+.4)$);
      \coordinate (c2) at ($(LC)-(0,1.25cm)$);
      \coordinate (c3) at ($(t1)-(.325cm,0)$);
      \coordinate (c4) at ($(x2)+(+.125,+.175)$);
      \coordinate (c5) at ($(x2)+(+.125,-.175)$);
      \node[ultra thick, cross=6pt, rotate=0, color=blue] at (c1a) {};
      \node[ultra thick, cross=6pt, rotate=0, color=blue] at (c1b) {};
      \node[ultra thick, cross=6pt, rotate=0, color=blue] at (c2) {};
      \node[ultra thick, cross=6pt, rotate=0, color=blue] at (c3) {};
      \node[ultra thick, cross=6pt, rotate=0, color=blue] at (c4) {};
      \node[ultra thick, cross=6pt, rotate=0, color=blue] at (c5) {};
      \coordinate (HP) at ($(LC)-(0,2.25cm)$);
      \node[align=center] (eqH) at (HP) {$\displaystyle\langle\mathcal{Y}15689\rangle\,=\,0$ \\
                                         $\displaystyle(x_1+2y_{12}+x_2+2h_2\,=\,0)$};
     \end{scope}
    \end{tikzpicture}
    \\
    &
    \begin{tikzpicture}[ball/.style = {circle, draw, align=center, anchor=north, inner sep=0}, cross/.style={cross out, draw, minimum size=2*(#1-\pgflinewidth), inner sep=0pt, outer sep=0pt}]
     \begin{scope}[scale={.675}, transform shape]
      \coordinate (LC) at (0,0);
      \coordinate[label=right:{$x_1$}] (x1) at (-1.25cm,0);
      \coordinate[label=left:{$x_2$}] (x2) at (+1.25cm,0);
      \coordinate[label=above:{$y_{12}$}] (y12) at ($(LC)+(0,1.25cm)$);
      \coordinate[label=below:{$y_{21}$}] (y21) at ($(LC)-(0,1.25cm)$);
      \coordinate (t1) at ($(x1)-(.325cm,0)$);
      \coordinate (t2) at ($(x2)+(.325cm,0)$);
      \coordinate[label=left:{$h_1$}] (h1) at ($(t1)-(.325cm,0)$);
      \coordinate[label=right:{$h_2$}] (h2) at ($(t2)+(.325cm,0)$);
      \draw[very thick] (LC) circle (1.25cm);
      \draw[very thick] (t1) circle (.325cm);
      \draw[very thick] (t2) circle (.325cm);
      \draw[fill] (x1) circle (2pt);
      \draw[fill] (x2) circle (2pt); 
      \coordinate (c1) at ($(LC)+(0,1.25cm)$);
      \coordinate (c2a) at ($(x1)+(+.1125,-.4)$);
      \coordinate (c2b) at ($(x2)+(-.1125,-.4)$);
      \coordinate (c3) at ($(t1)-(.325cm,0)$);
      \coordinate (c4) at ($(t2)+(.325cm,0)$);
      \node[ultra thick, cross=6pt, rotate=0, color=blue] at (c1) {};
      \node[ultra thick, cross=6pt, rotate=0, color=blue] at (c2a) {};
      \node[ultra thick, cross=6pt, rotate=0, color=blue] at (c2b) {};
      \node[ultra thick, cross=6pt, rotate=0, color=blue] at (c3) {};
      \node[ultra thick, cross=6pt, rotate=0, color=blue] at (c4) {};
      \coordinate (HP) at ($(LC)-(0,2.25cm)$);
      \node[align=center] (eqH) at (HP) {$\displaystyle\langle\mathcal{Y}2348(10)\rangle\,=\,0$ \\
                                         $\displaystyle(x_1+2y_{21}+x_2\,=\,0)$};
     \end{scope}
     \begin{scope}[shift={(4,0)}, scale={.675}, transform shape]
      \coordinate (LC) at (0,0);
      \coordinate[label=right:{$x_1$}] (x1) at (-1.25cm,0);
      \coordinate[label=left:{$x_2$}] (x2) at (+1.25cm,0);
      \coordinate[label=above:{$y_{12}$}] (y12) at ($(LC)+(0,1.25cm)$);
      \coordinate[label=below:{$y_{21}$}] (y21) at ($(LC)-(0,1.25cm)$);
      \coordinate (t1) at ($(x1)-(.325cm,0)$);
      \coordinate (t2) at ($(x2)+(.325cm,0)$);
      \coordinate[label=left:{$h_1$}] (h1) at ($(t1)-(.325cm,0)$);
      \coordinate[label=right:{$h_2$}] (h2) at ($(t2)+(.325cm,0)$);
      \draw[very thick] (LC) circle (1.25cm);
      \draw[very thick] (t1) circle (.325cm);
      \draw[very thick] (t2) circle (.325cm);
      \draw[fill] (x1) circle (2pt);
      \draw[fill] (x2) circle (2pt); 
      \coordinate (c1) at ($(LC)+(0,1.25cm)$);
      \coordinate (c2a) at ($(x1)+(+.1125,-.4)$);
      \coordinate (c2b) at ($(x2)+(-.1125,-.4)$);
      \coordinate (c3) at ($(x1)+(-.125,+.175)$);
      \coordinate (c4) at ($(x1)+(-.125,-.175)$);
      \coordinate (c5) at ($(x2)+(+.125,+.175)$);
      \coordinate (c6) at ($(x2)+(+.125,-.175)$);
      \node[ultra thick, cross=6pt, rotate=0, color=blue] at (c1) {};
      \node[ultra thick, cross=6pt, rotate=0, color=blue] at (c2a) {};
      \node[ultra thick, cross=6pt, rotate=0, color=blue] at (c2b) {};
      \node[ultra thick, cross=6pt, rotate=0, color=blue] at (c3) {};
      \node[ultra thick, cross=6pt, rotate=0, color=blue] at (c4) {};
      \node[ultra thick, cross=6pt, rotate=0, color=blue] at (c5) {};
      \node[ultra thick, cross=6pt, rotate=0, color=blue] at (c6) {};
      \coordinate (HP) at ($(LC)-(0,2.25cm)$);
      \node[align=center] (eqH) at (HP) {$\displaystyle\langle\mathcal{Y}23479\rangle\,=\,0$ \\
                                         $\displaystyle(x_1+2y_{21}+x_2+2h_1+2h_2\,=\,0)$};
     \end{scope}
     \begin{scope}[shift={(8,0)}, scale={.675}, transform shape]
      \coordinate (LC) at (0,0);
      \coordinate[label=right:{$x_1$}] (x1) at (-1.25cm,0);
      \coordinate[label=left:{$x_2$}] (x2) at (+1.25cm,0);
      \coordinate[label=above:{$y_{12}$}] (y12) at ($(LC)+(0,1.25cm)$);
      \coordinate[label=below:{$y_{21}$}] (y21) at ($(LC)-(0,1.25cm)$);
      \coordinate (t1) at ($(x1)-(.325cm,0)$);
      \coordinate (t2) at ($(x2)+(.325cm,0)$);
      \coordinate[label=left:{$h_1$}] (h1) at ($(t1)-(.325cm,0)$);
      \coordinate[label=right:{$h_2$}] (h2) at ($(t2)+(.325cm,0)$);
      \draw[very thick] (LC) circle (1.25cm);
      \draw[very thick] (t1) circle (.325cm);
      \draw[very thick] (t2) circle (.325cm);
      \draw[fill] (x1) circle (2pt);
      \draw[fill] (x2) circle (2pt); 
      \coordinate (c1) at ($(LC)+(0,1.25cm)$);
      \coordinate (c2a) at ($(x1)+(+.1125,-.4)$);
      \coordinate (c2b) at ($(x2)+(-.1125,-.4)$);
      \coordinate (c3) at ($(x1)+(-.125,+.175)$);
      \coordinate (c4) at ($(x1)+(-.125,-.175)$);
      \coordinate (c5) at ($(t2)+(.325cm,0)$);
      \node[ultra thick, cross=6pt, rotate=0, color=blue] at (c1) {};
      \node[ultra thick, cross=6pt, rotate=0, color=blue] at (c2a) {};
      \node[ultra thick, cross=6pt, rotate=0, color=blue] at (c2b) {};
      \node[ultra thick, cross=6pt, rotate=0, color=blue] at (c3) {};
      \node[ultra thick, cross=6pt, rotate=0, color=blue] at (c4) {};
      \node[ultra thick, cross=6pt, rotate=0, color=blue] at (c5) {};
      \coordinate (HP) at ($(LC)-(0,2.25cm)$);
      \node[align=center] (eqH) at (HP) {$\displaystyle\langle\mathcal{Y}2347(10)\rangle\,=\,0$ \\
                                         $\displaystyle(x_1+y_{21}+x_2+2h_1\,=\,0)$};
     \end{scope}
     \begin{scope}[shift={(12,0)}, scale={.675}, transform shape]
      \coordinate (LC) at (0,0);
      \coordinate[label=right:{$x_1$}] (x1) at (-1.25cm,0);
      \coordinate[label=left:{$x_2$}] (x2) at (+1.25cm,0);
      \coordinate[label=above:{$y_{12}$}] (y12) at ($(LC)+(0,1.25cm)$);
      \coordinate[label=below:{$y_{21}$}] (y21) at ($(LC)-(0,1.25cm)$);
      \coordinate (t1) at ($(x1)-(.325cm,0)$);
      \coordinate (t2) at ($(x2)+(.325cm,0)$);
      \coordinate[label=left:{$h_1$}] (h1) at ($(t1)-(.325cm,0)$);
      \coordinate[label=right:{$h_2$}] (h2) at ($(t2)+(.325cm,0)$);
      \draw[very thick] (LC) circle (1.25cm);
      \draw[very thick] (t1) circle (.325cm);
      \draw[very thick] (t2) circle (.325cm);
      \draw[fill] (x1) circle (2pt);
      \draw[fill] (x2) circle (2pt); 
      \coordinate (c1) at ($(LC)+(0,1.25cm)$);
      \coordinate (c2a) at ($(x1)+(+.1125,-.4)$);
      \coordinate (c2b) at ($(x2)+(-.1125,-.4)$);
      \coordinate (c3) at ($(t1)-(.325cm,0)$);
      \coordinate (c4) at ($(x2)+(+.125,+.175)$);
      \coordinate (c5) at ($(x2)+(+.125,-.175)$);
      \node[ultra thick, cross=6pt, rotate=0, color=blue] at (c1) {};
      \node[ultra thick, cross=6pt, rotate=0, color=blue] at (c2a) {};
      \node[ultra thick, cross=6pt, rotate=0, color=blue] at (c2b) {};
      \node[ultra thick, cross=6pt, rotate=0, color=blue] at (c3) {};
      \node[ultra thick, cross=6pt, rotate=0, color=blue] at (c4) {};
      \node[ultra thick, cross=6pt, rotate=0, color=blue] at (c5) {};
      \coordinate (HP) at ($(LC)-(0,2.25cm)$);
      \node[align=center] (eqH) at (HP) {$\displaystyle\langle\mathcal{Y}23489\rangle\,=\,0$ \\
                                         $\displaystyle(x_1+x_2+2y_{21}+2h_2\,=\,0)$};
     \end{scope}
    \end{tikzpicture}
    \\
    &
    \begin{tikzpicture}[ball/.style = {circle, draw, align=center, anchor=north, inner sep=0}, cross/.style={cross out, draw, minimum size=2*(#1-\pgflinewidth), inner sep=0pt, outer sep=0pt}]
     \begin{scope}[scale={.675}, transform shape]
      \coordinate (LC) at (0,0);
      \coordinate[label=right:{$x_1$}] (x1) at (-1.25cm,0);
      \coordinate[label=left:{$x_2$}] (x2) at (+1.25cm,0);
      \coordinate[label=above:{$y_{12}$}] (y12) at ($(LC)+(0,1.25cm)$);
      \coordinate[label=below:{$y_{21}$}] (y21) at ($(LC)-(0,1.25cm)$);
      \coordinate (t1) at ($(x1)-(.325cm,0)$);
      \coordinate (t2) at ($(x2)+(.325cm,0)$);
      \coordinate[label=left:{$h_1$}] (h1) at ($(t1)-(.325cm,0)$);
      \coordinate[label=right:{$h_2$}] (h2) at ($(t2)+(.325cm,0)$);
      \draw[very thick] (LC) circle (1.25cm);
      \draw[very thick] (t1) circle (.325cm);
      \draw[very thick] (t2) circle (.325cm);
      \draw[fill] (x1) circle (2pt);
      \draw[fill] (x2) circle (2pt); 
      \coordinate (c1) at ($(x1)+(+.1125,+.4)$);
      \coordinate (c2) at ($(x1)+(+.1125,-.4)$);
      \coordinate (c3) at ($(t1)-(.325cm,0)$);
      \node[ultra thick, cross=6pt, rotate=0, color=blue] at (c1) {};
      \node[ultra thick, cross=6pt, rotate=0, color=blue] at (c2) {};
      \node[ultra thick, cross=6pt, rotate=0, color=blue] at (c3) {};
      \coordinate (HP) at ($(LC)-(0,2.25cm)$);
      \node[align=center] (eqH) at (HP) {$\displaystyle\langle\mathcal{Y}13489\rangle\,=\,0$ \\
                                         $\displaystyle(x_1+y_{12}+y_{21}\,=\,0)$};
     \end{scope}
     \begin{scope}[shift={(4,0)}, scale={.675}, transform shape]
      \coordinate (LC) at (0,0);
      \coordinate[label=right:{$x_1$}] (x1) at (-1.25cm,0);
      \coordinate[label=left:{$x_2$}] (x2) at (+1.25cm,0);
      \coordinate[label=above:{$y_{12}$}] (y12) at ($(LC)+(0,1.25cm)$);
      \coordinate[label=below:{$y_{21}$}] (y21) at ($(LC)-(0,1.25cm)$);
      \coordinate (t1) at ($(x1)-(.325cm,0)$);
      \coordinate (t2) at ($(x2)+(.325cm,0)$);
      \coordinate[label=left:{$h_1$}] (h1) at ($(t1)-(.325cm,0)$);
      \coordinate[label=right:{$h_2$}] (h2) at ($(t2)+(.325cm,0)$);
      \draw[very thick] (LC) circle (1.25cm);
      \draw[very thick] (t1) circle (.325cm);
      \draw[very thick] (t2) circle (.325cm);
      \draw[fill] (x1) circle (2pt);
      \draw[fill] (x2) circle (2pt); 
      \coordinate (c1) at ($(x1)+(+.1125,+.4)$);
      \coordinate (c2) at ($(x1)+(+.1125,-.4)$);
      \coordinate (c3) at ($(x1)+(-.125,+.175)$);
      \coordinate (c4) at ($(x1)+(-.125,-.175)$);
      \node[ultra thick, cross=6pt, rotate=0, color=blue] at (c1) {};
      \node[ultra thick, cross=6pt, rotate=0, color=blue] at (c2) {};
      \node[ultra thick, cross=6pt, rotate=0, color=blue] at (c3) {};
      \node[ultra thick, cross=6pt, rotate=0, color=blue] at (c4) {};
      \coordinate (HP) at ($(LC)-(0,2.25cm)$);
      \node[align=center] (eqH) at (HP) {$\displaystyle\langle\mathcal{Y}13479\rangle\,=\,0$ \\
                                         $\displaystyle(x_1+y_{12}+y_{21}+2h_1\,=\,0)$};
     \end{scope}
     \begin{scope}[shift={(8,0)}, scale={.675}, transform shape]
      \coordinate (LC) at (0,0);
      \coordinate[label=right:{$x_1$}] (x1) at (-1.25cm,0);
      \coordinate[label=left:{$x_2$}] (x2) at (+1.25cm,0);
      \coordinate[label=above:{$y_{12}$}] (y12) at ($(LC)+(0,1.25cm)$);
      \coordinate[label=below:{$y_{21}$}] (y21) at ($(LC)-(0,1.25cm)$);
      \coordinate (t1) at ($(x1)-(.325cm,0)$);
      \coordinate (t2) at ($(x2)+(.325cm,0)$);
      \coordinate[label=left:{$h_1$}] (h1) at ($(t1)-(.325cm,0)$);
      \coordinate[label=right:{$h_2$}] (h2) at ($(t2)+(.325cm,0)$);
      \draw[very thick] (LC) circle (1.25cm);
      \draw[very thick] (t1) circle (.325cm);
      \draw[very thick] (t2) circle (.325cm);
      \draw[fill] (x1) circle (2pt);
      \draw[fill] (x2) circle (2pt); 
      \coordinate (c1) at ($(x2)+(-.1125,+.4)$);
      \coordinate (c2) at ($(x2)+(-.1125,-.4)$);
      \coordinate (c3) at ($(t2)+(.325cm,0)$);
      \node[ultra thick, cross=6pt, rotate=0, color=blue] at (c1) {};
      \node[ultra thick, cross=6pt, rotate=0, color=blue] at (c2) {};
      \node[ultra thick, cross=6pt, rotate=0, color=blue] at (c3) {};
      \coordinate (HP) at ($(LC)-(0,2.25cm)$);
      \node[align=center] (eqH) at (HP) {$\displaystyle\langle\mathcal{Y}1247(10)\rangle\,=\,0$ \\
                                         $\displaystyle(x_2+y_{12}+y_{21}\,=\,0)$};
     \end{scope}
     \begin{scope}[shift={(12,0)}, scale={.675}, transform shape]
      \coordinate (LC) at (0,0);
      \coordinate[label=right:{$x_1$}] (x1) at (-1.25cm,0);
      \coordinate[label=left:{$x_2$}] (x2) at (+1.25cm,0);
      \coordinate[label=above:{$y_{12}$}] (y12) at ($(LC)+(0,1.25cm)$);
      \coordinate[label=below:{$y_{21}$}] (y21) at ($(LC)-(0,1.25cm)$);
      \coordinate (t1) at ($(x1)-(.325cm,0)$);
      \coordinate (t2) at ($(x2)+(.325cm,0)$);
      \coordinate[label=left:{$h_1$}] (h1) at ($(t1)-(.325cm,0)$);
      \coordinate[label=right:{$h_2$}] (h2) at ($(t2)+(.325cm,0)$);
      \draw[very thick] (LC) circle (1.25cm);
      \draw[very thick] (t1) circle (.325cm);
      \draw[very thick] (t2) circle (.325cm);
      \draw[fill] (x1) circle (2pt);
      \draw[fill] (x2) circle (2pt); 
      \coordinate (c1) at ($(x2)+(-.1125,+.4)$);
      \coordinate (c2) at ($(x2)+(-.1125,-.4)$);
      \coordinate (c3) at ($(x2)+(+.125,+.175)$);
      \coordinate (c4) at ($(x2)+(+.125,-.175)$);
      \node[ultra thick, cross=6pt, rotate=0, color=blue] at (c1) {};
      \node[ultra thick, cross=6pt, rotate=0, color=blue] at (c2) {};
      \node[ultra thick, cross=6pt, rotate=0, color=blue] at (c3) {};
      \node[ultra thick, cross=6pt, rotate=0, color=blue] at (c4) {};
      \coordinate (HP) at ($(LC)-(0,2.25cm)$);
      \node[align=center] (eqH) at (HP) {$\displaystyle\langle\mathcal{Y}12479\rangle\,=\,0$ \\
                                         $\displaystyle(x_1+y_{12}+y_{21}+2h_2\,=\,0)$};
     \end{scope}
    \end{tikzpicture}
\end{align*}
There are three {\it natural} hyperplanes where to restrict this cosmological polytope, two being of codimension-$1$ and one of codimension-$2$:
\begin{equation}\label{eq:CPhr}
 \begin{split}
   &\mathcal{H}_1\: :=\:\left\{\mathcal{Y}\,\in\,\mathbb{P}^5\,\big|\,\mathfrak{h}_1(\mathcal{Y})\, :=\,\mathcal{Y}\cdot\tilde{\mathbf{h}}_1\,=\,h_1\,=\,0,\;\mathbf{h}_1\cdot\tilde{\mathbf{h}}_1\,=\,1,\:
       (\mathbf{x}_s,\,\mathbf{y}_e,\,\mathbf{h}_2)\cdot\tilde{\mathbf{h}}_1\,=\,0\right\},\\
   &\mathcal{H}_2\: :=\:\left\{\mathcal{Y}\,\in\,\mathbb{P}^5\,\big|\,\mathfrak{h}_2(\mathcal{Y})\, :=\,\mathcal{Y}\cdot\tilde{\mathbf{h}}_2\,=\,h_2\,=\,0,\;\mathbf{h}_2\cdot\tilde{\mathbf{h}}_2\,=\,1,\:
       (\mathbf{x}_s,\,\mathbf{y}_e,\,\mathbf{h}_1)\cdot\tilde{\mathbf{h}}_2\,=\,0\right\},\\
   &\mathcal{H}_{12}\: :=\:\left\{\mathcal{Y}\,\in\,\mathbb{P}^5\,\bigg|\,
    \begin{array}{l}
     \mathfrak{h}_1(\mathcal{Y})\, :=\,\mathcal{Y}\cdot\tilde{\mathbf{h}}_1\,=\,h_1\,=\,0,\:\mathbf{h}_1\cdot\tilde{\mathbf{h}}_1\,=\,1,\:
      (\mathbf{x}_s,\,\mathbf{y}_e,\,\mathbf{h}_2)\cdot\tilde{\mathbf{h}}_1\,=\,0\\
     \mathfrak{h}_2(\mathcal{Y})\, :=\,\mathcal{Y}\cdot\tilde{\mathbf{h}}_2\,=\,h_2\,=\,0,\:\mathbf{h}_2\cdot\tilde{\mathbf{h}}_2\,=\,1,\:
      (\mathbf{x}_s,\,\mathbf{y}_e,\,\mathbf{h}_1)\cdot\tilde{\mathbf{h}}_2\,=\,0
    \end{array}
    \right\}
 \end{split}
\end{equation}
where the equation $\mathcal{Y}\,\cdot\,\mathcal{W}\,=\,0$ identifying each facet is indicated both projectively and in our preferred local coordinate system below each graph -- as explained in Section \ref{subsec:CP}, the markings on the graphs indicate those vertices that do not belong to the facet, and the double marking in the tadpole subgraphs close to its side indicates (the absence of) the very same vertex.

There is a number of information about the resulting covariant forms that can be deduced from the graphs without knowing the explicit expression for the canonical form of the cosmological polytope we are restricting. Let us discuss in detail the covariant restriction of the canonical form of $(\mathbb{P}^5,\,\mathcal{P})$ onto $\mathcal{H}_1$ and $\mathcal{H}_{12}$ -- indeed, the analysis of the covariant restriction onto $\mathcal{H}_2$ follows from the former.

Let us begin with the restriction onto $\mathcal{H}_1$. Being a codimension-one hyperplane, the covariant form obtained has degree-$1$, with at most double poles. The first information we can predict is the child polytope itself $(\mathbb{P}^{4},\,\mathcal{P}_{\mathcal{H}_1})$, with $\mathcal{P}_{\mathcal{H}_1}\,=\,\mathcal{P}\cap\mathcal{H}_1$  as well as exactly which poles of the canonical form of the parent polytope collapses to generate double poles in the covariant form of the child polytope, {\it i.e.} which facets intersect $\mathcal{H}_1$ in the same subspace. 

\begin{wrapfigure}{l}{4cm}
 \centering
 \begin{tikzpicture}[line join = round, line cap = round, ball/.style = {circle, draw, align=center, anchor=north, inner sep=0}, 
                     axis/.style={very thick, ->, >=stealth'}, pile/.style={thick, ->, >=stealth', shorten <=2pt, shorten>=2pt}, every node/.style={color=black}]
 \begin{scope}[scale={.75}, transform shape]
  \coordinate (LC) at (0,0);
  \coordinate[label=right:{$x_1$}] (x1) at (-1.25cm,0);
  \coordinate[label=left:{$x_2$}] (x2) at (+1.25cm,0);
  \coordinate[label=above:{$y_{12}$}] (y12) at ($(LC)+(0,1.25cm)$);
  \coordinate[label=below:{$y_{21}$}] (y21) at ($(LC)-(0,1.25cm)$);
  \coordinate (t1) at ($(x1)-(.325cm,0)$);
  \coordinate (t2) at ($(x2)+(.325cm,0)$);
  \coordinate[label=right:{$h_2$}] (h2) at ($(t2)+(.325cm,0)$);
  
  \draw[very thick] (LC) circle (1.25cm);
  \draw[very thick] (t2) circle (.325cm);
  \draw[fill] (x1) circle (2pt);
  \draw[fill] (x2) circle (2pt);
 \end{scope}
 \end{tikzpicture}
\end{wrapfigure}

The crucial observation is that on the restriction onto $\mathcal{H}_1$, the vertices $Z_7\,:=\,2\mathbf{x}_1-\mathbf{h}_1$ and $Z_8\,:=\:\mathbf{h}_1$ of $\mathcal{P}$
are not on $\mathcal{H}_1$. Hence, the child polytope $(\mathbb{P}^4,\,\mathcal{P}_{\mathcal{H}_1})$ is related to graph which is the one associated to the parent polytope but without the tadpole whose edge has weight $h_1$. Consequently, the facets intersecting $\mathcal{H}_1$ in the same subspace have the structure $\langle\mathcal{Y}7ijkl\rangle$ and $\langle\mathcal{Y}8ijkl\rangle$ for fixed $Z_i,\,Z_j,\,Z_k,\,Z_l$. From the facet structure listed above for each of the $16$ facets, it is easy to see the facets have the same intersection in pairs, so that the covariant form of degree $1$ associated to the child polytope $(\mathbb{P}^4,\,\mathcal{P}_{\mathcal{H}_1})$ have $8$ double poles, each corresponding to a facet of $(\mathbb{P}^4,\,\mathcal{P}_{\mathcal{H}_1})$. Notice further that the parent polytope contains {\it all} the facets of the child polytope as codimension-$2$ faces, relating in this way the residue of the canonical form of the parent polytope along the codimension-$2$ faces to the leading Laurent coefficient along the boundaries of the child polytope. This is readly seen by comparing the vertex structure of, for example, the codimension-$2$ face of the parent polytope defined by the conditions $\langle\mathcal{Y}23579\rangle\,=\,0\,=\,\langle\mathcal{Y}23589\rangle$, and the vertex structure of the facet of the child polytope identified by $\langle\mathcal{Y}_{\mathcal{H}_1}2359\rangle\,=\,0$
\begin{equation*}
     \begin{tikzpicture}[ball/.style = {circle, draw, align=center, anchor=north, inner sep=0}, cross/.style={cross out, draw, minimum size=2*(#1-\pgflinewidth), inner sep=0pt, outer sep=0pt}]
      \begin{scope}[scale={.75}, transform shape]
      \coordinate (LC) at (0,0);
      \coordinate[label=right:{$x_1$}] (x1) at (-1.25cm,0);
      \coordinate[label=left:{$x_2$}] (x2) at (+1.25cm,0);
      \coordinate[label=above:{$y_{12}$}] (y12) at ($(LC)+(0,1.25cm)$);
      \coordinate[label=below:{$y_{21}$}] (y21) at ($(LC)-(0,1.25cm)$);
      \coordinate (t1) at ($(x1)-(.325cm,0)$);
      \coordinate (t2) at ($(x2)+(.325cm,0)$);
      \coordinate[label=left:{$h_1$}] (h1) at ($(t1)-(.325cm,0)$);
      \coordinate[label=right:{$h_2$}] (h2) at ($(t2)+(.325cm,0)$);
      \draw[very thick] (LC) circle (1.25cm);
      \draw[very thick] (t1) circle (.325cm);
      \draw[very thick] (t2) circle (.325cm);
      \draw[fill] (x1) circle (2pt);
      \draw[fill] (x2) circle (2pt); 
      \coordinate (c1) at ($(LC)+(0,1.25cm)$);
      \coordinate (c2) at ($(LC)-(0,1.25cm)$);
      \coordinate (c3) at ($(x1)+(-.125,+.175)$);
      \coordinate (c4) at ($(x1)+(-.125,-.175)$);
      \coordinate (c5) at ($(x2)+(+.125,+.175)$);
      \coordinate (c6) at ($(x2)+(+.125,-.175)$);
      \coordinate (c7) at ($(t1)-(.325cm,0)$);
      \node[ultra thick, cross=6pt, rotate=0, color=blue] at (c1) {};
      \node[ultra thick, cross=6pt, rotate=0, color=blue] at (c2) {};
      \node[ultra thick, cross=6pt, rotate=0, color=blue] at (c3) {};
      \node[ultra thick, cross=6pt, rotate=0, color=blue] at (c4) {};
      \node[ultra thick, cross=6pt, rotate=0, color=blue] at (c5) {};
      \node[ultra thick, cross=6pt, rotate=0, color=blue] at (c6) {};
      \node[ultra thick, cross=6pt, rotate=0, color=blue] at (c7) {};
      \coordinate (HP) at ($(LC)-(0,2.25cm)$);
      \node[align=center] (eqH) at (HP) {$\displaystyle\langle\mathcal{Y}23579\rangle\,=\,0$ \\
                                         $\displaystyle\langle\mathcal{Y}23589\rangle\,=\,0$};
     \end{scope}
     \begin{scope}[shift={(4,0)}, scale={.75}, transform shape]
      \coordinate (LC) at (0,0);
      \coordinate[label=right:{$x_1$}] (x1) at (-1.25cm,0);
      \coordinate[label=left:{$x_2$}] (x2) at (+1.25cm,0);
      \coordinate[label=above:{$y_{12}$}] (y12) at ($(LC)+(0,1.25cm)$);
      \coordinate[label=below:{$y_{21}$}] (y21) at ($(LC)-(0,1.25cm)$);
      \coordinate (t1) at ($(x1)-(.325cm,0)$);
      \coordinate (t2) at ($(x2)+(.325cm,0)$);
      \coordinate[label=right:{$h_2$}] (h2) at ($(t2)+(.325cm,0)$);
  
      \draw[very thick] (LC) circle (1.25cm);
      \draw[very thick] (t2) circle (.325cm);
      \draw[fill] (x1) circle (2pt);
      \draw[fill] (x2) circle (2pt);
      
      \coordinate (c1) at ($(LC)+(0,1.25cm)$);
      \coordinate (c2) at ($(LC)-(0,1.25cm)$);
      \coordinate (c3) at ($(x2)+(+.125,+.175)$);
      \coordinate (c4) at ($(x2)+(+.125,-.175)$);
      \node[ultra thick, cross=6pt, rotate=0, color=blue] at (c1) {};
      \node[ultra thick, cross=6pt, rotate=0, color=blue] at (c2) {};
      \node[ultra thick, cross=6pt, rotate=0, color=blue] at (c3) {};
      \node[ultra thick, cross=6pt, rotate=0, color=blue] at (c4) {};
      \coordinate (HP) at ($(LC)-(0,2.25cm)$);
      \node[align=center] (eqH) at (HP) {$\displaystyle\langle\mathcal{Y}_{\mathcal{H}_1}2359\rangle\,=\,0$};
     \end{scope}
    \end{tikzpicture}
\end{equation*}

Recall that the marking singles out the vertices which are {\it not} on the face and, consequently, the codimension-$2$ face of the parent polytope and the facet of the child polytope are the same. Notice also that in local coordinates the two conditions defining the codimension-$2$ face of the parent polytope write $x_1+x_2+2h_1+2h_2\,=\,0\,=\,x_1+x_2+2h_2$, which also imply $h_1\,=\,0$ the defining condition for the hyperplane $\mathcal{H}_1$.

\begin{wrapfigure}{l}{4cm}
 \centering
 \begin{tikzpicture}[line join = round, line cap = round, ball/.style = {circle, draw, align=center, anchor=north, inner sep=0}, 
                     cross/.style={cross out, draw, minimum size=2*(#1-\pgflinewidth), inner sep=0pt, outer sep=0pt}, every node/.style={color=black}]
 \begin{scope}[scale={.75}, transform shape]
  \coordinate (LC) at (0,0);
  \coordinate[label=left:{$x_1$}] (x1) at (-1.25cm,0);
  \coordinate[label=right:{$x_2$}] (x2) at (+1.25cm,0);
  \coordinate[label=above:{$y_{12}$}] (y12) at ($(LC)+(0,1.25cm)$);
  \coordinate[label=below:{$y_{21}$}] (y21) at ($(LC)-(0,1.25cm)$);
  
  \draw[very thick] (LC) circle (1.25cm);
  \draw[fill] (x1) circle (2pt);
  \draw[fill] (x2) circle (2pt);
 \end{scope}
 \end{tikzpicture}
\end{wrapfigure}

Let us now turn to the restriction onto $\mathcal{H}_{12}$. Being a codimension-two hyperplane, the covariant form obtained from the restriction has degree-$2$, with at most poles of multiplicity $3$. Again, it is straightforward to predict the child polytope $(\mathbb{P}^3,\,\mathcal{H}_{12})$, with $\mathcal{P}_{\mathcal{H}_{12}}\,:=\,\mathcal{P}\cap\mathcal{H}_{12}$: the vertices $Z_7,\,Z_8,\,Z_9,\,Z_{19}$ are not on $\mathcal{H}_{12}$. Hence, $\mathcal{P}_{\mathcal{H}_{12}}$ is the convex hull of the vertices of two triangles intersecting each other in both their midpoints of their two intersectable facets, {\it i.e.} it is a truncated tetrahedron in $\mathbb{P}^3$ (see Figure \ref{fig:CP}), and its associated graph is one-loop two site graph. Furthermore, notice that the four facets in the first three lines in the list above intersect the hyperplane $\mathcal{H}_{12}$ in the same codimension-$3$ hyperplane, which is a facet of the child polytope, while the facets in the last line intersect it in the same hyperplane in pairs. So, one would expect the covariant form of degree-$2$ associated to the child polytope to have three poles of multipliticity $4$ and two double poles. However, a covariant form of degree-$2$ can have at most poles with multiplicity $3$! Recall that the multiplicity of the pole of the covariant form of the child polytope is also given by the codimension of the face matching a facet of the child polytope. For the sake of concreteness, let us consider the following facet for the child polytope

\begin{equation*}
 \begin{tikzpicture}[line join = round, line cap = round, ball/.style = {circle, draw, align=center, anchor=north, inner sep=0}, 
                     cross/.style={cross out, draw, minimum size=2*(#1-\pgflinewidth), inner sep=0pt, outer sep=0pt}, every node/.style={color=black}]
  \begin{scope}[scale={.75}, transform shape]
   \coordinate (LC) at (0,0);
   \coordinate[label=left:{$x_1$}] (x1) at (-1.25cm,0);
   \coordinate[label=right:{$x_2$}] (x2) at (+1.25cm,0);
   \coordinate[label=above:{$y_{12}$}] (y12) at ($(LC)+(0,1.25cm)$);
   \coordinate[label=below:{$y_{21}$}] (y21) at ($(LC)-(0,1.25cm)$);
  
   \draw[very thick] (LC) circle (1.25cm);
   \draw[fill] (x1) circle (2pt);
   \draw[fill] (x2) circle (2pt);
  
   \coordinate (c1) at ($(LC)+(0,1.25cm)$);
   \coordinate (c2) at ($(LC)-(0,1.25cm)$);
   \node[ultra thick, cross=6pt, rotate=0, color=blue] at (c1) {};
   \node[ultra thick, cross=6pt, rotate=0, color=blue] at (c2) {};
  \end{scope}
 \end{tikzpicture}
\end{equation*}
which corresponds, in local coordinate, to the facet $x_1+x_2\,=\,0$. Now we should ask the question which higher codimension face of the parent polytope has only the vertices of such a facet. Looking at all the facets of the parent polytope listed above, it is easy to see that the higher codimension face we are looking for is contained in the following four facets

\begin{equation*}\label{eq:CPfct}
    \begin{tikzpicture}[ball/.style = {circle, draw, align=center, anchor=north, inner sep=0}, cross/.style={cross out, draw, minimum size=2*(#1-\pgflinewidth), inner sep=0pt, outer sep=0pt}]
     \begin{scope}[scale={.675}, transform shape]
      \coordinate (LC) at (0,0);
      \coordinate[label=right:{$x_1$}] (x1) at (-1.25cm,0);
      \coordinate[label=left:{$x_2$}] (x2) at (+1.25cm,0);
      \coordinate[label=above:{$y_{12}$}] (y12) at ($(LC)+(0,1.25cm)$);
      \coordinate[label=below:{$y_{21}$}] (y21) at ($(LC)-(0,1.25cm)$);
      \coordinate (t1) at ($(x1)-(.325cm,0)$);
      \coordinate (t2) at ($(x2)+(.325cm,0)$);
      \coordinate[label=left:{$h_1$}] (h1) at ($(t1)-(.325cm,0)$);
      \coordinate[label=right:{$h_2$}] (h2) at ($(t2)+(.325cm,0)$);
      \draw[very thick] (LC) circle (1.25cm);
      \draw[very thick] (t1) circle (.325cm);
      \draw[very thick] (t2) circle (.325cm);
      \draw[fill] (x1) circle (2pt);
      \draw[fill] (x2) circle (2pt); 
      \coordinate (c1) at ($(LC)+(0,1.25cm)$);
      \coordinate (c2) at ($(LC)-(0,1.25cm)$);
      \coordinate (c3) at ($(t1)-(.325cm,0)$);
      \coordinate (c4) at ($(t2)+(.325cm,0)$);
      \node[ultra thick, cross=6pt, rotate=0, color=blue] at (c1) {};
      \node[ultra thick, cross=6pt, rotate=0, color=blue] at (c2) {};
      \node[ultra thick, cross=6pt, rotate=0, color=blue] at (c3) {};
      \node[ultra thick, cross=6pt, rotate=0, color=blue] at (c4) {};
      \coordinate (HP) at ($(LC)-(0,2.25cm)$);
      \node[align=center] (eqH) at (HP) {$\displaystyle\langle\mathcal{Y}2358(10)\rangle\,=\,0$ \\
                                         $\displaystyle(x_1+x_2\,=\,0)$};
     \end{scope}
     \begin{scope}[shift={(4,0)}, scale={.675}, transform shape]
      \coordinate (LC) at (0,0);
      \coordinate[label=right:{$x_1$}] (x1) at (-1.25cm,0);
      \coordinate[label=left:{$x_2$}] (x2) at (+1.25cm,0);
      \coordinate[label=above:{$y_{12}$}] (y12) at ($(LC)+(0,1.25cm)$);
      \coordinate[label=below:{$y_{21}$}] (y21) at ($(LC)-(0,1.25cm)$);
      \coordinate (t1) at ($(x1)-(.325cm,0)$);
      \coordinate (t2) at ($(x2)+(.325cm,0)$);
      \coordinate[label=left:{$h_1$}] (h1) at ($(t1)-(.325cm,0)$);
      \coordinate[label=right:{$h_2$}] (h2) at ($(t2)+(.325cm,0)$);
      \draw[very thick] (LC) circle (1.25cm);
      \draw[very thick] (t1) circle (.325cm);
      \draw[very thick] (t2) circle (.325cm);
      \draw[fill] (x1) circle (2pt);
      \draw[fill] (x2) circle (2pt); 
      \coordinate (c1) at ($(LC)+(0,1.25cm)$);
      \coordinate (c2) at ($(LC)-(0,1.25cm)$);
      \coordinate (c3) at ($(x1)+(-.125,+.175)$);
      \coordinate (c4) at ($(x1)+(-.125,-.175)$);
      \coordinate (c5) at ($(x2)+(+.125,+.175)$);
      \coordinate (c6) at ($(x2)+(+.125,-.175)$);
      \node[ultra thick, cross=6pt, rotate=0, color=blue] at (c1) {};
      \node[ultra thick, cross=6pt, rotate=0, color=blue] at (c2) {};
      \node[ultra thick, cross=6pt, rotate=0, color=blue] at (c3) {};
      \node[ultra thick, cross=6pt, rotate=0, color=blue] at (c4) {};
      \node[ultra thick, cross=6pt, rotate=0, color=blue] at (c5) {};
      \node[ultra thick, cross=6pt, rotate=0, color=blue] at (c6) {};
      \coordinate (HP) at ($(LC)-(0,2.25cm)$);
      \node[align=center] (eqH) at (HP) {$\displaystyle\langle\mathcal{Y}23579\rangle\,=\,0$ \\
                                         $\displaystyle(x_1+x_2+2h_1+2h_2\,=\,0)$};
     \end{scope}
     \begin{scope}[shift={(8,0)}, scale={.675}, transform shape]
      \coordinate (LC) at (0,0);
      \coordinate[label=right:{$x_1$}] (x1) at (-1.25cm,0);
      \coordinate[label=left:{$x_2$}] (x2) at (+1.25cm,0);
      \coordinate[label=above:{$y_{12}$}] (y12) at ($(LC)+(0,1.25cm)$);
      \coordinate[label=below:{$y_{21}$}] (y21) at ($(LC)-(0,1.25cm)$);
      \coordinate (t1) at ($(x1)-(.325cm,0)$);
      \coordinate (t2) at ($(x2)+(.325cm,0)$);
      \coordinate[label=left:{$h_1$}] (h1) at ($(t1)-(.325cm,0)$);
      \coordinate[label=right:{$h_2$}] (h2) at ($(t2)+(.325cm,0)$);
      \draw[very thick] (LC) circle (1.25cm);
      \draw[very thick] (t1) circle (.325cm);
      \draw[very thick] (t2) circle (.325cm);
      \draw[fill] (x1) circle (2pt);
      \draw[fill] (x2) circle (2pt); 
      \coordinate (c1) at ($(LC)+(0,1.25cm)$);
      \coordinate (c2) at ($(LC)-(0,1.25cm)$);
      \coordinate (c3) at ($(x1)+(-.125,+.175)$);
      \coordinate (c4) at ($(x1)+(-.125,-.175)$);
      \coordinate (c5) at ($(t2)+(.325cm,0)$);
      \node[ultra thick, cross=6pt, rotate=0, color=blue] at (c1) {};
      \node[ultra thick, cross=6pt, rotate=0, color=blue] at (c2) {};
      \node[ultra thick, cross=6pt, rotate=0, color=blue] at (c3) {};
      \node[ultra thick, cross=6pt, rotate=0, color=blue] at (c4) {};
      \node[ultra thick, cross=6pt, rotate=0, color=blue] at (c5) {};
      \coordinate (HP) at ($(LC)-(0,2.25cm)$);
      \node[align=center] (eqH) at (HP) {$\displaystyle\langle\mathcal{Y}2357(10)\rangle\,=\,0$ \\
                                         $\displaystyle(x_1+x_2+2h_1\,=\,0)$};
     \end{scope}
     \begin{scope}[shift={(12,0)}, scale={.675}, transform shape]
      \coordinate (LC) at (0,0);
      \coordinate[label=right:{$x_1$}] (x1) at (-1.25cm,0);
      \coordinate[label=left:{$x_2$}] (x2) at (+1.25cm,0);
      \coordinate[label=above:{$y_{12}$}] (y12) at ($(LC)+(0,1.25cm)$);
      \coordinate[label=below:{$y_{21}$}] (y21) at ($(LC)-(0,1.25cm)$);
      \coordinate (t1) at ($(x1)-(.325cm,0)$);
      \coordinate (t2) at ($(x2)+(.325cm,0)$);
      \coordinate[label=left:{$h_1$}] (h1) at ($(t1)-(.325cm,0)$);
      \coordinate[label=right:{$h_2$}] (h2) at ($(t2)+(.325cm,0)$);
      \draw[very thick] (LC) circle (1.25cm);
      \draw[very thick] (t1) circle (.325cm);
      \draw[very thick] (t2) circle (.325cm);
      \draw[fill] (x1) circle (2pt);
      \draw[fill] (x2) circle (2pt); 
      \coordinate (c1) at ($(LC)+(0,1.25cm)$);
      \coordinate (c2) at ($(LC)-(0,1.25cm)$);
      \coordinate (c3) at ($(t1)-(.325cm,0)$);
      \coordinate (c4) at ($(x2)+(+.125,+.175)$);
      \coordinate (c5) at ($(x2)+(+.125,-.175)$);
      \node[ultra thick, cross=6pt, rotate=0, color=blue] at (c1) {};
      \node[ultra thick, cross=6pt, rotate=0, color=blue] at (c2) {};
      \node[ultra thick, cross=6pt, rotate=0, color=blue] at (c3) {};
      \node[ultra thick, cross=6pt, rotate=0, color=blue] at (c4) {};
      \node[ultra thick, cross=6pt, rotate=0, color=blue] at (c5) {};
      \coordinate (HP) at ($(LC)-(0,2.25cm)$);
      \node[align=center] (eqH) at (HP) {$\displaystyle\langle\mathcal{Y}23589\rangle\,=\,0$ \\
                                         $\displaystyle(x_1+x_2+2h_2\,=\,0)$};
     \end{scope}
    \end{tikzpicture}
\end{equation*}
Such for facets of the parent polytope are exactly the ones which intersect $\mathcal{H}_{12}$ in the same subspace. Now, in order to extract a codimension-$l$ face, we need to check which $l$ facets have enough vertices in common to span $\mathbb{P}^{5-l}$ and these vertices are precisely the ones whose convex hull is precisely the facet of the child polytope we are interested in. In principle, we find the desired vertex configuration taking three possible intersections among the four facets listed above: we can take the first two facets; the first, the third and the fourth; or the second, the third and the fourth. In the first case, the face would be of codimension-$2$ and in the other two cases it would be of codimension-$3$. Are all these intersection actually possible? Let us check whether the common vertices are enough to span $\mathbb{P}^3$ in the first case, and $\mathbb{P}^2$ in the other two. Given that we are looking at a specific vertex configuration, the vertices are the same in all three cases and are given by
\begin{equation*}
    \{\mathbf{x}_1+\mathbf{y}_{12}-\mathbf{x}_2,\:-\mathbf{x}_1+\mathbf{y}_{12}+\mathbf{x}_2,\:
      \mathbf{x}_1+\mathbf{y}_{21}-\mathbf{x}_2,\:-\mathbf{x}_1+\mathbf{y}_{21}+\mathbf{x}_2\}.
\end{equation*}
Importantly, they are not linearly independent and they lie on a $2$-plane. Hence, one has $3$ linearly independent vertices, which indeed can span $\mathbb{P}^2$ but they {\it cannot} span $\mathbb{P}^3$. Hence, the first two facets of the four of the parent polytope listed above {\it do not} intersect with each other, which means that when we take the residue of the canonical form with respect a pole related to any of these two facets, the other pole become spurious ({\it i.e.} the numerator develops a zero which cancel it). Thus, the facet of the child polytope of interest corresponds to a codimension-$3$ face of the parent polytope: the pole of the covariant form associated to the child polytope has a pole of multiplicity three along this facet, matching the expectations. Hence, the canonical form of the parent polytope develops a simple zero at the location of the pole on the covariant restriction onto $\mathcal{H}_{12}$ which lowers the multiplicity of the pole to $3$.

\section{Conclusions and Outlook}\label{sec:Concl}

In this paper we started to scratch the surface of a combinatorial and geometrical characterisation of differential forms with non-logarithmic singularities, whose understanding is crucial in physics as they describe scattering amplitudes in flat space-time and the wavefunction of the universe in cosmology. 

Specifically, we characterised meromorphic differential forms with multiple pole by relating them to projective polytopes via the notion of covariant forms and covariant pairings. Covariant forms are meromorphic differential forms with multiple poles and a certain $GL(1)$-scaling. Their distinctive feature is to have multiple poles only along the boundaries of the associated projective polytope such that its leading Laurent coefficient along any of the boundaries is a differential form associated to the relevant boundary of the projective polytope enjoying this same feature. The covariant pairing instead associates a meromorphic differential form with multiple poles to a polytope, with the differential form having poles along the boundaries of a certain signed triangulation of the polytope.  This includes those subsets of boundaries which sign-triangulate the empty set, with the special feature that the multiplicity of the poles related to such subsets is lowered upon summation. The form is expressed as the sum of covariant forms associated to the elements of the signed triangulation.

Contrarily to what happens for canonical forms, which are in a $1-1$ correspondence with a positive geometry, given a polytope there is a  full class of covariant forms and differential forms which can be in covariant pairing with it. Hence the geometry and combinatorics of the polytope does not determine completely these meromorphic forms with multiple poles. However, a complete geometrical and combinatorial characterisation of both covariant forms and forms in covariant pairing with a given polytope is possible if we think of this polytope as obtained as a restriction of a higher dimensional polytope onto a certain hyperplane. In the paper we named the higher dimensional polytope as {\it parent polytope}, and the polytope obtained as its restriction onto a hyperplane as {\it child polytope} with respect to that hyperplane. Then, a meromorphic differential form associated to the child polytope can be obtained as {\it covariant restriction} of the canonical form of the parent polytope, {\it i.e.} it is the leading Laurent coefficient (which is of order zero) of the canonical form of the parent polytope along the chosen hyperplane. If the hyperplane intersects the parent polytope only inside, then the differential form obtained as covariant restriction is a covariant form of the child polytope, while if the facets of the parent polytope intersect the hyperplane also outside, then it is in covariant pairing with the child polytope having poles along boundaries outside of the child polytope. Interestingly, this picture also provides a geometrical interpretation for the multiplicity of each pole, which is given by the number of facets of the parent polytope intersecting the hyperplane in the same subspace minus the multiplicity of the zero where this subspace were to be on the hypersurface determining the zeroes of the canonical form of the parent polytope. For covariant forms, this latter situation cannot occur given that the intersection between parent polytope and hyperplane is inside the parent polytope.

We have seen how differential forms obtained as restrictions from the canonical form of a parent polytope (a simplex) can be used to triangulate a given polytope. 
For general projective polytopes, we know that their canonical function can be computed applying the operation of Jeffrey-Kirwan residue to a differential form. This differential form turns out to be the restriction of the canonical form of a simplex onto the hyperplanes identified by the fibers of the original polytope (seen as a projection from the simplex). The form is not the canonical form of the fiber, but it is in covariant pairing with it. In \cite{moh:2020}, these type of forms will also be defined in the context of objects which are more general than polytopes, such as the \emph{amplituhedra}.

For cosmological polytopes generated as convex hull of the vertices of triangles and segments intersected in their midpoints, there are natural hyperplanes onto which perform the covariant restriction of their canonical form. these are such that the differential form obtained is a covariant form encoding the wavefunction of the universe for certain massive scalar states as well as massless ones in FRW cosmologies in arbitrary dimensions. Curiously, these special hyperplanes relate the covariant form associated to the child polytope obtained as covariant restriction on them, to the canonical form of the child polytope itself. The covariant form of the child polytope can be obtained from the action of a differential operator onto its canonical form, with the order of the derivative operator given by the codimension of the hyperplane where the parent polytope is restricted to give the child polytope. In other words, the canonical coefficient of the parent polytope is the Newton's difference quotient of the canonical form of the child polytope. This point raises the more general question of which differential operators can be thought of in this polytope picture. Beyond having mathematical interest, this question is also physically motivated. The parent-child polytope relation that we observed in cosmological polytopes is the geometrical-combinatorial realisation of a relation between wavefunctions of the universe with different propagating states via a very simple differential operator \cite{Benincasa:2019vqr}. However, this relation can be generalised to enlarge the type of propagating states in the wavefunction, but involves a more complicated differential operator \cite{Benincasa:2019vqr}. Classifying which covariant restrictions can be interpreted as derivative operators and which derivative operators have a geometrical-combinatorial picture in terms of polytopes is then crucial to have a geometrical-combinatorial picture for more general wavefunctions.

As mentioned, we explored a very little corner of the relation between positive geometries and differential forms with non-logarithmic singularities. From a physics perspective, non-logarithmic singularities naturally appear in the context of less supersymmetric Yang-Mills theories \cite{Benincasa:2015zna, Benincasa:2016awv} and gravity \cite{Herrmann:2016qea, Heslop:2016plj} for which in both cases a Grassmannian picture is present but it is neither fully understood nor it has been characterised in terms of any positive geometry. A generalisation of our ideas to the Grassmannian has the potential to fill this gap. 
In this direction, as projective polytopes can be obtained as the image of simplices under a map induced by a fixed matrix, see Section \ref{subsec:jkex}, one can consider the image of the positive part of Grassmannians (or more in general, of their cells, or of partial flags etc.) under similar maps. In this sense, \emph{amplituhedra} \cite{Arkani-Hamed:2013jha}, \emph{Grassmann Polytopes} \cite{Lam:grasspol}, \emph{Momentum Amplituhedra} \cite{Damgaard:2019ztj, Lukowski:2020dpn}, etc. are all natural (but highly non-trivial, and non-linear) generalisations of projective polytopes.
The next direction would then be to generalise our framework for these type of geometries as well. In particular, it would be interesting to explore the constructions of parent and child positive geometries, the corresponding covariant restrictions on hypersurfaces, the geometric-combinatorial description of the resulting poles structure and their Laurent leading coefficients.

In summary, the need to tame non-logarithmic singularities comes not only from the mathematical quest of providing a natural generalisation of the framework of positive geometries, but it also stems on the evidence that non-logarithmic singularities enters any attempt to geometrise physical observables in full generality.

\section*{Acknowledgements}

P.B. would also like to thank the developers of SageMath \cite{sagemath}, Maxima \cite{maxima}, and Tikz \cite{tantau:2013a}. Some of the polytope analysis has been performed with the aid of {\tt polymake} \cite{polymake:2000} and {\tt TOPCOM} \cite{Rambau:TOPCOM_ICMS_2002}. P.B. is supported in part by a grant from the Villum Fonden, an ERC-StG grant (N. 757978) and the Danish National Research Foundation (DNRF91). M.P. would like to acknowledge the support of the ERC grant number 724638. M.P would like to thank the Niels Bohr Institute where this work initiated for its hospitality. M.P. would also like to thank F. Mohammadi and L. Monin for stimulating discussions. 

\bibliographystyle{nb}

\bibliography{nonlogsing}

\end{document}